\documentclass[reqno,11pt]{amsart}
\usepackage{graphicx}
\usepackage{tikz-cd}
\usepackage{slashed}
\usepackage{amscd}
\usepackage{tensor}
\usepackage{amssymb}
\usepackage{esint}
\usepackage{enumitem}
\usepackage{enumerate}
\usepackage{accents}
\usepackage{mathtools}
\usepackage{pst-grad} 
\usepackage{pst-plot} 
\usepackage[mathscr]{eucal}
\usepackage{stackengine,graphicx,amssymb}
\stackMath

\numberwithin{equation}{section}
\allowdisplaybreaks[1]

\title[H\"older Continuity]{H\"older Continuity of The Integrated Causal Lagrangian in Minkowski Space}
\author[M.\ Oppio]{Marco Oppio \\ \\ September 2021}
\address{Fakult\"at f\"ur Mathematik \\ Universit\"at Regensburg \\ D-93040 Regensburg \\ Germany}
\email{marco.oppio.r@gmail.com}

\newtheorem{Def}{Definition}[section]
\newtheorem{Thm}[Def]{Theorem}

\newtheorem{Prp}[Def]{Proposition}
\newtheorem{Lemma}[Def]{Lemma}
\newtheorem{Remark}[Def]{Remark}

\newcommand{\beq}{\begin{equation}}
\newcommand{\eeq}{\end{equation}}
\newcommand{\Proof}{\begin{proof}}
	\newcommand{\QED}{\end{proof} \noindent}

\newcommand{\la}{\langle}
\newcommand{\ra}{\rangle}

\newcommand{\Sl}{\mbox{$\prec \!\!$ \nolinebreak}}
\newcommand{\Sr}{\mbox{\nolinebreak $\succ$}}

\newcommand{\C}{\mathbb{C}}
\newcommand{\R}{\mathbb{R}}

\newcommand{\N}{\mathbb{N}}

\renewcommand{\L}{{\mathcal{L}}}

\newcommand{\reg}{{\,\mathrm{reg}}}

\DeclareMathOperator{\supp}{supp}
\renewcommand{\H}{\mathscr{H}}

\newcommand{\Lin}{\text{\rm{L}}}
\newcommand{\F}{{\mathscr{F}}}

\newcommand{\K}{{\mathscr{K}}}

\DeclareMathOperator{\im}{Im}

\newcommand{\scrM}{\mycal M}

\newcommand{\x}{{\textit{\myfont x}}}
\newcommand{\y}{{\textit{\myfont y}}}
\newcommand{\z}{{\textit{\myfont z}}}
\newcommand{\bitem}{\begin{itemize}[leftmargin=2.5em]}
	\newcommand{\eitem}{\end{itemize}}

\newcommand*{\myfont}{\fontfamily{ppl}\selectfont}		
\DeclareFontFamily{OT1}{rsfso}{}
\DeclareFontShape{OT1}{rsfso}{m}{n}{ <-7> rsfso5 <7-10> rsfso7 <10-> rsfso10}{}
\DeclareMathAlphabet{\mycal}{OT1}{rsfso}{m}{n}

\def\scF{\mathscr{F}}

\def\scL{L}

\def\scH{\mathscr{H}}

\def\gR{\mathfrak{R}}
\def\bI{\mathbb{I}}

\def\gR{\mathfrak{R}}

\newcommand{\V}[1]{{\bf{#1}}}

\begin{document}

\maketitle

\begin{abstract}
It is proven that the kernel of the fermionic projector of regularized Dirac sea vacua in Minkowski Space is $L^4$-integrable. The proof is carried out in the specific setting of a continuous exponentially-decaying cutoff in momentum space. As a direct consequence, the corresponding causal Lagrangian is shown to be $L^1$-integrable. Some topological features of the integrated causal Lagrangian are analyzed. In particular, local H\"older-like estimates are proved for continuous regular variations of spacetime, of which a few examples are discussed. Particular emphasis is placed on first-order perturbations of Dirac sea vacua induced by external electromagnetic fields.
\end{abstract}

\tableofcontents

\section{Introduction}

In the theory of causal fermion systems spacetime and all the objects therein are described by a Borel measure on a distinguished family of Hilbert-space operators of rank at most $2n$, where $n$ is the characteristic \textit{spin dimension}. An important role is played by those operators which have maximal rank, the so-called \textit{regular points}.  In a recent work \cite{banach}, extending methods and results from \cite{gauge} to the infinite-dimensional setting, the authors were able to demonstrate that such a distinguished family has the structure of an infinite-dimensional Banach manifold. Some fundamental properties, such as the chain rule, are then  analyzed in the case of H\"older-continuous functions. In particular, the examples of the causal Lagrangian and its integral over spacetime are discussed and proven to  satisfy local H\"older-type inequalities.  

The main objective of this paper is to study local H\"older-continuity of the integrated Lagrangian  in the concrete example of Dirac sea vacua in Minkowski Space. With this in mind, the preliminary and uppermost task is to prove that the causal Lagrangian is in fact integrable on Minkowski Space: such a proof boils down to proving the integrability of the fourth power of the corresponding kernel of the fermionic projector. This is the content of  Propositions \ref{integrabilityP} and \ref{propinteglagrangian} and it is the main conclusion of Sections \ref{sectionfermionicprojex} and \ref{sectionlagrangian}, where also explicit realizations  in terms of Bessel functions are carried out. The proofs of the aforementioned propositions are quite laborious and for this reason postponed to Appendix \ref{appendixproofL4}. 
In Section \ref{sectionCFS} the basic preliminaries on causal fermion systems are provided, placing particular  emphasis on their primary topological features, such as the continuity of the eigenvalues and of the Lagrangian and the existence of continuous families of pseudo-orthonormal bases of the spin spaces. To this aim, some general results on the dependence of the eigenvalues on the corresponding operators are recalled in Appendix \ref{appendixconteigen}. In Section \ref{sectionholder}, we then provide and further elaborate on the necessary preliminaries on local H\"older continuity from \cite{banach}.  In particular, we introduce the notion of an \textit{admissible point} and prove local Lipschitz continuity on regular points of the generalized inverse function $\mathrm{g}$, which appears in the fundamental condition \eqref{boundintP} of Theorem \ref{teoremaholder}.
Such notions are concretely implemented later on in Section \ref{sectionholdervariation}, for the causal fermion systems describing Dirac sea vacua in Minkowski space (whose construction and some useful properties are compactly recalled  in Section \ref{sectionDE}). More precisely, focussing on perturbations of spacetime which are realized as  regular variations of the local correlation operators, we are able to prove  H\"older-type estimates of the integrated Lagrangian which are quantitative, in the sense that they are expressed in terms of the $L^1$ norm of the wave functions forming the spin spaces. As a conclusion of Section \ref{sectionholdervariation}, we provide and discuss a few examples, showing how the aforementioned variations can be concretely realized in practice. In this respect, more emphasis is placed on the general method of  \textit{varying the regularization operators} (see Example 4), which is later analyzed to first order in the concrete example of variations of Minkowski Dirac sea vacua induced by electromagnetic fields (Example 5).
In Appendix \ref{appendixbessel} some necessary properties of the Bessel functions of the second kind are recalled. Finally, for the sake of readibility,   several proofs of the paper are postponed to Appendix \ref{appendixproofs}.

\section{Basics on Causal Fermion Systems}\label{sectionCFS}

\subsection{The General Setting }\label{sectiongeneralsetting}
We start with a brief summary of the basic mathematical objects in the theory of causal fermion systems. We only recall those structures which will be needed in this paper
(for a more complete account see~\cite[Section~1.1]{cfs}). 

\begin{Def}\label{defCFS} Given a separable complex Hilbert space~$(\mathscr{H}, \la .|. \ra)$, we let~$\mathscr{F} \subset \Lin(\H)$ be the set of all self-adjoint operators of finite rank which (counting multiplicities) have at most $n$ positive and $n$ negative eigenvalues.
\end{Def} 
Note that the set~$\F$ is not a linear space, because the sum of two operators in~$\F$ will in general have rank larger than four. In fact, $\F$ has the structure of a {\em{closed double cone}}, meaning
that~$\F$ is closed in the $\sup$-norm topology and that for every~$A \in \F$, the ray~$\R A$ is also
contained in~$\F$ (see \cite[Theorem 4.2]{oppio}). 

Next, we let~$\rho$ be a Borel measure on~$\F$, where by a Borel measure we always mean a measure on the Borel
algebra on~$\F$ (with respect to the $\sup$-norm topology).

\begin{Def}\label{definitioncfs}
	The triple~$(\mathscr{H}, \F, \varrho)$ is referred to as a {\bf causal fermion system}. 
\end{Def}

The parameter $n$ is called the \textbf{spin dimension}. A causal fermion system describes spacetime together with all structures and objects therein.
{\bf Spacetime}, denoted by~$M$, is defined as the support of~$\rho$,
\[ M := \text{supp}\, \rho \subset \F \:. \]
Equipped with the $\sup$-norm topology, $M$ is a topological space.
The fact that spacetime points are operators gives rise to additional structures.
For every~$\x \in \F$ we define
\begin{equation*}
\text{{\bf Spin space} at $\x$:}\quad S_\x:= x(\H)
\end{equation*}
It is a subspace of~$\H$ of dimension at most $2n$. It is endowed with the {\bf{spin scalar product}}  (at $\x$),
\beq \label{ssp}
\Sl \cdot | \cdot \Sr_\x := -\la \,\cdot \,|\, \x \,\cdot\, \ra \::\: S_\x \times S_\x \rightarrow \C \:,
\eeq
which is an indefinite inner product of signature~$(p,q)$ with~$p,q \leq n$.

Each spin space can be decomposed into the orthogonal direct sum  (with respect to the Hilbert scalar product) of the positive and negative spectral subspaces of $\x$ (as an operator on $\scH$). More precisely, 
\begin{equation}\label{decompsign}
S_\x=S_\x^-\oplus S_\x^+\quad\mbox{with}\quad S_\x^+:=\x_+(\H)\ \text{and}\  S_\x^-=\x_-(\H),
\end{equation}
where $\x_\pm$ are the positive and negative components of the operator $\x$, i.e.
\begin{equation}\label{posneg}
\x=\x_++\x_-,\quad \x_+:=\frac{\x+|\x|}{2}\quad\mbox{and}\quad \x_-:=\frac{\x-|\x|}{2}.
\end{equation} 
We denote their dimensions by
\begin{equation}\label{npm}
	n_\pm(\x):=\dim S_\x^\pm\le n.
\end{equation} 
By construction, on the subspaces $S_\x^+,S_x^-$ the operator $\x$ is positive and negative defined, respectively. More precisely, 
$$
\x|_{S_\x^+}=|\x||_{S_\x^+}\quad\mbox{and}\quad \x|_{S_\x^-}=-|\x||_{S_\x^-}.
$$
As a consequence, $S_\x^+,S_x^-$ define negative and positive definite subspaces of \eqref{ssp}, respectively. Moreover, they are orthogonal to each other also with respect to the spin scalar product. 
In other words, the subspaces $S_\x^\pm$ define a canonical \textit{fundamental decomposition} of the indefinite inner space $(S_\x,\Sl\,\cdot\,,\,\cdot\,\Sr_\x)$.

\subsection{Continuity of Eigenvalues and the Lagrangian}\label{sectioncontinuitylag}
Many fundamental structures in the theory of causal fermion systems, such as the Lagrangian or the corresponding causal action, are defined explicitly in terms of the eigenvalues of (products of) operators in $\F$. In order to study the continuity of such functions, it is then primary to elaborate on the topological interplay between operators in $\F$ and their eigenvalues. In this section we review some of the most elementary results in this regard.

Let us first study the individual operators $\x\in\F$, we will then discuss products of the form $\x\y$. Referring to Appendix \ref{appendixconteigen}, we collect the corresponding  positive and negative eigenvalues (repeated according to their multiplicities) into sequences
$$
\{\lambda^\pm_k(\x)\}_{k\in\N}\subset\R_\pm,
$$
 where we adopted the convention that $\lambda^\pm_k(\x)=0$ for every $k>n_\pm(\x)$ and where the eigenvalues are enumerated by non-increasing absolute value. 
 
 Since every spin space is at most $2n$-dimensional, it is convenient to consider only the first $2n$ elements of the sequences above and reorder them as follows:
\begin{equation*}\label{eigenx}
\lambda_{k}^\x:=\begin{cases}
\lambda_k^-(\x) & \mbox{ if $1\le k\le n$}\\[0.3em]
\lambda_{2n-k+1}^+({\x}) & \mbox{ if $n< k\le 2n$}
\end{cases}.
\end{equation*}
With this definition, the $\lambda_k^\x$ cover all of the spectrum of $\x$ (possibly up to the eigenvalue zero in the case $\dim S_\x=2n$), with its elements repeated according to their multiplicities. Applying Proposition \ref{propcontiself}, we then have the following result.
\begin{Prp}\label{teoremacont}
	Let $\x,\y\in\F$. Then, for every $k\in \{1,\dots,2n\}$,
	\begin{equation*}
	|\lambda_k^\x-\lambda_k^\y|\le \|\x-\y\|
	\end{equation*}
\end{Prp}
This result shows that, as long as the right ordering is chosen, the eigenvalues can be arranged to be \textit{Lipschitz continuous} on $\F$. It is important to stress that it is in general not possible to carry out a similar arrangement for the eigenvectors, so that they depend continuously on the operators, as a famous example by Rellich shows (see Example 5.3 in \cite[Section II.5]{Kato}). Nevertheless, continuous families of pseudo-orthonormal bases can always constructed in neighborhoods of regular points (see Section \ref{subsectionsign}).

As a next step, we study products of the form $\x\y$. Note that such operators are in general not self-adjoint, nor normal, for the two factors may not commute. 

The product operator has $n({\x\y})\le 2n$ non-zero eigenvalues. Let us introduce
\begin{equation}\label{eigenvaluesxy}
\lambda_1^{\x\y},\dots,\lambda_{2n}^{\x\y}\in\C\ \mbox{ with }\ |\lambda_{1}^{\x\y}|\ge \dots\ge |\lambda_{2n}^{\x\y}|,
\end{equation}
where, repeating according to the algebraic multiplicity,
\vspace{0.2cm}
\begin{itemize}[leftmargin=2em]
	\item[\rm{(1)}] $\lambda_k^{\x\y}\ $ for  $ 1\le k\le n(\x\y)$ are the non-zero eigenvalues of $\x\y$\\[-0.5em]
	\item[\rm{(2)}] $\lambda^{\x\y}_k=0\ $ for all  $\ n(\x\y)<k\le 2n\ $ (whenever $n(\x\y)< 2n$).
\end{itemize}
\vspace{0.2cm}
If the operator $\x\y$ is normal, then (see Appendix \ref{appendixconteigen} for the notion of a \textit{singular value})
$$
|\lambda_k^{\x\y}|=s_k(\x\y)\quad\mbox{for every }k=1,\dots,2n.
$$
However, as already mentioned above, the normality of the product $xy$ is in general not satisfied for arbitrary $\x,\y\in\F$. 

At this point, using Theorem \ref{continuityenum} and the fact that the function \eqref{Lagrangian} below is independent of the enumeration of the eigenvalues, we have the following result.

\begin{Prp}\label{proplagrangian}
	The \textbf{Lagrangian} defined by
\begin{equation}\label{Lagrangian}
\L:\F\times\F\ni (\x,\y)\mapsto \frac{1}{4n}\sum_{i,j=1}^{2n}\big(|\lambda_i^{\x\y}|-|\lambda_j^{\x\y}|\big)^2\in \R_+,
\end{equation}
is symmetric, non-negative and continuous.
\end{Prp}

The mutual relations between the eigenvalues define a \textbf{causal structure} in spacetime. This notion will be briefly addressed  in Section \ref{sectionlagrangian} in the example of  Dirac sea vacua in Minkowski space.
Two spacetime points $\x,\y\in M$ are said to be \\[-0.8em]
	\begin{itemize}[leftmargin=2em]
		\item[{\rm (1)}] \textbf{Spacelike separated} \textit{if all the $\lambda^{\x\y}_i$ have the same absolute value}\\[-0.3em]
		\item[{\rm (2)}] \textbf{Timelike separated} \textit{if the $\lambda_i^{\x\y}$ are  real and do not have the same absolute value}\\[-0.3em]
		\item[{\rm (3)}] \textbf{Lightlike separated}  \textit{otherwise}
	\end{itemize}
\vspace{0.15cm}

Given this definition, we see that \eqref{Lagrangian} vanishes for spacelike separated points. In this sense, the Lagrangian is said to be \textbf{causal}.

To conclude this section, we recall other two fundamental objects of the theory. For any~$\x,\y \in M$ we define the
{\bf{kernel of the fermionic projector}} by
\beq \label{Pxydef}
\mathrm{P}(\x,\y) = \pi_\x \,\y|_{S_\y} \::\: S_\y \rightarrow S_\x \:.
\eeq
This is a mapping from one spin space to another, thereby
inducing relations between different spacetime points. A connected concept is the \textbf{closed chain}, defined as
\begin{equation}\label{closedchain}
\mathrm{A}_{\x\y}:= \mathrm{P}(\x,\y)\mathrm{P}(\y,\x): S_\x\rightarrow S_\x.
\end{equation} 
Note that both mappings can be understood as finite-rank operators on $\H$: one simply needs to extend them by zero on the orthogonal of the spin spaces. The spectrum of the closed chain coincides with the non-zero spectrum of $\x\y$ (see \cite[Section 1.1.3]{cfs}). Therefore, \textit{the kernel of the fermionic projector encodes the causal structure of spacetime.}

It follows directly from the definition that, for any fixex $\x\in\F$, the mappings 
$$
\mathrm{A}_{\x\,\cdot}\quad\mbox{and}\quad \mathrm{P}(\x,\,\cdot\,)
$$
are continuous as functions from $\F$ to $\mathfrak{B}(\H)$, with respect to both the operator and the trace norms (the latter case can be proved similarly as in Proposition \ref{proplagrangian}). Although this is all we need in this paper,  it should be stressed that continuity in the first variable may in general fail, as the next simple example shows:
$$
\mbox{Let }e\in\mathbb{S}_\H.\ \mbox{ Then}\quad \x(t)=t \langle e,\,\cdot\,\rangle e \rightarrow 0\quad\mbox{but}\quad \pi_{\x(t)}=\langle e,\,\cdot\,\rangle e\not\to 0=\pi_0\quad\mbox{as}\quad t\searrow 0.
$$
As we will see in the next section, such counterexamples are ruled out once we focus our attention on  maximal rank operators.

\subsection{Local Signature and Regular Systems }\label{subsectionsign}
The situation in which the spin spaces have maximal rank $2n$ is of great importance: this turns out to be the case for Dirac sea vacua in Minkowski space and also in presence of particles and anti-particles (see \cite[Section 5]{oppio}).

Referring to \eqref{npm}, we define the \textbf{local signature} at a point  $\x\in\F$ as
$$
\mathrm{sign}(\x):=(n_-(\x),n_+(\x))
$$
In particular, $\mathrm{dim}\, S_\x=n_-(\x)+n_+(\x)$.
\begin{Def}\label{Defregular}
	A point $\x\in \F$ is said to be \textbf{regular} if $\mathrm{sign}(\x)=(n,n)$. A Borel measure on $\F$ is said to be regular if every point of its support is regular.
\end{Def}
This condition is clearly equivalent to the requirement $\dim S_\x=2n$. It is convenient to give the set of these operators its own symbol:
$$
\F^\reg:=\{x\in\F\:|\:\mathrm{sign}(\x)=(n,n)\}.
$$
Note that this set fulfills $\R_\pm\F^\reg\subset \F^\reg$, but it lacks closedness.

\begin{Prp}\label{openreg}
	The following statements hold true.\\[-1em]
	\begin{itemize}[leftmargin = 2em]
		\item[{\rm (i)}] For all $\x\in\F$ there is an $r>0$ such that
		$$
		n_\pm(\x)\le n_\pm(y)\quad\mbox{for all }\y\in B_r(\x).
		$$
		\item[{\rm (ii)}] The set $\F^{\mathrm{reg}}$ is an open and dense subset of $\F$.
	\end{itemize}
\end{Prp}
The set $\F^\reg$ of regular points owns much richer structures than what we are going to need and discuss in this paper. Of all the features, it is absolutely worth mentioning that $\F^\reg$ can be equipped with an \textit{infinite-dimensional smooth Banach manifold structure} (see \cite[Section 3.2]{banach}).

For regular systems, it is convenient to consider families of pseudo-orthonormal bases which also respect the signature. Let us introduce the \textit{signature vector}: 
\begin{equation}\label{signaturematrix}
s:=(\underbrace{1,\dots,1}_{n\mathrm{-times}},\underbrace{-1,\dots,-1}_{n\mathrm{-times}}).
\end{equation}
For any regular point $\x\in\F^\reg$, we say that a  linear basis $\{e_j\}_j$ of $S_\x$ is  \textbf{faithful} if it is \textit{pseudo-orthonormal} with respect to the spin scalar product and if it respects the decomposition into positive and negative spaces of $\x$, i.e. if it satisfies:
	\begin{equation*}\label{GS}
	\mbox{For all $i,j=1,\dots,2n$}\quad \Sl e_i|e_j\Sr_{\x}=s_{i}\delta_{ij}\quad\mbox{and}\quad e_j\in\begin{cases}
	S_{\x}^- &\mbox{if }j\le n\\[0.2em]
	S_{\x}^+ &\mbox{if }j>n
	\end{cases}.
	\end{equation*}

Any faithful linear basis of $S_\x$ can be turned into an orthonormal basis with respect to the Hilbert scalar product.
Taking into account the definition of spin scalar product \eqref{ssp} and the fact that the definite spin spaces $S_\x^\pm$ are invariant under the action of any bounded measurable function of $|\x|$, it is not difficult to see that
\begin{equation}\label{defbasis}
\hat{e}_j:=\sqrt{|\x|}\,e_j,\quad j\in\{1,\dots,2n\}
\end{equation}
is a Hilbert basis of $S_\x$ with respect to the Hilbert scalar product, which also preserves the sign decomposition \eqref{decompsign}.
Such bases can be chosen to continuously depend on the spacetime points, at least locally. More generally, we have the following result, which can be proved by combining the proof of Proposition \ref{openreg}-(i) with a Gram-Schmidt orthogonalization argument. Also, bear in mind that the function $$\F\ni \x\mapsto \sqrt{|\x|}\in\mathfrak{B}(\H)$$ is continuous in the $\sup$-norm topology (see for example (1.5.16) in \cite{cfs})
\begin{Prp}\label{existenceframe}
	Let $\Omega$ be a topological space and $\varphi:\Omega\rightarrow\F^{\reg}$ be continuous. Then, for every $\omega_0\in\Omega$ there exist $2n$-continuous functions
	$$
	e_j:\Omega_{\omega_0}\rightarrow \scH\quad\mbox{on an open neighborhood $\Omega_{\omega_0}$ of $\omega_0$}
	$$
	such that, for every $\omega\in \Omega_{\omega_0}$ 
	the set
	\begin{equation*}
	\{e_j(\omega)\:|\: j=1,\dots,2n\}\quad\mbox{is a a faithful basis of $S_{\varphi(\omega)}$}.
	\end{equation*}
	This is called a \textbf{local spin frame} of $\varphi$. 
The corresponding family of Hilbert bases
\begin{equation*}
\hat{e}_j:=\sqrt{|\varphi(\,\cdot\,)|}\,e_j:\Omega_{\omega_0}\rightarrow \scH
\end{equation*}
is referred to as a \textbf{local Hilbert frame} of $\varphi$.
\end{Prp}
As a special example, one considers $\Omega\subset\F^\reg$ and as $\varphi$ the identity map. In this case, one simply talks of \textit{local spin frame} and \textit{local Hilbert frame}.
Note that, by construction, any local Hilbert frame respects the decomposition \eqref{decompsign} into positive and negative subspaces of the spin spaces.
Whether such frames can be chosen globally on $\Omega$ depends of course on the topological properties of the mapping $\varphi$. This is in fact true in the example of Minkowski vacuum, as we will see later  in Proposition \ref{propbasismink}. 

To conclude this section, as already mentioned at the end of Section \ref{sectioncontinuitylag}, we stress that the kernel of the fermionic projector and the closed chain are in fact \textit{continuous in both variables}, when evaluated on $\F^\reg$. This can be seen, for example, exploiting Proposition \ref{existenceframe}, which yields, with obvious notation,
\begin{equation*}\label{orthogprojecont}
\pi_{\varphi(\omega)}^-=\sum_{j=1}^n \langle\hat{e}_j(\omega)|\,\cdot\,\rangle\, \hat{e}_j(\omega)\quad\mbox{and}\quad\pi_{\varphi(\omega)}^+=\sum_{j=n+1}^{2n} \langle\hat{e}_j(\omega)|\,\cdot\,\rangle\, \hat{e}_j(\omega).
\end{equation*}
The continuity of the local frames imply the continuity of the corresponding projectors.

\section{H\"older Continuity of the Integrated Lagrangian}\label{sectionholder}

In \cite[Section 5]{banach} the continuity properties of the causal Lagrangian \eqref{Lagrangian} are  analyzed in detail. In the present work, we are mostly interested in Remark 5.4-(2) therein, whose statement is summarized in the next theorem for convenience. By 
$
\pi_U
$
we denote the orthogonal projector on a closed subspace $U$ of $\H$. 
\begin{Thm}\label{holderlagrangian}
	Each $\x\in\F\setminus\{0\}$ has a neighborhood $U\subset\F$  such that the inequality
	\begin{equation}\label{estimatefirst}
	|\L(\x,\y)-\L(\z,\y)|\le c\,\|\x\|^{2-\alpha}\,\|\pi_J\,\y\,\pi_J\|^2\,\|\x-\z\|^{\alpha}
	\end{equation}
	holds for all $\z\in U$ and for all $\y\in\F$, where $J:=\mathrm{span}\{S_\x,S_\z\}$ and the constants $c,\alpha$ depend only on the spin dimension.
\end{Thm}
This estimate provides a \textit{H\"older-like inequality} for the Lagrangian. The explicit form of the parameters $c,\alpha$ can be found in the aforementioned reference and it is not relevant for the purposes of this paper. We want to exploit the inequality \eqref{estimatefirst} to analyze the H\"older continuity of another function, namely the \textbf{integrated Lagrangian}.  First, we need to distinguish the spacetime points for which such a function is well-defined.
\begin{Def}
	Given a Borel measure $\rho$ on $\F$, a point $\x\in\F$ is said to be \textbf{admissibile} for  $\rho$ if $\L(\x,\,\cdot\,)\in L^1(M,d\rho)$. We denote the set of admissible points by $\mbox{\rm Adm}(\rho)$. For such points, the following function is well-defined:
	\begin{equation}\label{ell}
	\ell :\mbox{\rm Adm}(\rho)\ni \x\mapsto \int_M \L(\x,\y)\,d\rho(\y)\in\R_+\,.
	\end{equation}
\end{Def}
The function $\ell$ plays a crucial role in the analysis of the causal action principle, for the corresponding \textit{Euler-Lagrange equations} demand that $\ell$ is constant on the support of minimizers (see for instance \cite[Section 2.2]{hamilt}). In particular, for such minimizing measures, one has that $M\subset \mathrm{Adm}(\rho)$. 

We now give a  sufficient condition for a spacetime point $\x\in M$ to be admissible,  which involves the integrability of the kernel of the fermionic projector. To this aim, we first introduce the \textbf{generalized inverse} operators \footnote{In \cite{banach} the generalized inverse $\mathrm{g}(\y)$ of $\y\in\F$ is denoted by $Y^{-1}$}:
\begin{equation*}\label{geninverse}
	\mbox{For all $\x\in\F\ $ let }\  \mathrm{g}(\x):=
	\begin{cases}
		0 & \mbox{if $\x=0$} \\
		(\x|_{S_\x})^{-1}\oplus 0 & \mbox{if $\x\neq 0$}  
	\end{cases}
\end{equation*}
It readily follows that 
$$
\mathrm{g}:\F\rightarrow\F,\quad \mathrm{g}(\F^\reg)\subset\F^\reg\quad \mbox{and}\quad  \mathrm{g}(\x)\,\x=\x \,\mathrm{g}(\x)=\pi_\x.
$$ Moreover, $\mathrm{g}$ is continuous on $\F^\reg$, as made clear by the following statement, whose proof can be found in Appendix \ref{appendixproofs}.
\begin{Lemma}\label{lemmaregcont}
	For every $\x\in \F^\reg$ there exists $B_r(\x)\subset\F^\reg$ such that
	\begin{equation*}\label{lipischitz}
		\|\mathrm{g}(\y)-\mathrm{g}(\x)\|\le 6\|\mathrm{g}(\x)\|^2\|\y-\x\|\quad\mbox{for all $\y\in B_r(\x)$}.
	\end{equation*}
	Thus, $\mathrm{g}$ is continuous on $\F^\reg$. Moreover, $\mathrm{g}$ is everywhere discontinuous on $\F\setminus\F^\reg$.
\end{Lemma}
We now consider the following two inequalities, where the norms are as always the operator norms. The first estimate follows directly from the general properties of operators and eigenvalues (recall that the spectrum of $\mathrm{A}_{\x\y}$ coincide with the non-zero spectrum of $\x\y$), while the second one follows from the properties of the generalized inverse.
\begin{equation*}
\begin{split}
		\mathrm{i)}&\ |\lambda^{\x\y}_i|\le \|\mathrm{A}_{\x\y}\|\le \|\mathrm{P}(\x,\y)\|\|\mathrm{P}(\y,\x)\|\quad\mbox{for all }i=1,\dots,2n\,,\\[0.2em]
\mathrm{ii)}&\ \|\mathrm{P}(\y,\x)\|=\|\pi_\y\,\x\|=\|\mathrm{g}(\y)\,\y\,\pi_\x\,\x\|\le \|\mathrm{g}(\y)\|\|\x\|\|\pi_\x\,\y\|=\|\x\|\|\mathrm{g}(\y)\|\|\mathrm{P}(\x,\y)\|.
\end{split}
\end{equation*}
\noindent The next result is a direct consequence of these estimates and the definition of the Lagrangian \eqref{Lagrangian}. For technical simplicity, we here make the the additional \textit{assumption   that the measure $\rho$ is regular}, i.e. 
$$
M\subset\F^\reg.
$$
In this case, the integral below in \eqref{intPpower4} is well-defined, for the integrand is continuous for any \textit{fixed} $\x\in \F$ (see the discussion at the end of Section \ref{sectioncontinuitylag} and Lemma \ref{lemmaregcont}). It should be stressed that this condition may in principle be relaxed, because Proposition \ref{prp4admiss} applies in fact to every Borel measure on whose support the generalized inverse is measurable. Note, though, that regularity is indeed realized in the important example of Dirac sea vacua in Minkowski Space. The last statement of Proposition \ref{prp4admiss} follows trivially from $\mathrm{P}(0,\y)=0$.
\begin{Prp}\label{prp4admiss}
	Let $\rho$ be regular. Then, any  $\x\in\F$ that satisfies the following condition is admissible:
	\begin{equation}\label{intPpower4}
		\int_M\|\mathrm{P}(\x,\y)\|^4\,\|\mathrm{g}(\y)\|^2\,d\rho(\y)<\infty.
	\end{equation}
	In particular, the trivial point $\x=0$ is always admissible.
\end{Prp}
Although being well-defined, the function $\ell$ may in general be discontinuous on the set of admissible points and additional assumptions on the measure must be given in order to achieve  regularity of any kind.
Here, we state a strengthened version of condition \eqref{intPpower4} which is sufficient to imply local H\"older continuity.

Consider an admissible point $\x\in\F\setminus\{0\}$ and let $U$ be a neighborhood of $\x$ as in Theorem \ref{holderlagrangian}. Let now 
\begin{equation*}\label{nonemptinessintersection}
 \z\in \mathrm{Adm}(\rho)\cap U\ \  \mbox{($\neq \emptyset$, as it contains $\x$) }.
\end{equation*}
From \eqref{estimatefirst}, one then immediately infers that
	\begin{equation*}
	\begin{split}
	|\ell(\x)-\ell(\z)|&\le \int_M |\L(\x,\y)-\L(\z,\y)|\,d\rho(\y)\le\\ 
	&\le c\,\|\x\|^{2-\alpha}\,\|\x-\z\|^\alpha\,\int_M\|\pi_J\, \y\,\pi_J\|^2\,d\rho(\y),
	\end{split}
	\end{equation*}
Following the argument as in \cite[Theorem 5.9]{banach}, the integral on the right can be estimates in terms of the kernel of the fermionic projector, namely,
%
	\begin{equation*}
	\begin{split}
	 \int_M\|\pi_J\, \y\,\pi_J\|^2\,d\rho(\y)\le 8\int_M\|\mathrm{P}(\x,\y)\|^4\|\mathrm{g}(\y)\|^2\,d\rho(\y)+8\int_M\|\mathrm{P}(\z,\y)\|^4\|\mathrm{g}(\y)\|^2\,d\rho(\y).
	\end{split}
	\end{equation*}
As above, the integrands are continuous and hence the integrals are well-defined.
	Having control on the $L^4$-norm of the kernel of the fermionic projector, the above estimates provide H\"older continuity conditions for the function $\ell$. In particular, the integrated Lagrangian is continuous. 
	More generally, one has the following slight generalization of \cite[Theorem 5.9]{banach}. This result follows directly from the estimates just proven and Proposition \ref{prp4admiss}. 
	\begin{Thm}\label{teoremaholder}
		Let $\rho$ be regular. Let $\Omega$ be a topological space and let $G\in C^0(\Omega,\F)$ satisfy
		\begin{equation}\label{boundintP}
		\sup_{\omega\in \Omega}\: \int_M\|P(G(\omega),\y)\|^4\|\mathrm{g}(\y)\|^2\,d\rho(\y)<\infty.
		\end{equation}\\[-0.2em]
		Then, the following properties hold.\\[-1em]
		\begin{itemize}[leftmargin=2.5em]
		\item[\rm{(i)}] $G(\Omega)\subset \mathrm{Adm}(\rho)$,\\[-0.9em]
		\item[\rm{(ii)}]Let $\omega_0\in\Omega$ satisfy $G(\omega_0)\neq 0$. Then, there is a neighborhood $\Omega_0\subset\Omega$ of $\omega_0$ and  a constant $K>0$ such that, for all  $\omega\in\Omega_0$,
		$$
		|\ell(G(\omega))-\ell(G(\omega_0))|\le K\,\|G(\omega_0)\|^{2-\alpha}\,\|G(\omega)-G(\omega_0)\|^\alpha. 
		$$
		In particular, $\ell\circ G$ is continuous at $\omega_0$.
		\end{itemize}

	\end{Thm}
	In the special case of translation invariant systems, such as Dirac sea vacua in Minkowski space, the spacetime operators are unitary equivalent to each other,  
	$
	\|\mathrm{g}(\y)\|
	$
	is then constant on the support of the measure
	 (see (iii) in Section \ref{subsectionCFSM}) and can therefore be neglected from condition \eqref{boundintP}. 
	 
	In this paper we will see how the strategies explained in this section can be implemented in the specific example of causal fermion systems in Minkowski space. More precisely, we will consider perturbations of regularized Dirac sea vacua, realizing the function $G$ above as a continuous transformation of the local correlation function (see Section \ref{sectionholdervariation}).

\section{The Dirac Equation in Minkowski space}\label{sectionDE}

In the present work we mainly focus on causal fermion systems obtained by regularizing
the vacuum Dirac sea in Minkowski space~$\scrM$ as analyzed in detail in~\cite{cfs}, \cite{neumann} and \cite{oppio}, where we address the interested reader for more details.
For notational simplicity, we work in a fixed reference frame and identify Minkowski
space with~$\R^{1,3}$. This is endowed with the standard Minkowski inner product of signature
convention $(+,-,-,-)$, denoted here by~$u \!\cdot\! v$. 
We denote spacetime indices by~$i,j \in \{0, \ldots, 3\}$ and spatial indices by~$\alpha, \beta \in \{1,2,3\}$.
We use natural units $\hbar = c = 1$.
The Minkowski metric gives rise to a {\em{light cone structure}}:
The sets 
$$
L_0=
\{ \xi \in \scrM \,|\, \xi \cdot \xi =0\},\ \ I_0 = \{ \xi \in \scrM \,|\, \xi \cdot \xi > 0 \},\ \ J_0 = \{ \xi \in \scrM \,|\, \xi \cdot \xi \geq 0 \}
$$ 
are referred to as the
\textit{null cone, interior light cone} and \textit{closed light cone}, respectively.
By translation, we obtain corresponding cones centered at
any point $x\in\R^{1,3}$. They will be denoted by $L_x,I_x$ and $J_x$, respectively. Similarly, one denotes by $J(K)$ the closed light cone generated by a compact set $K$. An index $\wedge$ or $\vee$ will indicate the lower and upper half of the cone, respectively.

\subsection{The Equation and its Solution Space}\label{sectiondiracequation}

The starting point is the free Dirac equation. Let us introduce the differential operator
\begin{equation}\label{Diracequation}
\mathrm{D}:=i\gamma^j\partial_j-m: C^\infty(\R^4,\C^4)\rightarrow C^\infty(\R^4,\C^4).
\end{equation}
From the theory of symmetric hyperbolic systems, it follows that, for every $t\in\R$, the Cauchy problem
\begin{equation}\label{CP}
	\left\{
\begin{split}
\mathrm{D}f&=0\\
f(t,\cdot)&=\varphi\in C^\infty(\R^3,\C^4)
\end{split}
\right.
\end{equation}
 admits a unique smooth solution. Moreover, \textit{finite propagation speed} ensures that solutions with compactly supported initial data are \textit{spatially compact}, i.e. they belong to the space
$$
C_{\mathrm{sc}}^\infty(\R^4,\C^4):=\{f\in C^\infty(\R^4,\C^4)\:|\: f(t,\cdot)\in C_0^\infty(\R^3,\C^4)\ \mbox{for all }t\in\R \}.
$$
This  determines the following class of linear isomorphisms: For $t\in\R$,
\begin{equation}\label{mappingE}
\begin{split}
\mathrm{E}_t:C_0^\infty(\R^3,\C^4)&\rightarrow\scH_m^{\rm{sc}}:=\ker\mathrm{D}\cap C^\infty_{\mathrm{sc}}(\R^{4},\C^4),\\
&\mathrm{E}_t(\varphi)(t,\cdot)=\varphi,
\end{split}
\end{equation}
which propagates the compactly supported initial data at time $t$ to all of spacetime.

The linear space $\scH_m^{\rm{sc}}$ is independent of the chosen time $t$ in \eqref{mappingE} and can be given a pre-Hilbert space structure by equipping it with the $L^2$-scalar product on the initial data,
\begin{equation}\label{innerproduct}
(  f,g ) := \int_{\R^3} f(0, \V{x})^\dagger g(0, \V{x})\: d^3\V{x}\qquad \mbox{for all }f,g\in\scH_m^{\rm{sc}} \:,
\end{equation}
where the dagger denotes complex conjugation and transposition. The integration over any other Cauchy surface $\{ t=\mbox{const}\}$ would give the same result, due to \textit{current conservation}.

The inner product \eqref{innerproduct} makes the mappings~$\mathrm{E}_t$ linear isometries once we endow $C_0^\infty(\R^3,\C^4)$ with the standard $L^2$-product.
The \textit{one-particle Hilbert space}~
is defined
as the Hilbert space completion of $\scH_m^{\rm{sc}}$. It coincides with the topological completion of $\scH_m^{\rm{sc}}$ within the space of locally square integrable functions,
$$
\scH_{m}:=\overline{\H_m^{\rm{sc}}}\subset L^2_{\rm{loc}}(\R^{4},\C^4).
$$ 
The corresponding scalar product will be denoted by $\langle\,\cdot\,|\,\cdot\,\rangle_m$. As a consequence, the isomorphisms  $\mathrm{E}_t$ extend continuously to a unitary operators on the space $\scL^2(\R^3,\C^4)$, which will be again denoted by $\mathrm{E}_t$. It is important to stress these operators gives back the unique smooth solution of \eqref{CP} also when they acts on smooth functions of $L^2(\R^3,\C^4)$ \textit{without} compact support.
This is a direct consequence of the energy inequalities and uniqueness of weak solutions for symmetric hyperbolic systems.

The operators $\mathrm{E}_t$ induce a \textit{time evolution} on the initial-data space $L^2(\R^3,\C^4)$,
\begin{equation*}\label{unitev}
U_{t,t_0}:\R\ni t\mapsto \mathrm{E}_t^{-1}\,\mathrm{E}_{t_0}\in \mathfrak{B}(L^2(\R^3,\C^4)).
\end{equation*}
This mapping turns out to be a strongly-continuous one-parameter group of unitary operators. Its self-adjoint generator is the \textit{Dirac Hamiltonian}:
	$$
	H:=-i\gamma^0\gamma^\alpha\partial_\alpha+m\gamma_0,\quad \mathfrak{D}(H):=W^{1,2}(\R^3,\C^4)\quad \mbox{\big(with  $U_{t,t_0}=e^{-i(t-t_0)H}$\big)}.
	$$
The spectrum of $H$ is purely continuous and  given by $\{\omega\:|\: |\omega|\ge m\}$. This form of the spectrum corresponds to an orthogonal splitting of $L^2(\R^3,\C^4)$ into a positive and a negative spectral subspace of $H$. Through the isometries $\mathrm{E}_t$ this orthogonal decomposition can be lifted to a corresponding splitting  (which is independent of $t$)
\begin{equation*}
	\H_m=\H_m^+\oplus \H_m^-.
\end{equation*}

To conclude, we note that the action of the operators $\mathrm{E}_t$ can be merged into the following function
\begin{equation}\label{functionE}
	\begin{split}
\mathrm{E}:&\  C^\infty(\R^3,\C^4)\cap L^2(\R^3,\C^4)\rightarrow C^0(\R^4\times \R,\C^4),\\ 
&\qquad\quad\mathrm{E}[\varphi](t,\V{x};s):=\mathrm{E}_s(\varphi)(t,\V{x}).
\end{split}
\end{equation}
In the next section we will see how these mappings can in fact be represented as distributions in spacetime.

\subsection{The Causal Propagator and Frequency Splitting}\label{sectioncausalprop}
	
In this section we recall the construction of the causal propagator of the Dirac equation and its corresponding frequency (or energy) splitting.	 It should be stressed that this is a general method, which applies also in presence of \textit{static} electromagnetic fields.

In the remainder of the paper we will use the following notation\footnote{Similar conventions apply to the spaces of compactly supported and tempered distributions.}:
$$
\mathcal{D}(\R^n,\C^m):=C_0^\infty(\R^n,\C^m),\quad \mathcal{D}'(\R^n,\C^m):=\mathfrak{B}(\mathcal{D}(\R^n,\C),\C^m),
$$ 
where by $\mathfrak{B}(X,Y)$ we denote the space of linear and continuous function between the topological vector spaces $X,Y$. All spaces above are equipped with the standard topologies of distribution theory. Notice that every $T\in \mathcal{D}'(\R^n,\C^m)$ is equivalent to the assignment of a family $(T_\sigma)_{\sigma=1,\dots,m}\subset\mathcal{D}'(\R^n,\C)$ determined by 
$
T_\sigma(\varphi)=T(\varphi)_\sigma.
$
For simplicity of notation, we omit the symbol of the target space when $m=1$. 
 
\subsubsection{The Causal Fundamental Solution}
As a first step, we rewrite \eqref{functionE} in the following way:
\begin{equation}\label{Edistri}
\mathrm{E}:\mathcal{D}(\R^3,\C^4)\ni \varphi\mapsto \mathrm{E}[\varphi]\in \mathcal{D}'(\R^4\times\R,\C^4 ),
\end{equation}
where the action of $\mathrm{E}[\varphi]_\sigma$ is given  in terms of the standard $L^2$ product. Exploiting current conservation and finite propagation speed, the following result follows from a compontentwise application of the Schwartz Kernel Theorem to \eqref{Edistri}. 

By $s$ we denote again the signature vector of the spin scalar product (see \eqref{signaturematrix}), which in the current case is given by
$$
s_\mu=(\gamma^0)^{\mu\mu}\quad\mu=1,2,3,4.
$$ 
\begin{Thm}\label{theoremcausalpropag}
	There exists a unique $k\in \mathcal{D}'(\R^4\times\R^4,\mathrm{Mat}(4,\C))$ such that, \\[-0.4em]
	$$
	(f,
	\mathrm{E}[h\mathfrak{e}_\nu]_\mu)_{L^2(\R^5)}=s_{\nu} \,k_{\mu\nu}( \overline{f}\otimes h)\quad\mbox{for all $h\in \mathcal{D}(\R^3)$, $f\in \mathcal{D}(\R^5)$},
	$$ \\[-0.6em]
	and all indices $\mu,\nu=1,2,3,4$. 
	This is called the \textbf{causal fundamental solution}. \\[0.1em]
	Formally, for all $\varphi\in \mathcal{D}(\R^3,\C^4)$,
	$$
	\mathrm{E}[\varphi](t,\V{x};s)=\int_{\R^3} k(t,\V{x};s,\V{y})\,\gamma^0\,\varphi(\V{y})\,d^3\V{y}.
	$$
	Moreover, in the sense of distributions,
\begin{equation*}
[i\slashed{\partial}_1-m]k=0.
\end{equation*}
Let $u\in \mathcal{D}(\R^4,\C^4)$. Then, the following function $k(\,\cdot\,,u)$ belongs to $\H_m^{\mathrm{sc}}$:
\begin{equation*}
\begin{split}
 k(x,u):=\int_\R \mathrm{E}[u(s,\cdot)](x,s)\,ds
\end{split}
\end{equation*}
Formally, 
$$
k(x,u)=\int_{\R^4} k(x,y)\,\gamma^0\,u(y)\,d^4y.
$$
Finally, 
$
\supp k(\,\cdot\,,u)\subset J(\supp u).
$
\end{Thm}

We now consider the projection of this distribution onto the positive and negative energy subspaces of the Dirac Hamiltonian. The negative component plays a distinguished role in the theory of causal fermion system.

\subsubsection{Frequency Splitting} Let $\mathbb{I}^\pm\in\mathfrak{B}(L^2(\R^3,\C^4))$ denote  the orthogonal projectors onto the positive and negative spectral subspaces of the Dirac Hamiltonian $H$. It can be shown that these operators maps compactly supported functions into smooth (Schwartz) functions (see \cite[Section 2.2]{oppio}). 
Thus, similarly as in the previous section, we can define
\begin{equation}\label{Edistripm}
\mathrm{E}^\pm: \mathcal{D}(\R^3,\C^4)\ni \varphi\mapsto \mathrm{E}[\bI^\pm\varphi]\in  \mathcal{D}'(\R^4\times\R,\C^4 ) .
\end{equation}
The following result can be proved with a similar argument as in Theorem \ref{theoremcausalpropag}.
\begin{Thm}\label{theoremcausalpropagpm}
	There exists unique $P^\pm\in \mathcal{D}'(\R^4\times\R^4,\mathrm{Mat}(4,\C))$ such that\\[-0.4em] 
	$$
	(f,
	\mathrm{E}^{\pm}[h\mathfrak{e}_\nu]_\mu)_{L^2(\R^5)}=s_\nu\,P_{\mu\nu}^\pm( \overline{f}\otimes \varphi)\quad\mbox{for all }h\in \mathcal{D}(\R^3),\ f\in \mathcal{D}(\R^5),
	$$ \\[-0.6em]
 	and all indices $\mu,\nu=1,2,3,4$. $P^-$ is called the \textbf{kernel of the fermionic projector}.  \\[0.1em]
 	Formally, for all $\varphi\in \mathcal{D}(\R^3,\C^4)$,
	$$
	\mathrm{E}^\pm[\varphi](t,\V{x};s)=\int_{\R^3} P^\pm(t,\V{x};s,\V{y})\,\gamma^0\,\varphi(\V{y})\,d^3\V{y}.
	$$
		Moreover, in the sense of distributions,
	\begin{equation*}\label{diracequationP}
	[i\slashed{\partial}_1-m]P^\pm=0.
	\end{equation*}
	Let $u\in \mathcal{D}(\R^4,\C^4)$. Then, the following function $P^\pm(\,\cdot\,,u)$ belongs to $\H^\pm_m$:
	\begin{equation*}
	\begin{split}
	P^\pm(x,u):=\int_\R \mathrm{E}^\pm[u(s,\cdot)](x,s)\,ds
	\end{split}
	\end{equation*}
 	Formally, one writes
	$$
	P^\pm(x,h)=\int_{\R^4} P^\pm(x,y)\,\gamma^0\,h(y)\,d^4y.
	$$
\end{Thm}
It should be stressed that the construction in Theorem \ref{theoremcausalpropagpm} can be carried out also in presence of a \textit{static electromagnetic field}, although the function \eqref{functionE} may in general be more complicated to handle.
In our specific setting of a free particle, the kernel of the fermionic projector has the following explicit representation
\begin{equation}\label{bidistribution}
\begin{split}
P^\pm(x,y)&:=\int_{\R^4}\frac{d^4 k}{(2\pi)^4} \:(\slashed{k}+m) \:\delta(k^2-m^2)\:\Theta(\pm k_0)\: e^{-i k\cdot (x-y)}.
\end{split}
\end{equation}

It is important to note that the projection onto the positive or negative energy spaces destroys localization. More specifically, solution with fixed sign of the energy \textit{cannot} have spatially compact support. This is the content of \textit{Hegerfeldt's theorem} (see \cite{hegerfeldt1974remark}; for the connections with causal fermion systems in Minkowski space see \cite{neumann}). 
 In other words, the distributions $P^\pm$ are not supported on the light cone. In particular, one sees that the kernel of the fermionic projector is \textit{not causal}, for it does not vanish for spacelike separated points.


\subsection{Regularization by Smooth Cutoff in Momentum Space}\label{sectionreg}
In the context of causal fermion systems, in order to take into account the presence of a minimal length scale, an ultraviolet regularization is introduced. This can be done in several different ways. For technical simplicity, in this work we focus on simple cutoffs in momentum space. The starting point is the following observation, whose proof can be found in the appendix.

\begin{Lemma}\label{regularization}
	Let $\psi\in L^2(\R^3,\C^4)$ and $\varepsilon>0$. Then,
	$
	e^{-\varepsilon|H|}\,\psi\in C^\infty(\R^3,\C^4).
	$
\end{Lemma}
Because smooth initial data evolve into smooth solutions, the composition of the propagating operator $\mathrm{E}_0$ with the contraction operator in Lemma \ref{regularization}  provides a good candidate for a \textbf{regularization operator}:
	\begin{equation}\label{regulop}
	\gR_{\varepsilon} : \scH_m\ni \mathrm{E}_0(\psi) \mapsto \mathrm{E}_0(e^{-\varepsilon|H|}\psi)\in \scH_m.
	\end{equation}
	This mapping fulfills the following properties, where $\bI_m^\pm$ denote the orthogonal projectors in $\H_m$ corresponding to \eqref{functionE}:
	\begin{itemize}[leftmargin=2.5em]
		\vspace{0.4em}
		\item[{\rm{(i)}}] $\gR_\varepsilon(\mathbb{I}_m^\pm(\scH_m))\subset\mathbb{I}_m^\pm(\scH_m)\cap C^\infty(\R^{1,3},\C^4).$\\[-0.5em]
		\item[{\rm{(ii)}}] $\gR_{\varepsilon}$ is bounded, symmetric, injective and fulfills $\|\gR_{\varepsilon}\|\le 1$\\[-0.5em]
		\item[{\rm{(iii)}}] $\gR_{\varepsilon}u\to u$  as $\varepsilon\searrow 0$ for every $u\in\scH_m$.\\[-0.5em]
		\item[{\rm{(iv)}}] There is $C> 0$\footnote{$C$ may depend on the point $x$ for regularizations of more general type (see \cite[Definition 1.2.3]{cfs}). } such that
		$
		|\gR_{\varepsilon}u (x)|\le C\|u\|_m
		$
		for all $u\in\scH_m$ and $x\in\R^{1,3}$.\\[-0.5em]
		\item[{\rm{(v)}}] For all $x\in\R^{1,3}$ and $\delta>0$ there is $r>0$ such that
		$$
		|\gR_{\varepsilon} u(x)-\gR_{\varepsilon}u (y)|\le \delta \|u\|_m\quad\mbox{for all } u\in\H_m,\ y\in B_r(x).
		$$ 
	\end{itemize}
For a proof of points (i)-(iii) see for example \cite[Proposition 2.5]{neumann}. Points (iv)-(v) can be proved similarly working in momentum representation, where the regularization operator \eqref{regulop} corresponds to a multiplication operator given in terms of a exponentially decaying Schwartz function (see for example \eqref{3D}).

The following definition will  be exploited in Section \ref{sectionexamples}, when the H\"older continuity of the integrated Lagrangian is studied under variations of the regularization operator.
\begin{Def}\label{defevaluation}
	Let $\Omega$ be a topological space and $\H\subset\H_m$ a closed subspace. The space of (spinorial) \textbf{evaluation operators} of $\H$ on $\Omega$ is defined by
	\begin{equation*}
	\begin{split}
	\mathcal{E}(\Omega, \H,\C^4):=\left\{ A\in C^0(\Omega,  \mathfrak{B}(\H,\C^4))\:\bigg|\: \sup_{x\in\Omega}\|A(x)\|_{\mathfrak{B}(\H,\C^4)}<\infty\right\},
	\end{split}
	\end{equation*}
where $\C^4$ is equipped with the standard Euclidean scalar product.
\end{Def} 
By means of such operators, one can represent the vectors in $\H$ in terms of arbitrary continuous spinor-valued functions on $\Omega$:
$$
(Au)(x):=A(x)u\in\C^4,\ \mbox{ for}\ u\in\H,\ x\in\Omega.
$$
The boundedness and continuity conditions will prove useful in Section \ref{sectionexamples}.

Note in particular that $A(x)\in\mathfrak{B}(\H,\C^4)$ for every $x\in\Omega$. Therefore, equipping the spinor space with the canonical 
\begin{equation}\label{spinscalarprodC4}
\textit{Spin scalar product}\quad\Sl a,b\Sr:=a^\dagger\gamma^0 b\quad \mbox{for all $a,b\in\C^4$},
\end{equation}
one can take the adjoint of $A(x)$ with respect to the Hilbert scalar product and the spin scalar product \eqref{spinscalarprodC4}\footnote{The ordinary adjoint with respect to Euclidean scalar product of $\C^4$ is $A(x)^\dagger = A(x)^*\gamma^0$ \label{footnoteadjoint}},
\begin{equation}\label{adjointspin}
	A(x)^*\in\mathfrak{B}(\C^4,\H),\quad \Sl a, A(x)u\Sr = \langle A(x)^*a|u\rangle_m\quad\mbox{for all $a\in\C^4$ and $u\in\H$},
\end{equation}
where again the symbol $\mathfrak{B}(\C^4,\H)$ refers to the  Euclidean scalar product of $\C^4$. 

The regularization operator \eqref{regulop} already provided us with a canonical realization in terms of continuous wave functions. 
In fact, using points (i), (iv) and (v) above, it is not difficult to prove that $\gR_\varepsilon$ fulfills the conditions of Definition \ref{defevaluation}. More precisely, we have the following result.
\begin{Prp}\label{prpregbound}
Let $\gR_{\varepsilon}(x):=(\gR_{\varepsilon}\, \cdot\, )(x)$ for $x\in\R^{1,3}$. Then, $\gR_{\varepsilon}\in \mathcal{E}(\R^4, \H_m,\C^4)$.
\end{Prp}


Such a regularization operator admits a corresponding \textit{regularized distributional kernel}. To obtain this, we modify  \eqref{Edistripm} to
\begin{equation*}\label{regularevolution}
\mathrm{E}^{\pm}_\varepsilon: \mathcal{D}(\R^3,\C^4)\ni \varphi\mapsto \mathrm{E}[e^{-\varepsilon|H|}(\bI^\pm\varphi)]\in  \mathcal{D}'(\R^4\times\R,\C^4).
\end{equation*}
Using again the Schwarz Kernel Theorem, one obtains a regular distribution on $\R^4\times\R^4$, called the \textbf{regularized kernel of the fermionic projector}, which can be explicity represented as follows (cf. \eqref{bidistribution}),
\begin{equation}\label{Pregvacuum}
P_\pm^\varepsilon(x,y):=\int_{\R^4}\frac{d^4 k}{(2\pi)^4}\,\delta(k^2-m^2)\,\Theta(\pm k_0)\,(\slashed{k}+m)\, e^{-\varepsilon |k_0|}\,e^{-ik\cdot (x-y)}.
\end{equation}
Such distributions fulfill the following properties (see \cite{neumann} or \cite{oppio}).\\[-0.9em]
	\begin{itemize}[leftmargin = 2em]
		\item[\rm{(i)}] $P^{\varepsilon}_\pm\in C^\infty(\R^4\times\R^4,\mathrm{Mat}(4,\C))$ and $\|D^\alpha P^\varepsilon_{\pm}\|_\infty<\infty$ for any multiindex $\alpha$.\\[-0.5em]
		\item[\rm{(ii)}] $P^{\varepsilon}_\pm(\,\cdot\,,x) a\in \H_m^\pm$ for all $a\in\C^4$\\[-0.5em]
		\item[\rm{(iii)}] Let $\psi\in L^2(\R^3,\C^4)\cap L^1(\R^3,\C^4)$. Then, 
		$$
		\gR_\varepsilon\big( \mathrm{E}^\pm_s(\psi)\big)(x)=\int_{\mathbb{R}^3}P^\varepsilon_\pm(x,s,\V{y})\,\gamma^0\,\psi(\V{y})\,d^3\V{y}.
		$$
	\end{itemize}
The integral in (iii) is to be understood in the standard sense of Lebesgue's integration. By definition of $P^\varepsilon$, if the function $\psi$ is smooth and has compact support, the same identity can also be understood in the distributional sense.

   In the next section we resume the construction of a causal fermion system associated with the Dirac equation and the regularization operator.

\subsection{Causal Fermion Systems in Minkowski Space}\label{subsectionCFSM}
From now on, we always restrict attention to the negative energy subspace 
$$
\scH_m^{-}:=\mathbb{I}_m^-(\scH_m)
$$ 
Given the regularization $\gR_\varepsilon$, Proposition \ref{regularization}-(iv) and Fr\'{e}chet-Riesz Representation Theorem ensure the existence of a unique operator-valued function
$$
\mathrm{F}^\varepsilon:\R^{1,3}\rightarrow\mathfrak{B}(\scH_m^-)
$$
which encodes information on the local behavior of the wave functions in~$\scH_m^-$ at any point $x\in\R^{1,3}$
via the identity
\begin{equation}\label{defF}
	\langle u| \mathrm{F}^\varepsilon(x)v\rangle=- \Sl \gR_{\varepsilon}u(x) \:|\:\mathfrak{R}_\varepsilon v(x) \Sr\quad\mbox{for any }u,v\in\scH_m^- \:.
\end{equation}

The function~$\mathrm{F}^\varepsilon$ is referred to as the \textbf{local correlation map}.
The construction and a few properties are summarized in the following (see \cite{oppio} or \cite{neumann})
	\vspace{0.3em}
	\begin{itemize}[leftmargin=2.5em]
		\item[\rm{(i)}] $\mathrm{F}^\varepsilon(\R^{4})$ is a closed subset of $\F^\reg$. \\[-0.3cm]
		\item[\rm{(ii)}] $\mathrm{F}^\varepsilon$ is a homeomorphism onto its image.\\[-0.3cm]
		\item[\rm{(iii)}] $\mathrm{F}^\varepsilon(x)$ and $\mathrm{F}^\varepsilon(y)$ are unitarily equivalent for all $x,y\in\R^{1,3}$.\\[-0.3cm]
		\item[\rm{(iv)}] 
		$
		\mathrm{F}^\varepsilon(x)\,u=2\pi\,P^{\varepsilon}(\,\cdot\,,x)\,\big(\gR_{\varepsilon} u\big)(x)\ \mbox{ for all }u\in\H_m^-
		$\\[-0.3cm]
		\item[\rm{(v)}] $\sigma_p(\mathrm{F}^\varepsilon(x))=\{0, \nu^-(\varepsilon), \nu^+(\varepsilon) \}$ with $\nu^+(\varepsilon)>0$ and $\nu^-(\varepsilon)<0$\\[-0.3cm]
		\item[\rm{(vi)}]
		The vectors $e^\varepsilon_{\mu}(x):=P^{\varepsilon}(\,\cdot\,,x)\,\mathfrak{e}_\mu$ with $\mu\in\{1,2,3,4\}$ fulfill
		\begin{align*} 
			\mathrm{F}^\varepsilon(x)\,e^\varepsilon_{\mu}(x)=
			\begin{cases}
				\nu^-(\varepsilon)\, e^\varepsilon_{\mu}(x)& \mu=1,2\\ 
				\nu^+(\varepsilon)\,e^\varepsilon_{\mu}(x) & \mu=3,4,
			\end{cases}
		\end{align*}
		where $\{\mathfrak{e}_\mu\:|\: \mu\in \{1,2,3,4\}\}$ denotes the canonical basis of $\C^4$.\\[-0.4em]
		\item[\rm{(vii)}] $2\pi\,\langle P^\varepsilon(\,\cdot\,,x)a\,|\; P^\varepsilon(\,\cdot\,,yx)b\rangle = -\Sl a\,|\,P^{2\varepsilon}(x,y)b\,\Sr$ for all $a,b\in\C^4$.\\[-0.8em]
	\end{itemize}

The explicit form of the eigenvalues $\nu^\pm(\varepsilon)$ in points (v)-(vi) can be found in \cite{neumann} (where it differs by a factor $2\pi$) and is not important for the purposes of this paper. The only crucial observation is that they do not depend on the point $x$, which is a consequence of translation invariance. Point (iii) is another manifestation of this fact.
Exploiting the rotation-invariance of the converging factor $\exp{(-\varepsilon\omega(\V{k}))}$, the kernel of the fermionic projector becomes (see for instance \cite[Remark 2.10]{neumann})
\begin{equation*}\label{matrixP}
P^{2\varepsilon}(x,x)=
\frac{1}{2\pi}\left(
\begin{matrix}
\nu^-(\varepsilon)\,\bI_2 & 0\\
0 & \nu^+(\varepsilon)\,\bI_2
\end{matrix}
\right).
\end{equation*}

From (vi) we see that the vectors $e_1^\varepsilon(x),e_2^\varepsilon(x)$ belong to the negative spectral subspace of $\mathrm{F}^\varepsilon(x)$, while $e_3^\varepsilon(x),e_4^\varepsilon(x)$ belong to the positive one (compare with Section \ref{sectiongeneralsetting}).   Moreover, it follows from (vii)   that these four vectors are orthogonal. Therefore, after a suitable renormalization, we get the following result. The continuity of \eqref{frameminkowski} can be proved working in momentum space, using \eqref{3D}, \eqref{innerproduct} and Lebesgue's dominated convergence theorem.
\begin{Prp}\label{propbasismink}
The family $\{\hat{e}_1,\hat{e}_2,\hat{e}_3,\hat{e}_4\}$ defined by
\begin{equation}\label{frameminkowski}
\hat{e}^{\,\varepsilon}_\mu:\R^{4}\ni x\mapsto \frac{2\pi }{\sqrt{|\nu_\mu(\varepsilon)|}}\,P^\varepsilon(\,\cdot\,,x)\mathfrak{e}_\mu\in\H_m^-
\end{equation}
is a global Hilbert frame on $\R^{4}$ in the sense of Proposition \ref{existenceframe}. The corresponding global spin frame is
$$
\R^{4}\ni x\mapsto \frac{2\pi }{|\nu_\mu(\varepsilon)|}\,P^\varepsilon(\,\cdot\,,x)\mathfrak{e}_\mu\in\H_m^-.
$$
\end{Prp}

The wave functions spanning the spin spaces have the following three-dimensional representation, which can be obtained by carrying out the $k^0$ integration in \eqref{Pregvacuum},
\begin{equation}\label{3D}
P^\varepsilon(z,x)=-\int_{\R^3}\frac{d^3\V{k}}{(2\pi)^4}\,p_-(\V{k})\,\gamma^0\,e^{-\varepsilon\omega(\V{k})}\,e^{i(\omega(\V{k})(t_z-t_x)+\V{k}\cdot(\V{z}-\V{x}))},
\end{equation}
where
$$
\omega(\V{k}):=\sqrt{\V{k}^2+m^2}\quad\mbox{and}\quad p_-(\V{k}):=\frac{\slashed{k}+m}{2\,k^0}\gamma^0\big|_{k^0=-\omega(\V{k})}\,.
$$
These wave functions are peaked around the light-cone centered at $x$ and decay faster than any polynomial at spatial infinity. More precisely, due to the exponential converging factor in \eqref{3D} and the invariance of the Schwartz space under Fourier transform it follows that
\begin{equation}\label{L1initialdata}
P^\varepsilon(t,\cdot,x)a\in \mathcal{S}(\R^3,\C^4)\subset L^1(\R^3,\C^4)\quad\mbox{for all $t\in\R$}.
\end{equation}
The $\varepsilon$-scaling of these distinguished quantum states can be exploited to recover the causal structure of Minkowski space, for details see \cite{neumann}.

Exploiting the continuity of the function $\mathrm{F}^\varepsilon$,
it is possible to take the push-forward of the Lebesgue measure of $\R^4$ into $\F$.
	The causal fermion system on $\H_m^-$ induced by  $$\rho_{\rm{vac}}:=(\mathrm{F}^\varepsilon)_*(d^4 x)$$ is referred to as a
	{\bf{regularized Dirac sea vacuum}}.
Points (i) and (ii) below \eqref{defF} imply a one-to-one correspondence between points in Minkowski space
and points in the support of the measure~$\rho_{\rm{vac}}$. More precisely,
\begin{equation*} \label{id1}
M_{\rm{vac}}:=\supp\rho_{\rm{vac}}=\mathrm{F}^\varepsilon(\R^{1,3})\subset\F^\reg \:.
\end{equation*}
Moreover, for every $x\in\R^{1,3}$ there is a canonical identification of the space of Dirac spinors
with the spin space at the corresponding point of~$\mathrm{F}^\varepsilon(x) \in M_{\rm{vac}}$
\begin{equation} \label{id2}
\Phi_x:= \gR_{\varepsilon}(x)|_{S_{\mathrm{F}^\varepsilon(x)}}: S_{\mathrm{F}^\varepsilon(x)}=\mathrm{Im}\  \mathrm{F}^\varepsilon(x) \ni u\mapsto \gR_{\varepsilon} u(x)\in  \C^4
\end{equation}
(for details see~\cite[Proposition~1.2.6]{cfs} or~\cite[Theorem~4.16]{oppio}).
This identification is unitary, i.e. it is surjective and it preserves the spin scalar products.
As a final remark, we point out that the identification~\eqref{id2} allows for an explicit realization of the abstract kernel of the fermionic projector defined in~\eqref{Pxydef} in terms of the regularized distribution introduced in~\eqref{Pregvacuum} (for details see for example~\cite[Theorem 5.19]{oppio}),
\begin{equation}\label{id3}
\begin{split}
	\Phi_x\, \mathrm{P}\big(\mathrm{F}^\varepsilon(x), \mathrm{F}^\varepsilon(y) \big)\,\Phi_y^{-1}&=2\pi\,P^{2\varepsilon}(x,y)
	\end{split}
\end{equation}
The identification~\eqref{id3} also provides us with a corresponding realization of the
closed chain \eqref{closedchain}
in terms of the product of the two regularized distributions,
\begin{equation}\label{id4}
	\begin{split}
		A_{xy}^\varepsilon&:=(2\pi)^{-2}\:\Phi_x\: \mathrm{A}_{\mathrm{F}^\varepsilon(x)\mathrm{F}^\varepsilon(y)}\,\Phi_x^{-1}= P^{2\varepsilon}(x,y)P^{2\varepsilon}(y,x).
	\end{split}
\end{equation}
The eigenvalues of this matrix will be briefly discussed in Section \ref{sectionlagrangian}.

\section{The Fermionic Projector of the Dirac Sea in Minkowski Space}\label{sectionfermionicprojex}

In this section we analyze the explicit form of the fermionic projector of  Dirac sea vacua in Minkowski space and study its $L^4$-integrability. As a corollary, we will obtain integrability of the Lagrangian.

\subsection{An Explicit Form in Terms of Bessel Functions}\label{sectionexplicityproj}

The starting point is identity \eqref{3D}, which gives
\begin{equation*}\label{espressioneP}
\begin{split}
P^{\varepsilon}(x,y)=\sum_{j=0}^3v_j(x-y)\: \gamma^j+\beta(x-y) \quad\mbox{for all }x,y\in\R^{1,3},
\end{split}
\end{equation*}
for smooth functions (with raised indices)
\vspace{0.05cm}
\begin{equation}\label{functionsvbeta}
\begin{split}
v^0(\xi)&=-\frac{1}{2}\int_{\R^3}\frac{d^3\V{k}}{(2\pi)^4}\,e^{-\varepsilon\omega(\V{k})}\, e^{-i(-\omega(\V{k})\xi^0-\V{k}\cdot\boldsymbol{\xi})}\\
v^\alpha(\xi)&= \frac{1}{2}\int_{\R^3}\frac{d^3\V{k}}{(2\pi)^4}\frac{k^\alpha}{\omega(\V{k})}\,e^{-\varepsilon\omega(\V{k})}\, e^{-i(-\omega(\V{k})\xi^0-\V{k}\cdot\boldsymbol{\xi})}\\
\beta(\xi)&=\frac{1}{2}\int_{\R^3}\frac{d^3\V{k}}{(2\pi)^4}\frac{m}{\omega(\V{k})}\,e^{-\varepsilon\omega(\V{k})}\, e^{-i(-\omega(\V{k})\xi^0-\V{k}\cdot\boldsymbol{\xi})} \:.
\end{split}
\end{equation}
Note that
\begin{equation}\label{derivatV}
v^j =\frac{i}{m}\,\partial^j \beta,\quad j=0,1,2,3.
\end{equation}
Using the results of Appendix \ref{appendixbessel}, it is possible to compute these expressions explicitly in terms of the Bessel functions of the second type. \\[-0.5em]
\begin{itemize}[leftmargin=2em]
	\item[(i)]{\em The contribution $\beta$: } Because this function is spherically symmetric in the spatial variable, it is more convenient to evaluate it at $\xi=(\xi^0,se_3)$, for $s>0$. In polar coordinates $(r,\theta,\phi)$:
	\begin{equation*}
	\begin{split}
	\quad\beta(\xi)&=\frac{m}{2(2\pi)^4}\,2\pi\int_0^\infty dr \frac{r^2}{\sqrt{m^2+r^2}}\,e^{-(\varepsilon-i\xi^0)\sqrt{r^2+m^2}}\underbrace{\int_0^\pi d\theta\,\sin\theta\,e^{irs\cos\theta}}_{2\frac{sin(sr)}{sr}}\\
	&=\frac{m}{(2\pi)^3}\frac{1}{s}\int_0^\infty dr\frac{r}{\sqrt{r^2+m^2}}\,e^{-(\varepsilon-i\xi^0)\sqrt{r^2+m^2}}\sin(sr)
	\end{split}
	\end{equation*}\\[0.1em]
	Introducing the notation
	$$
	\xi_\varepsilon:=(\xi^0+i\varepsilon, \boldsymbol{\xi}),
	$$
	the following result is a direct consequence of Lemma \ref{lemmaintK}.
	\begin{Lemma}
		The contribution $\beta$ in \eqref{functionsvbeta} has the form
		\begin{equation*}\label{beta}
		\beta(\xi)=\frac{m^2}{(2\pi)^3}\,\frac{K_1\big(m\sqrt{-\xi_\varepsilon^2}\big)}{\sqrt{-\xi_\varepsilon^2}}\quad\mbox{for all }\xi\in\R^{1,3}.
		\end{equation*}
	\end{Lemma}
	\noindent Before moving on, we stress that the following function is analytic:
	\begin{equation}\label{G}
	G:\Omega_\pi\ni z\mapsto \frac{m^2}{(2\pi)^3}\frac{K_1\big(m\sqrt{z}\big)}{\sqrt{z}}\in\C
	\end{equation}
	(see Appendix \ref{appendixbessel} for details and notation). Moreover, notice that,  for all $\xi\in\R^{1,3},$
	\begin{equation}\label{squarereg}
	-\xi_\varepsilon^2=(-\xi_0^2+\boldsymbol{\xi}^2+\varepsilon^2)-2i\varepsilon\xi_0\in \Omega_\pi.
	\end{equation}
	With this notation, the contribution $\beta$ can be rewritten more compactly as
	\begin{equation}\label{Gbeta}
	\beta(\xi)=G(-\xi_\varepsilon^2)
	\end{equation}
	Differentiating $G$, and exploiting the two identities in Lemma \ref{lemmaderivK}, gives
	\begin{equation}\label{derivG}
	\begin{split}
	G'(z)&= -\frac{m^2}{2(2\pi)^3}\left(\frac{m}{2}\frac{K_0\big(m\sqrt{z}\big)+K_2\big(m\sqrt{z}\big)}{z}+\frac{K_1\big(m\sqrt{z}\big)}{(\sqrt{z})^3}\right)\\
	&=-\frac{m^3}{2(2\pi)^3}\frac{K_2\big(m\sqrt{z}\big)}{z}.
	\end{split}
	\end{equation}
	\item[(ii)]{\em The contribution $v^j$: } 
	Using \eqref{derivatV}, \eqref{Gbeta} and \eqref{derivG} we can now compute the integral expression defining the functions $v^j$:
	\begin{Lemma}
		The contribution $v^j$ in \eqref{functionsvbeta} has the form
		\begin{equation*}\label{vj}
		v^j(\xi)=i\frac{m^2}{(2\pi)^3}\,\frac{K_2\big(m\sqrt{-\xi_\varepsilon^2}\big)}{-\xi_\varepsilon^2}\,\xi_\varepsilon^j\quad\mbox{for all $\xi\in\R^{1,3}$}\:.
		\end{equation*}
	\end{Lemma}
	
	\noindent Similar considerations as for $\beta$ can be done for the functions $v^j$: these can be rewritten more compactly as
	\begin{equation}\label{F}
	v^j(\xi)=F(-\xi_\varepsilon^2)\,\xi_\varepsilon^j,\quad j =0,1,2,3,
	\end{equation}
	where $F$ is the analytic function
	\begin{equation}\label{F2}
	F:=\frac{2}{im}\, G':\Omega_\pi\rightarrow\C.
	\end{equation}
\end{itemize}
\vspace{0.2cm}
Putting all together, the following result follows.
\begin{Prp}\label{propoexplicitP}
	The kernel of the fermionic projector has the following form:\\[-0.5em]
	\begin{equation}\label{expansionP}
	\begin{split}
	P^\varepsilon(x,y)&=F(-\xi_\varepsilon^2)\,\slashed{\xi}_\varepsilon+G(-\xi_\varepsilon^2)\\[0.1em]
	&=i\frac{m^2}{(2\pi)^3}\,\frac{K_2\big(m\sqrt{-\xi_\varepsilon^2}\big)}{-\xi_\varepsilon^2}\,\slashed{\xi}_\varepsilon+\frac{m^2}{(2\pi)^3}\,\frac{K_1\big(m\sqrt{-\xi_\varepsilon^2}\big)}{\sqrt{-\xi_\varepsilon^2}},
	\end{split}
	\end{equation}
	where $F$ and $G$ are defined in \eqref{G} and \eqref{F2}.
\end{Prp}

Now, we want to apply the asymptotics of the Bessel functions from Appendix \ref{appendixbessel} in order to prove $L^4$-integrability of the kernel of the fermionic projector: this is the content of Proposition \ref{integrabilityP}.
To this aim, let us first introduce the following notation:
$$
\mbox{For all $M\in\mathrm{Mat}(4,\C)$ we define }\ |M|_2:= \sup\{|Ma|\:|\: a\in \C^4,\ |a|=1\}
$$
The proof of finiteness of the right integral in \eqref{finitenessP4} is postponed to Appendix \ref{appendixproofL4}, while the first inequality follows immediately from \eqref{id4} and the following identities (see \eqref{Pregvacuum}):
\begin{equation}\label{commutxy}
|\gamma^0|_2=1\quad\mbox{and}\quad P^\varepsilon(y,x)=\gamma^0\,P^\varepsilon(x,y)^\dagger\,\gamma^0.
\end{equation}
The independence of the integrals from the point $x$ is again a manifestation of translation invariance of the Dirac sea. 
\begin{Prp}\label{integrabilityP}
	Let $x\in\R^{1,3}$. Then,
	\begin{equation}\label{finitenessP4}
\int_{\R^4}\big|A^\varepsilon_{xy}\big|_2^2\ d^4y\le \int_{\R^4} |P^{2\varepsilon}(x,y)|_2^4\ d^4y <\infty
	\end{equation}
	Moreover, both the right and the left integrals are independent from $x$.
\end{Prp}

The proof of this proposition is based on several estimates which follow from the asymptotics of the Bessel functions. 
Away from the light cone, the regularized kernel \eqref{expansionP} converges locally uniformly to a smooth function in the limit $\varepsilon\searrow 0$, which is obtained by simply setting $\varepsilon$ equal zero (see \cite[Section 1.2.5]{cfs} or \cite[Section 2.4]{neumann}). This limit function clearly provides a smooth representation of \eqref{bidistribution} in this region.  Therefore, away from the lightcone, the same asymptotics of the regularized kernel apply to \eqref{bidistribution} and similar estimates can be carried out.

\section{The Lagrangian of the Dirac Sea in Minkowski Space}\label{sectionlagrangian}

\subsection{An Explicit Form of the Lagrangian on Minkowski Space}

We want to obtain an explicit expression for the Lagrangian of the Dirac sea. To this aim, we first need to compute the spectrum of the closed chain \eqref{id4}
\begin{equation}\label{closedchain2}
A_{xy}^\varepsilon:=P^{2\varepsilon}(x,y)P^{2\varepsilon}(y,x).
\end{equation}
The explicit form of the eigenvalues of this matrix can be found in \cite[Lemma 2.6.1]{cfs}. Here, we simply state the result.
\begin{Prp}
	The spectrum of \eqref{closedchain2} is formed by two eigenvalues $\lambda_\pm^{xy}$, each with algebraic multiplicity two. Explicitly,
	\begin{equation*}\label{eigenvalues}
	\lambda_\pm^{xy} := v\overline{v}+\beta\overline{\beta}\pm\sqrt{(v\overline{v})^2-v^2\overline{v}^2+(v\overline{\beta}+\beta\overline{v})^2}\,\bigg|_{\xi:=x-y}.
	\end{equation*}
	where $v,\beta$ are the smooth functions introduced in \eqref{functionsvbeta}, with $\varepsilon$ replaced by $2\varepsilon$. 
\end{Prp}
The argument of the square root is a real number, but the sign can vary. For negative arguments the eigenvalues form a complex conjugate pair.
For simplicity we henceforth use the notation $\lambda_\pm=a\pm\sqrt{b}$, where a $a,b$ are real-valued.

 By means of the homeomorphism 
 \begin{equation}\label{localcorrelationfunction}
 \mathrm{F}^\varepsilon:\R^{1,3}\rightarrow M_{\mathrm{vac}},
 \end{equation}
 we can pull  the abstract structures on $M_{\mathrm{vac}}$ back to Minkowski space $\R^{1,3}$, as was done for example with the kernel of the fermionic projector in \eqref{id3} and the closed chain in \eqref{id4}. In particular, referring to \eqref{eigenvaluesxy}, we have 
 $$
 \lambda_+^{xy}=\lambda_1^{\mathrm{F}^\varepsilon(x)\,\mathrm{F}^\varepsilon(y)}=\lambda_2^{\mathrm{F}^\varepsilon(x)\,\mathrm{F}^\varepsilon(y)}\ \mbox{ and }\  \lambda_-^{xy}=\lambda_3^{\mathrm{F}^\varepsilon(x)\,\mathrm{F}^\varepsilon(y)}=\lambda_4^{\mathrm{F}^\varepsilon(x)\,\mathrm{F}^\varepsilon(y)}.
 $$
 In the same way, we can pull  the abstract \textbf{causal structure} introduced after Proposition \ref{proplagrangian} back to Minkowski Space and of course compare it with the preexisting one. The corresponding causal relations between different points are then determined by the sign of $b$. We can state it as a definition.

\begin{Def}
	Two Minkowski Space points $x,y\in\R^{1,3}$ such that
	$$
	b(x,y)=0,\ b(x,y)> 0\ \mbox{ or }\ b(x,y)< 0,
	$$
	are said to be \textbf{lightlike, timelike} or \textbf{spacelike separated} (in the regularized sense).
\end{Def}
Accordingly, one can introduce corresponding (regularized)  \textbf{causal cones}
\begin{equation*}
\begin{split}
L_x^\varepsilon:&=\{y\in\R^{1,3}\:|\: b(x,y)=0 \}\\
I_x^\varepsilon:&=\{y\in\R^{1,3}\:|\: b(x,y)>0 \}\\
J_x^\varepsilon:&=\{y\in\R^{1,3}\:|\: b(x,y)\ge 0 \}
\end{split}
\end{equation*}

The regularized causal structure is connected to the values attained by the Lagrangian. 
Note that the abstract Lagrangian as in Proposition \ref{proplagrangian} is defined on all of $\F\times\F$. Restricting for the moment to spacetime points which correspond to points in Minkowski Space through the local correlation function \eqref{localcorrelationfunction}, the Lagrangian reduces to the following function on $\R^{1,3}\times\R^{1,3}$:
\begin{equation}\label{L2eigen}
\begin{split}
\L^\varepsilon(x,y):=\mathcal{L}(\mathrm{F}^\varepsilon(x),\mathrm{F}^\varepsilon(y))&=\big(|\lambda_+^{xy}|-|\lambda_-^{xy}|\big)^2=2\big(a^2+|b|-|a^2-b|\big),
\end{split}
\end{equation}
where in the last step we used the notation $a\pm \sqrt{b}$ for the eigenvalues.
This expression can be simplified even further. Indeed, note that
\begin{equation*}
\begin{split}
&(v^2-\beta^2)(\overline{\beta}^2-\overline{v}^2)=-|v^2-\beta^2|\le 0,\quad\mbox{hence}\\
&v^2\overline{\beta}^2-v^2\overline{v}^2+\beta^2\overline{v}^2\le \beta^2\overline{\beta}^2,\quad\mbox{hence}\\
&b=(v\overline{v})^2+v^2\overline{\beta}^2-v^2\overline{v}^2+\beta^2\overline{v}^2+2\beta\overline{\beta}v\overline{v}\le \beta^2\overline{\beta}^2 + (v\overline{v})^2+2\beta\overline{\beta}v\overline{v}=a^2.
\end{split}
\end{equation*}
Therefore, $a^2\ge b$ and hence $|a^2-b|=a^2-b$. As a consequence, \eqref{L2eigen} reduces to\footnote{By $|q|_+:=(q+|q|)/2$ we denote the \textit{positive part} of the real number $q$.}
	\begin{equation}\label{lagrangianlemma}
	\L^\varepsilon(x,y)=
	2(|b(x,y)|+b(x,y))=4|b(x,y)|_+\,,
	\end{equation}

An explicit expression of \eqref{lagrangianlemma} can be given using the Bessel functions as in Proposition \ref{propoexplicitP}. Let us analyze the quantity $b(x,y)$ in more detail.\\[-0.8em]
\begin{itemize}[leftmargin=2em]
	\item[{\rm (i)}] First, we show that $(v\overline{v})^2-v^2\overline{v}^2=-4|F|^4\,\varepsilon^2\,|\boldsymbol{\xi}|^2$, with $F$ as in \eqref{F2}. 
	This follows from \eqref{F}, which gives, with $\xi:=y-x$,
	\begin{equation*}
	\begin{split}
	\quad(v\overline{v})^2-v^2\overline{v}^2&=|F|^4\left[\big(\xi_{\varepsilon j}\,\overline{\xi_\varepsilon^j}\big)^2-\big(\xi_{\varepsilon i}\,\xi_\varepsilon^i\big)\overline{\big(\xi_{\varepsilon h}\,\xi_\varepsilon^h\big)}\right]\\
	&=|F|^4\left[\big(|\xi_0+i\varepsilon|^2-|\boldsymbol{\xi}|^2\big)^2-\big|(\xi_0+i\varepsilon)^2-|\boldsymbol{\xi}|^2\big|^2\right]\\
	&=|F|^4\left[\big(\xi^2+\varepsilon^2\big)^2-\big|\xi^2-\varepsilon^2+2i\varepsilon\xi_0|^2\right]\\
	&=|F|^4\left[\big((\xi^2)^2+\varepsilon^4+2\varepsilon^2\xi^2\big)-\big((\xi^2)^2+\varepsilon^4-2\varepsilon^2\xi^2+4\varepsilon^2\xi_0^2\big)\right]\\
	&=|F|^4\left[4\varepsilon^2\xi^2-4\varepsilon^2\xi_0^2\right]=-4|F|^4\varepsilon^2|\boldsymbol{\xi}|^2
	\end{split}
	\end{equation*}
\item[{\rm (ii)}] The remaining contribution can be rewritten as follows, with $G$ as in \eqref{G},
\begin{align*}
(v\overline{\beta}+\beta\overline{v})^2&=2\,\mathrm{Re}\big((F\,\overline{G})^2\,\xi_\varepsilon^2\big)+2|F|^2|G|^2\ \xi_\varepsilon\cdot\overline{\xi_\varepsilon}
\end{align*}
\end{itemize}
Summarizing, we have the following result.
\begin{Lemma}\label{bbessel}
	The Langrangian of the Dirac sea vacuum in Minkowski Space  has the following explicit expression on $\R^{1,3}$:\\[-0.2em]
	\begin{equation*}
	\begin{split}
	\L(\mathrm{F}^\varepsilon(x),\mathrm{F}^\varepsilon(y))&=4|b(x,y)|_+\\[0.2em] b(x,y):&=2\,\mathrm{Re}\big((F\,\overline{G})^2\,\xi_\varepsilon^2\big)+2|F|^2|G|^2\ \xi_\varepsilon\cdot\overline{\xi_\varepsilon}-4|F|^4\varepsilon^2|\boldsymbol{\xi}|^2,
	\end{split}
	\end{equation*}
	where $F$ and $G$ are defined in \eqref{F2}, \eqref{G} and are evaluated at $-\xi_\varepsilon^2$.
\end{Lemma}
As we already know from Section \ref{sectioncontinuitylag}, the Lagrangian vanishes for spacelike separated points (in the regularized sense), and in this sense it is \textit{causal}.

The features of the regularized light cone could in principle be analyzed using the explicit expression for $b$ as in  Lemma \ref{bbessel}. In particular,
\begin{equation*}\label{conditionlightcone}
\begin{cases}
\ y\in L_x^\varepsilon\quad \mbox{if and only if}\quad\xi:=y-x\quad\mbox{satisfies}\\[0.5em]
\  4|F|^4\varepsilon^2|\boldsymbol{\xi}|^2-2\,\mathrm{Re}\big((F\,\overline{G})^2\,\xi_\varepsilon^2\big)-2|F|^2|G|^2\ \xi_\varepsilon\cdot\overline{\xi_\varepsilon}=0
\end{cases}
\end{equation*}
This, however, is not straightforward and only the behavior for small and large vectors can be analyzed, exploiting the asymptotics of the Bessel functions. Nevertheless, this goes beyonds the scope of the present paper and will not be discussed here.

\subsection{Integrability of the Lagrangian on Minkowski Space}

The integrability of the Lagrangian in Lemma \ref{bbessel} is a direct consequence of, \eqref{L2eigen}, Lemma \ref{integrabilityP} and the fact that
$$
|\lambda_\pm^{xy}|\le \big|A^\varepsilon(x,y)\big|_2\le |P^\varepsilon(x,y)|_2^2,
$$
which follows from  \eqref{commutxy} and general properties of matrices and eigenvalues.
\begin{Prp}\label{propinteglagrangian}
	For all $x\in\R^4$, 
		$$
		\int_{\R^4}\big|\lambda_\pm^{0y}\big|^2\,d^4y=\int_{\R^4}\big|\lambda_\pm^{xy}\big|^2\,d^4y<\infty.
		$$
As a consequence, 
	$$
	\int_{\R^4}	\L(\mathrm{F}^\varepsilon(x),\mathrm{F}^\varepsilon(y))\,d^4y = \int_{\R^4}	\L(\mathrm{F}^\varepsilon(0),\mathrm{F}^\varepsilon(y))\,d^4y <\infty.
	$$
	In particular, $M_{\mathrm{vac}}\subset \mathrm{Adm}(\rho_{\mathrm{vac}})$ and $\ell\circ\mathrm{F}^\varepsilon$ is constant.
\end{Prp} 
As in Proposition \ref{integrabilityP}, the independence of the integrals above from the point $x$  is a direct manifestation of the translation invariance of the Dirac sea.

\vspace{0.5em}

\section{H\"older Continuity in Minkowski Space}\label{sectionholdervariation}

Up to now, we have evaluated the abstract Lagrangian \eqref{Lagrangian} only on spacetime points of the form $\x=\mathrm{F}^\varepsilon(x)$. However, the Lagrangian contributes to the causal action also through points which lie outside the support of the measure $\rho_{\rm vac}$. Studying which other points in $\F$ are actually \textit{admissible} for this measure is our next goal. The determination of \textit{all of these points} is  out of reach, and hence one needs to tackle the problem in a different way. We proceed in connection with Section \ref{sectionholder} by studying the behavior of the integrated Lagrangian \eqref{ell} under \textit{perturbations} of spacetime points.

\subsection{Regular Variations of Spacetime}

The goal of this section is to study the admissibility of points of $\F$ which lie in a neighborhood of $M_{\mathrm{vac}}$. This is achieved by implementing \textit{variations} of the local correlation functions. Such variations can be obtained, for example, by changing the regularization parameter $\varepsilon>0$. This and other few examples will be discussed later in Section \ref{sectionexamples}. In this paper we will be concerned with the following type of variations.

\begin{Def}\label{regularvariation}
	A \textbf{regular variation} of the local correlation function is a mapping
	$$
	\mathrm{F}\in C^0(I\times\R^{4},\F^\reg)\quad\mbox{with $I=(-\delta,\delta)\subset\R$ and }\ \mathrm{F}(0,\,\cdot\,)=\mathrm{F}^\varepsilon.
	$$
\end{Def}

Such a definition is quite general as we will see later in Section \ref{sectionexamples}. In principle, however, one could also include the following additional conditions: For all $\lambda\in I$
\begin{itemize}[leftmargin=2.5em]
	\item[\rm{(i)}] $\im \mathrm{F}(\lambda,\,\cdot\,)\subset\F^\reg$ is closed,\\[-0.8em]
	\item[\rm{(ii)}] $\mathrm{F}(\lambda,\,\cdot\,)$ is a homeomorphism onto its image,
\end{itemize} 
which hold true in the case of a Dirac sea vacuum (see (i)-(ii) in Section \ref{subsectionCFSM}).
 Under these assumptions, one would obtain that
\begin{equation*}\label{rholambda}
\begin{split}
M_\lambda :=\supp \rho_\lambda=\mathrm{F}(\lambda,\R^4),\quad\mbox{where  }\ \rho_\lambda:=\mathrm{F}(\lambda,\,\cdot\,)_*(d^4x).
\end{split}
\end{equation*}
Moreover, the composition
\begin{equation}\label{homeomorphism}
\mathrm{F}(\lambda,\,\cdot\,)\circ \mathrm{F}(0,\,\cdot\,)^{-1}: M_{\mathrm{vac}}\rightarrow M_\lambda,
\end{equation}
would be a homeomorphism between Dirac sea vacuum spacetime and its perturbation.
Nevertheless, for the sake of generality, we will not assume such conditions at this stage. The drawback is that, in this case 
$$
M_\lambda:=\supp \mathrm{F}(\lambda,\,\cdot\,)_*(d^4 x)=\overline{\mathrm{F}(\lambda,\R^{4})},
$$
and the function \ref{homeomorphism} between spacetimes may not be invertible.

To every regular variation $\mathrm{F}$, a new class of spin spaces is assigned, which depend continuously on the parameter $\lambda\in I$:
$$
S_{\mathrm{F}(\lambda,x)}:= \mathrm{F}(\lambda, x)(\H_m^-),\quad S_{\mathrm{F}(0,x)}=S_{\mathrm{F}^\varepsilon(x)}.
$$

From Proposition \ref{existenceframe}, we know that it is possible to find a local Hilbert frame around any point $(\lambda,x)$. However, this is not needed for what comes next and the following weaker notion will suffice for the goals of this section.
\begin{Def}
	Let $X$ be a set and $f:X\rightarrow \F^\reg$. A family $\underline{e}:=\{e_1,e_2,e_3,e_{4}\}$ of functions $e_\mu:X\rightarrow\H$ such that
	\begin{equation}
		\{ e_\mu(x),\, \mu=1,2,3,4\}\quad\mbox{is a Hilbert basis of $S_{f(x)}$ for all $x\in X$},
	\end{equation}
is referred to as a \textbf{bare Hilbert frame} of $f$.
\end{Def}

Clearly, every Hilbert frame as in Proposition \ref{existenceframe} is in particular a bare Hilbert frame. The term \textit{bare} was chosen with the intention of emphasizing that, differently from \eqref{existenceframe}, no topological assumptions are made on $X$, $f$ and $e_\mu$.
\begin{Lemma}\label{lemmafirstestimate}
For any regular variation $F$ and  any bare Hilbert frame $\underline{e}$ of $F$,
	\begin{equation}\label{estimate1}
		\int_{\R^4}\left\|\mathrm{P}(\mathrm{F} (\lambda, x),\mathrm{F}^\varepsilon(y))\right\|^4\,d^4y\le C|P^{2\varepsilon}(0,0)|^2_2\,\sum_{\mu=1}^{4}\|\gR_{\varepsilon}\, e_\mu(\lambda,x)\|^4_{L^4},
	\end{equation}
where the constant $C>0$  is independent of both $F$ and $\underline{e}$.
\end{Lemma}
\begin{proof}
	Let us analyze the integrand. Let $\|\cdot\|_{\mathrm{HS}}$ denote the Hilbert-Schmidt norm. Then, there is $C>0$ which depends only the spin dimension such that, for all $\lambda,x,y$,\\[-1em]
	\begin{equation*}
		\begin{split}
			\left\|\mathrm{P}(\mathrm{F}(\lambda, x),\mathrm{F}^\varepsilon(y))\right\|^2&=\|\pi_{\mathrm{F}(\lambda,x)}\mathrm{F}^\varepsilon(y)\|^2\le C\|\pi_{\mathrm{F}(\lambda, x)}\mathrm{F}^\varepsilon(y)\|^2_{HS}=C\|\mathrm{F}^\varepsilon(y)\pi_{\mathrm{F}(\lambda, x)}\|^2_{HS}.
		\end{split}
	\end{equation*}	
At this point, using the fact that $\{e_\mu(x),\, \mu=1,2,3,4\}$ is by definition a Hilbert basis of $S_{\mathrm{F}(\lambda,x)}$ and using the definition of Hilbert-Schmidt norm, we obtain:
\begin{equation*}
	\begin{split}
			\|\mathrm{F}^\varepsilon(y)\pi_{\mathrm{F}(\lambda, x)}\|^2_{HS}&= \sum_{\mu=1}^{4}\la 	e_\mu(\lambda,x)|\pi_{\mathrm{F}^\varepsilon(\lambda, x)}\mathrm{F}^\varepsilon(y)\mathrm{F}^\varepsilon(y)\pi_{F^\varepsilon(\lambda, x)}e_\mu(\lambda,x)\ra_m\\
			&=(2\pi)^2\,\sum_{\mu=1}^{4}\la 	P^\varepsilon(\,\cdot\,,y)\gR_{\varepsilon} e_\mu(\lambda,x)(y)|P^\varepsilon(\,\cdot\,,y)\gR_{\varepsilon} e_\mu(\lambda,x)(y)\ra_m\\
			&=2\pi\,\sum_{\mu=1}^{4}|\Sl \gR_{\varepsilon}\, e_\mu(\lambda,x) (y)| P^{2\varepsilon}(y,y)\,\gR_{\varepsilon}\, e_\mu(\lambda,x)(y)\Sr|\\
			&\le 2\pi\,\,|P^{2\varepsilon}(0,0)|_2\,\sum_{\mu=1}^{4}|\gR_{\varepsilon}\,e_\mu(\lambda,x) (y)|^2,
		\end{split}
	\end{equation*}
where we used the properties of $\mathrm{F}^\varepsilon$ in Section \ref{subsectionCFSM} and the fact that $P^\varepsilon(y,y)$ is constant in $y$, due to translation invariance. The claim follows from  $(a+b)^2\le 2(a^2+b^2)$.
\end{proof}

We now improve these estimates, by showing that the $L^4(\R^4,\C^4)$ norm of the wave functions in \eqref{estimate1} can be replaced by the $L^1(\R^3,\C^4)$ norm of the corresponding initial data, giving in this way better control on the estimates. To this aim, let us introduce the notation (see the beginning of Section \ref{sectiondiracequation})
\begin{equation*}\label{form22}
u_\psi:=\mathrm{E}_0(\psi)\in\H_m,\quad\psi\in L^2(\R^3,\C^4).
\end{equation*}
We now focus our analysis on the space of solutions with integrable initial data. 
\begin{Lemma}
The space of \textit{negative energy solutions with $L^1$ initial data}
 $$
 \H(L^1):=\{u_\psi\in\H_m^-\:|\:\psi\in  L^1(\R^3,\C^4)\}
 $$
 is a dense subspace of $\H_m^-$. 
Moreover, for all $u_\psi\in \H(L^1)$,
\begin{equation}\label{integralexpress}
\gR_{\varepsilon} u_\psi(x)=\int_{\R^3} P^\varepsilon(x,0,\V{z})\,\gamma^0\,\psi(\V{z})\,d^3\V{z}\quad \mbox{  for all $x\in\R^4$}
\end{equation}
\end{Lemma}
\begin{proof}
	The denseness follows from the fact that  	(see for example \cite[Lemma 2.17]{oppio}).
	$$
	\bI^-(S(\R^3,\C^4))\subset S(\R^3,\C^4)\subset L^1(\R^3,\C^4).
	$$
	The second statement appears already at the end of  Section \ref{sectionreg}.
\end{proof}
As a particular example, we see from \eqref{L1initialdata} that, for all $x\in\R^4$ and $a\in\C^4$,
\begin{equation}\label{PareL1}
	P^\varepsilon(\,\cdot\,,x)a\in \H(L^1).
\end{equation}
\begin{Remark}
	Note that the specific choice $\{t=0\}$ in the definition of $\H(L^1)$ has no relevance in what follows and was made out of mere simplicity. The forthcoming analysis applies identically to any other choice of Cauchy surface $\{t=\mathrm{const}\}$.
	It is important to stress, however, that the definition of  $\H(L^1)$ depends on the choice made, in that the condition of being $L^1$ is in general not preserved in time. As a consequence, different Cauchy surfaces may correspond to different spaces $\H(L^1)$. This ambiguity could be removed by making additional assumptions: For example, notice that a solution is of Schwartz type at a given time if and only if it is of  Schwartz type at every time (see the general form in \cite[Proposition 2.19]{oppio}).  Nevertheless, the existence of a single instant of time with the above property suffices for our purposes and the mentioned ambiguity will therefore no longer addressed in this paper.
\end{Remark}
\begin{Lemma}\label{L4intR}
 For any $u_\psi\in \H(L^1)$ the following inequality holds,
	$$
	\|\gR_{\varepsilon} u_\psi\|_{L^4}\le \|P^\varepsilon(0,\,\cdot\,)\|_{L^4}\,\|\psi\|_{L^1}<\infty.
	$$
	In particular,
	 $
	 \gR_{\varepsilon}(\H(L^1))\subset L^4(\R^4,\C^4).
	 $
\end{Lemma}
\begin{proof} From \eqref{integralexpress} it follows that, for every $u_\psi\in \H(L^1)$, 
\begin{equation*}
	\begin{split}
	|\gR_{\varepsilon} u_\psi(x) |^4&\le \int_{\R^{3}}\!d^3\V{z}_1\!\int_{\R^{3}}\!d^3\V{z}_2\!\int_{\R^{3}}\!d^3\!\V{z}_3\int_{\R^{3}}\!d^3\V{z}_4|\psi(\V{z}_1)||\psi(\V{z}_2)||\psi(\V{z}_3)||\psi(\V{z}_4)|\times\\[0.1em]
			&\qquad \qquad\qquad\times |P^\varepsilon(x,0,\V{z}_1)|_2|P^\varepsilon(x,0,\V{z}_2)|_2|P^\varepsilon(x,0,\V{z}_3)|_2|P^\varepsilon(x,0,\V{z}_4)|_2\\
\end{split}
\end{equation*}
Thus, from Fubini Theorem, the generalized H\"older inequality and the translational invariance of $P^\varepsilon$, 
\begin{equation*}
\begin{split}
			\int_{\R^4}|\gR_{\varepsilon} u_\psi(x)|^4\,d^4x&\le \left(\int_{\R^{3}}|\psi(\V{z})|\,d^3\V{z}\right)^4\,\int_{\R^4}|P^\varepsilon(x,0)|_2^4\,d^4x,
	\end{split}
\end{equation*}
 The last integral is finite, thanks to Proposition~\ref{integrabilityP}. 
\end{proof}
With this in mind, we now focus on trasformations of the local correlation functions whose corresponding spin spaces are generated by $L^1$ initial data.
\begin{Def}
	A regular variation $\mathrm{F}$ is  of \textbf{$\boldsymbol{L^1}$-type} if for all $\lambda\in I$ and $x\in \R^4$,
	$$
	S_{\mathrm{F}(\lambda,x)}=\mathrm{Im}\, \mathrm{F}(\lambda,x)\subset \H(L^1).
	$$
\end{Def}
 For such transformations, one can combine Lemma \ref{lemmafirstestimate} with Lemma \ref{L4intR} and improve estimates \eqref{estimate1}, getting better control on the $L^4$ norm. 
\begin{Prp}\label{prplastestimate}
Let $\mathrm{F}$ be of $L^1$-type and $\underline{e}$ a bare Hilbert frame of $F$ with
 \begin{equation}\label{form}
 e_\mu(\lambda,x)=u_{\psi_\mu(\lambda,x)}\quad\mbox{where  }\ \psi_\mu(\lambda,x)\in L^1(\R^3).
 \end{equation}
  Then, for all $\lambda\in I$ and $x\in\R^4$,
\begin{equation*}\label{estimate2}
\int_{\R^4}\left\|\mathrm{P}(\mathrm{F}(\lambda,x),\mathrm{F}^\varepsilon(y))\right\|^4\,d^4y\le K\,\sum_{\mu=1}^{4}\|\psi_\mu(\lambda,x)\|_{L^1}^4
\end{equation*}
where the constant $K$ can be chosen as $C\,|P^{2\varepsilon}(0,0)|_2^2\,\|P^\varepsilon(0,\,\cdot\,)\|^{4}_{L^4}$.
\end{Prp}
This estimate can now be exploited in order to prove local H\"older continuity of the Lagrangian along the continuous paths in $\F^\reg$ generated by the variation $\mathrm{F}$.
We first recall that, due to translation invariance, every
$\mathrm{F}^\varepsilon(x)$ is unitarily equivalent to $\mathrm{F}^\varepsilon(0)$.
The factor $\|\mathrm{g}(\mathrm{F}^\varepsilon(x))\|$ in \eqref{boundintP} is then constant and Theorem \ref{teoremaholder} applies. This, together with Proposition \ref{prp4admiss}, gives the next conclusive result.
\begin{Thm}\label{teoremvariationmionk}
	Let $\mathrm{F}$ of $L^1$-type. Then, the following statements hold.\\[-0.8em]
	\begin{itemize}[leftmargin = 2em]
		\item[{\rm (i)}] $\mathrm{F}(\lambda,\R^4)\subset \mathrm{Adm}(\rho_{\rm vac})$ for all $\lambda\in I$\\[-0.5em]
		\item[{\rm (ii)}] Let $\underline{e}$ be a bare Hilbert frame of $F$ 
		on an open set $W\subset I\times \R^4
		$ such that, referring to notation \eqref{form},
		\begin{equation}\label{boundpsi}
		\sup\big\{\|\psi_\mu(\lambda,x)\|_{L^1}\:\big|\:  (\lambda,x)\in W,\ \mu=1\dots 4\,\big\}< \infty.
		\end{equation}
		Then, every $(\lambda_0,x_0)\in W$ has an open neighborhood $W_1\subset W$ and  $C>0$, such that for all $(\lambda,x)\in W_1$,
		$$
		|\ell(\mathrm{F}(\lambda,x))-\ell(\mathrm{F}(\lambda_0,x_0))|\le C\|\mathrm{F}(\lambda_0,x_0)\|^{2-\alpha}\,\|\mathrm{F}(\lambda,x)-\mathrm{F}(\lambda_0,x_0)\|^\alpha
		$$
		In particular, the function $\ell\circ \mathrm{F}$ is continuous on $W$.
	\end{itemize}
\end{Thm}
Condition \eqref{boundpsi} can of course be taken as the starting point in the construction of a transformation of the local correlation function, provided that the \textit{defined} varied wave functions $\psi_\mu(\lambda,x)$ do generate Hilbert bases of the spin spaces. This will be discussed in more detail in the next section.

\subsection{A Few Examples}\label{sectionexamples}
In this section we analyze in a few examples how  variations of the local correlation function can be implemented concretely and how the $L^1$-condition can be arranged.
\subsubsection*{Example 1:} \textbf{Rescaling of the Regularization Operator}. Let us rescale the local correlation operators as follows:
	$$
	\mathrm{F}(\lambda,x):=\mathrm{F}^{\varepsilon+\lambda}(x)\quad\mbox{with }\lambda\in (-\varepsilon,\varepsilon).
	$$
	The results in Section \ref{subsectionCFSM} do not depend on the parameter $\varepsilon>0$. In particular, the conditions (i) and (ii) after Definition \ref{regularvariation} are fulfilled. From Proposition \ref{propbasismink} we know that
	\begin{equation*}
	\begin{split}
	e_\mu(\lambda,x):&=\hat{e}^{\,\varepsilon+\lambda}_\mu(x)\\ &=\frac{2\pi}{\sqrt{|\nu_\mu(\varepsilon+\lambda)|}}\, P^{\varepsilon+\lambda}(\,\cdot\,,x)\mathfrak{e}_\mu,\quad\mbox{with } \mu\in\{1,\dots,4\}
	\end{split}
	\end{equation*}
	depend continuously on $\lambda$ and define pointwise Hilbert bases of the corresponding spin spaces. Finally, note that, thanks to \eqref{PareL1},
	$$
	e_\mu(\lambda,x)\in \H(L_1).
	$$
	In conclusion, $\mathrm{F}$ is a regular variation of $L^1$ type.

\subsubsection*{Example 2:} 
	 \textbf{One-Parameter (Semi-)Groups of Operators.} Consider a densely defined self-adjoint operator $A$ on $\H_m^-$ and define a corresponding variation of the local correlation function as
	$$
	\mathrm{F}(\lambda,x):= e^{i\lambda A}\,\mathrm{F}^\varepsilon(x)\,e^{-it\lambda A}.
	$$
	The conditions (i) and (ii) after Definition \ref{regularvariation} are trivially fulfilled.
	However, $\mathrm{F}$ is not of $L^1$-type in general. To arrange this, one needs to 
	add additional assumptions like
	$$
	e^{i\lambda A}(\H(L^1))\subset \H(L^1).
	$$
	Operators of this kind are, for example, the translation generators  $P_\alpha$ in $\R^3$ and the rotation generators $J_\alpha$ in $\R^3$. It should be noted, however, that these two examples only provide internal transformations of spacetime that do not affect the support of the measure.
	
	Similar transformations can be realized by replacing the operator $A$ above with $iA$, under the additional assumption that $A\ge 0$. In this case, the family $e^{-\lambda A}$ defines a \textit{contraction semi-group}. Note that both variations are trivial at a point $x$ whenever the corresponding operator $\mathrm{F}^\varepsilon(x)$ commute with $A$.

\subsubsection*{Example 3:} \textbf{Variations of the Initial Data.} 
	Let us  project the spin spaces to the initial-data space
	$$
	\mathfrak{S}_\x:=\mathrm{E}_0^{-1}(S_\x)\subset L^2(\R^3,\C^4)\quad\mbox{for all }\x\in\F.
	$$
	Let $\{\hat{e}^{\,\varepsilon}_\mu\}_\mu$ be the global Hilbert frame as in \eqref{frameminkowski}. Then, the functions
	$$
	\mathfrak{S}_{F^\varepsilon(x)}\ni \psi_\mu(x):=\mathrm{E}_0^{-1}(\hat{e}^{\,\varepsilon}_\mu)\in S(\R^3,\C^4)\subset L^1(\R^3,\C^4),
	$$ 
	define an orthonormal basis of $\mathfrak{S}_{F^\varepsilon(x)}$.
	One can modify these functions into new wave functions $\psi_\mu(\lambda, x)$, so that they depend continuously on $(\lambda,x)$,
	remain orthogonal to each other and, most importantly, stay in $L^1(\R^3,\C^4)$. This can always be done, at least locally. As an example, consider four arbitrary functions
	$$
	f_\mu\in \mathrm{E}_0^{-1}(\H_m^-)\cap L^1(\R^3,\C^4)\ \mbox{ with}\ \|f_\mu\|_{L^2}=1,
	$$ 
	and define, for any $x\in\R^4$,
	$$
	\psi_\mu(\lambda,x):= \psi_\mu(x)+\lambda f_\mu\in \mathrm{E}_0^{-1}(\H_m^-)\cap L^1(\R^3,\C^4)
	$$
	If $\lambda$ is chosen sufficiently small, the above functions are linearly independent (see for example \cite[Lemma 5.2]{oppio}). By means of a Gram-Schmidt process, the above functions can be turned into an orthonormal set which depend continuously on both $x$ and $\lambda$.
	
	Going back to the general case, given the functions $\psi_\mu(x,\lambda)$, the corresponding transformed local correlation function can then be defined as
	$$
	\mathrm{F}(\lambda,x):= \sum_{\mu=1}^4 \nu_\mu(\lambda,x) \langle \mathrm{E}_0(\psi_\mu(\lambda,x)),\,\cdot\,\rangle_m\,\mathrm{E}_0(\psi_\mu(\lambda,x))
	$$
	where $\nu_\mu$ are arbitrary continuous functions which preserves the signature and reduce to $\nu_\mu(\varepsilon)$ in the limit $\lambda\to 0$.
\subsubsection*{Example 4:} \textbf{A General Method.} We now discuss a method which generalizes the examples discussed above. 
 Thanks to Proposition \ref{prpregbound}, we can reinterpret the regularization operator (restricted to $\H_m^-$) as an \textit{evaluation operator}
 $$
 \gR_{\varepsilon}\in \mathcal{E}(\R^4,\H_m^-, \C^4),\quad\mbox{with}\quad \gR_{\varepsilon}(x)\in\mathfrak{B}(\H_m^-,\C^4).
 $$
 In view of what comes next, it is also important to note that
 \begin{equation}\label{surjectivity}
 \gR(x):\H_m^-\rightarrow \C^4\quad\mbox{is surjective for all $x\in\R^{1,3}$},
 \end{equation}
 which is a consequence of regularity (see end of Section \ref{subsectionCFSM}, in particular \eqref{id2}).
 Referring to the spin scalar product in $\C^4$ (see \eqref{adjointspin}), we have the following characterization, which is a direct consequence of \eqref{defF}, \eqref{surjectivity} and (iv) in Section \ref{subsectionCFSM}.
 \begin{Prp}\label{relationRF}
 For all $x\in\R^{1,3}$ and all $a\in\C^4$,\\[-1.1em]
 	\begin{itemize}[leftmargin=2.5em]
 		\item[{\rm (i)}] $\mathrm{F}^\varepsilon(x)=-\gR_{\varepsilon}(x)^*\,\gR_{\varepsilon}(x).$\\[-0.7em]
 		\item[{\rm (ii)}] $\gR_{\varepsilon}(x)^*a=-2\pi\, P^\varepsilon(\,\cdot\,,x)a$
 	\end{itemize}
 \end{Prp}
Now, the idea is to modify the action of $\gR_\varepsilon(x)$ at every point $x\in\R^{1,3}$ and  take point (i) in Proposition \eqref{relationRF} as the \textit{definition} of a new local correlation operator.
 This approach is very general, as the following result shows.
 \begin{Thm}\label{localrepr}
 	Let $\mathrm{F}:\R^4\rightarrow \F^\reg$ be continuous.
 	Then, for every $x_0\in\R^4$ there exists an open neighborhood $\Omega_{x_0}\subset \R^4$ and an evaluation operator
 	\begin{equation}\label{Phi}
 	\Psi\in \mathcal{E}(\Omega_{x_0},\H_m^-,\C^4)
 	\end{equation}
 	such that,   for every $x\in \Omega_{x_0}$,
 	$
 	\Psi(x)
 	$
 	is surjective and  satisfies
 	$$
 	\mathrm{F}(x)=-\Psi(x)^*\Psi(x).
 	$$
 
 	\noindent Assume further that $\mathrm{F}$ is bounded in the operator norm, it admits a global Hilbert frame (as in Proposition \ref{existenceframe}) and it satisfies
 	\begin{equation}\label{assumptionbounded}
 	\mathrm{tr}\,|\mathrm{F}(x)|\ge   k>0\quad\mbox{for all $x\in\R^4$}.
 	\end{equation}
 	Then, the mapping $\Psi$ in \eqref{Phi} can be chosen globally on $\R^4.$
 \end{Thm}
In a few words, any continuous  realization of Minkowski space in $\F^\reg$ is locally equivalent to a pointwise variation of the regularized vectors of $\H_m^-$.   Whether the evaluation operator $\Phi$ can be chosen globally or not, depends on the global topological and spectral properties of $\mathrm{F}$. Note that the condition of boundedness from below of the trace norm \eqref{assumptionbounded} is fulfilled in examples like Minkowski Dirac sea vacua.  Therefore, for the applications in mind, it seems sufficient to restrict attention to \textit{global} transformation of this kind.  
 \begin{Def}\label{defvarreg}
 	A \textbf{variation of the regularization operator} is a family of evaluation operators of the form
 	$$
 	\Psi\in \mathcal{E}(I\times\R^4,\H_m^-,\C^4)\quad\mbox{with}\quad \Psi(0,\,\cdot\,)=\gR_{\varepsilon},
 	$$
 	such that $\Psi(\lambda,x)$ is surjective for every $\lambda\in I$ and $x\in\R^4$.
 \end{Def}
The last assumption in the definition above is needed in order to ensure regularity. More precisely, we have the following result.
\begin{Prp}\label{propF}
	Let $\Psi$ be a variation of the regularization operator. Then,
	\begin{equation}\label{Fpsi}
	\mathrm{F}(\lambda,x):=-\Psi(\lambda,x)^*\,\Psi(\lambda, x)\in\scF^\reg\quad\mbox{for every }x\in \R^{1,3}.
	\end{equation}
	The function $\mathrm{F}:I\times\R^4\rightarrow \mathfrak{B}(\H_m^-)$ is continuous and bounded. 
\end{Prp}
In this case, conditions (i) and (ii) after Definition \ref{regularvariation} are not automatically satisfied, and their validity needs to be checked case by case. Also the $L^1$-condition depends on how we modify the regularization operator. From \eqref{PareL1} and point (ii) in Definition \ref{relationRF} we see that
$$
\gR_\varepsilon(x)^*(\C^4)\subset \H(L^1).
$$
In modifying $\gR_{\varepsilon}(x)$ one then needs to make sure that the adjoint of the evaluation operator maps again into $\H(L^1)$. This needs to be checked case by case. 
\subsubsection*{Example 5:} \textbf{Perturbation Theory.} 
Although the general method presented in Example $4$ could in principle be applied to most of the interesting scenarios, a perturbative approach may sometimes be preferrable: This applies especially when only the first low-order terms of a variation are known.

To understand this, let us assume that $\Psi(\lambda,x)$ and $\mathrm{F}(\lambda,x)$ from Proposition \ref{propF} can be expanded  in Taylor form around $\lambda=0$ up to order $n\in\N$, i.e.
\begin{equation}\label{pertexpansionpsi}
	\begin{cases}\ 
\Psi(\lambda,x)=\displaystyle \sum_{p=0}^n \lambda^p\,\Psi^{(p)}(x)+ R_{n+1}(\lambda,x),& \Psi^{(0)}=\gR_{\varepsilon},\\[0.6em] \ \mathrm{F}(\lambda,x)=\displaystyle\sum_{p=0}^{n}\lambda^p\,\mathrm{F}^{(p)}(x)+Q_{n+1}(\lambda,x),&\, \mathrm{F}^{(0)}=\mathrm{F}^\varepsilon,
\end{cases}
\end{equation}
where $R_{n+1},Q_{n+1}$ are remainder terms. The quantity $X^{(p)}$ is called the \textit{p-th order perturbation term} of $X$. Such an expansion can always be arranged, assuming that the functions $\Psi$ and $F$ are sufficiently regular. The convergence for $n\to\infty$ is instead a trickier matter and may in general not hold.

An \textit{$n$-th order perturbation  of Minkowski vacuum} will then consist in neglecting the reminder term in \eqref{pertexpansionpsi} and focussing the analysis on the truncated sum. It is important to stress, though, that such an approximation of $\mathrm{F}(\lambda,x)$ will in general \textit{no longer belong to $\F$}. 
One can nonetheless study to which extent such a truncation fulfills the assumptions of a regular variation of $L^1$ type. 

Were we for example only interested in the first-order effects of a  variation of  spacetime, we would need to truncate \eqref{pertexpansionpsi} to order $p=1$, obtaining (cf. \eqref{Fpsi})
\begin{equation}\label{firstorderF}
	\begin{split}
\mathrm{F}(\lambda,x)&\stackrel{1\mathrm{st}}{=} \mathrm{F}^{(0)}(x)+\lambda\,\mathrm{F}^{(1)}(x) \\
&=\mathrm{F}^\varepsilon(x)+\lambda\big(-\gR_{\varepsilon}(x)^*\,\Psi^{(1)}(x)-\Psi^{(1)}(x)^*\,\gR_{\varepsilon}(x)\big),
\end{split}
\end{equation}
and the first-order perturbative analysis will then boil down to studying the term between brackets in \eqref{firstorderF}. 

We will now see this in action in the concrete example of variations of  Minkowski Dirac sea vacua induced by electromagnetic fields. We will show that  \eqref{firstorderF} is continuous and bounded for such perturbations and ranges within $\H(L^1)$.

\vspace{0.3cm}

\noindent \textbf{Electromagnetic Fields in Minkowski Space.}
 In presence of an electromagnetic field the Dirac operator \eqref{Diracequation} is modified via the minimal coupling $\partial_j\mapsto \partial_j-iA_j$, where $A\in C^\infty(\R^4,\R^4)$ is a corresponding \textit{electromagnetic potential}, i.e.
 $$
 \mathrm{D}\ \longmapsto\ \mathrm{D}+\slashed{A} :C^\infty(\R^4,\C^4)\to C^\infty(\R^4,\C^4).
 $$ 
 The Cauchy problem with  smooth initial data at time $t\in\R$ \eqref{CP}  becomes
\begin{equation}\label{diracpotCP}
	\left\{
	\begin{split}
		(\mathrm{D}+\slashed{A})f&=0\\
		f(t,\cdot)&=\varphi\in C^\infty(\R^3,\C^4)
	\end{split}\; .
\right.
\end{equation}
As in the free case, it follows from the theory of symmetric hyperbolic systems that the Cauchy problem \eqref{diracpotCP} admits a unique global smooth solution.

Identifying the initial data at time $t$, it is possible to define a one-to-one linear correspondence between smooth solutions of the free Dirac equations and smooth solutions in presence of an external field:
\begin{equation}\label{identification}
\mathcal{P}_A:\ker\mathrm{D}\ni h\to \mathcal{P}_A[h]\in \ker(\mathrm{D}+\slashed{A})\quad\mbox{with }\	 \mathcal{P}_A[h](t,\cdot)=h(t,\cdot\,)\:.
\end{equation}

This identification can be studied order by order in perturbation theory by means of the following iterative method. Let us focus on the space $\H^-_m\cap C^\infty(\R^4,\C^4)$ of smooth solutions with $L^2$ initial data and negative energy. Assume for simplicity that the electromagnetic potential has compact support and choose $t$ sufficiently small so that the potential lies in the future of $t$, i.e.
\begin{equation}\label{condiA}
A\in C_0^\infty(\R^4,\R^4)\quad\mbox{and}\quad \supp A \subset\{x\in\R^4\:|\: x^0>t \}.
\end{equation}
Starting from the unperturbed solution   (see Section \ref{sectiondiracequation})
$$
u_\varphi:=\mathrm{E}_t(\varphi)\in\H_m^-\cap C^\infty(\R^4,\C^4),\quad \varphi\in L^2(\R^3,\C^4)\cap C^\infty(\R^3,\C^4),
$$ 
one defines $\mathcal{P}_A^{(0)}[u_\varphi]:=u_\varphi$ and solves for $p\ge 1$ the inductive Cauchy problems
\begin{equation}\label{inductiveCP}
	\left\{\, 
	\begin{split}
		\mathrm{D}\,\mathcal{P}_A^{(p)}[u_\varphi]&=-\slashed{A}\,\mathcal{P}_A^{(p-1)}[u_\varphi]\\[0.3em]
		\mathcal{P}_A^{(p)}[u_\varphi](t,\cdot)&=0
	\end{split}\; ,
\right.
\end{equation}
which again admit unique global smooth solutions. This gives rise to the functions
\begin{equation}\label{contriP}
\mathcal{P}_A^{(p)}:\H_{m}^-\cap C^\infty(\R^4,\C^4)\ni u_\varphi\mapsto \mathcal{P}_A^{(p)}[u_\varphi] \in C^\infty(\R^4,\C^4),\quad p\in\N_0.
\end{equation}

A formal computation shows that the sum of all the contributions \eqref{contriP} does indeed solve \eqref{diracpotCP}. Without entering the mathematical details concerning the convergence of such a series, we conclude from the uniqueness of the solution of the Cauchy Problem that \eqref{identification} admits (at least formally) the following \textit{perturbation expansion}\footnote{This method is independent of the choice of $t$, as long as it lies in the past of the support of the electromagnetic potential.}:
\begin{equation}\label{identificationPpertub}
\mathcal{P}_A:\H_m^-\cap C^\infty(\R^4,\C^4)\ni u_\varphi\mapsto \mathcal{P}_A[u_\varphi]=\sum_{p=0}^\infty \mathcal{P}_A^{(p)}[u_\varphi].
\end{equation}

In order to make a clearer connection with the perturbation expansion \eqref{pertexpansionpsi} we now multiply the electromagnetic potential by a coupling parameter:
$$
A\mapsto \lambda A,\quad \lambda\in (-\delta,\delta).
$$
By means of the identification \eqref{identification} and using that $\gR_\varepsilon u\in \H_m^-\cap C^\infty(\R^4,\C^4)$ we can define a corresponding variation of the regularization operator by setting
\begin{equation}\label{newpsi}
	\Psi(\lambda,x):\H_m^-\ni u\mapsto \mathcal{P}_{\lambda A}(\gR_{\varepsilon}u)(x)\in\C^4
\end{equation}
(note that for $\lambda=0$ the variation \eqref{newpsi} does indeed give back the original regularization operator). 
On the other hand, the uniqueness of the solution of the Cauchy problem ensures that
\begin{equation*}\label{homogP}
\mathcal{P}_{\lambda A}^{(p)}=\lambda^p\,\mathcal{P}_A^{(p)}\quad\mbox{for all $p\in\N_0$}.
\end{equation*}
Putting all together, we conclude that the \textit{p-th order perturbation term} of \eqref{newpsi} is
	\begin{equation}\label{nthorder}
	\Psi^{(p)}(x):\H_m^-\ni u\mapsto \mathcal{P}_A^{(p)}[\gR_\varepsilon u](x)\in\C^4\, ,
	\end{equation}

Some of the arguments carried out so far (as for example the convergence of \eqref{identificationPpertub}) are just formal and would need a more solid mathematical treatment. Moreover, whether \eqref{newpsi} fulfills the conditions in Definition \ref{defvarreg} is not straightforward and will not be discussed here in full generality.  We will focus on the first-order perturbation term, of which a simple explicit representation exists. 
\vspace{0.3cm}

\noindent \textbf{A First-Order Perturbation Analysis.}
The starting point are the so-called \textit{retarded} and \textit{advanced Green's functions} of the Dirac equation,
$$
s^\wedge_m,\ s^\vee_m\in \mathcal{D}'(\R^4,\mathrm{Mat}(4,\C)),
$$
which are defined as as the unique distributional solutions of
\begin{equation*}\label{retadvgreen}
	\begin{split}
\mathrm{D}\,s_m=\delta^{(4)}\bI_4,\quad\mbox{with}\quad \supp s_m^\wedge\subset J_0^\vee\quad\mbox{and}\quad \supp s_m^\vee\subset J_0^\wedge,
\end{split}
\end{equation*}
respectively (they are in fact tempered distribution: for an explicit form of the kernel of these distribution see for example \cite[Section 2.1.3]{cfs}). 
One can check by direct computation that, in the distributional sense, the unique solution of the Cauchy problem \eqref{inductiveCP} for $p=1$ is given by the convolution
\begin{equation}\label{P1}
	\mathcal{P}_A^{(1)}[u_\varphi]=-s^\wedge_m*(\slashed{A}u_\varphi)
\end{equation}
Note that this function is indeed smooth, being the convolution of a distribution against a compactly supported smooth function. We then have the first statement of the following result, where a formal integral representation is used. The proof of points (i)-(v) is postponed to Appendix \ref{appendixproofs}.
\begin{Prp}\label{propfirstderiv}
	Let $A$ be as in \eqref{condiA}. Then, referring to \eqref{nthorder} and \eqref{P1} the first-order perturbation term is given by
	\begin{equation}\label{firstorder}
		\begin{split}
\Psi^{(1)}(x)u&=-\int_{\R^4} s_m^\wedge(x-y)\,\slashed{A}(y)\,\gR_\varepsilon u(y)\,d^4y\,,\quad \mbox{for $x\in\R^{1,3}$ and $u\in\H_m^-$}.
		\end{split}
	\end{equation}
Moreover, the following properties hold (cf. \eqref{firstorderF}).\\[-1em]
\begin{itemize}[leftmargin=2.5em]
	\item[\rm{(i)}] $\Psi^{(1)}\in \mathcal{E}(\R^4,\H_m^-,\C^4)$,\\[-0.7em]
	\item[\rm{(ii)}] $\mathrm{F}^{(1)}(x):= -\gR_{\varepsilon}(x)^*\,\Psi^{(1)}(x)-\Psi^{(1)}(x)^*\,\gR_{\varepsilon}(x)$ is bounded and self-adjoint,\\[-0.7em]
	\item[\rm{(iii)}]  $\mathrm{F}^{(1)}$ is continuous and bounded in the $\sup$-norm topology,\\[-0.7em]
	\item[\rm{(iv)}] $\mathrm{F}^{(1)}(x)(\H_m^-)\subset\H(L^1)$. More precisely, for every $t\in\R$ and $u\in\H_m^-$
	$$
	(\mathrm{F}^{(1)}(x)u)(t,\,\cdot\,)\in S(\R^3,\C^4),
	$$
	\item[\rm{(v)}] $\mathrm{F}^{(1)}(x)=0$ for every $x\not\in J^\vee(\supp A)$.
\end{itemize}
\end{Prp}
As we know from Section \ref{subsectionCFSM}, properties (iii) and  (iv) are also satisfied by the unperturbed local correlation function $\mathrm{F}^\varepsilon$. Let $I=(-\delta,\delta)$  and let us define the following first-order variation of the local correlation function (cf. \eqref{firstorderF})
$$
\mathrm{F}^{1\mathrm{st}}(\lambda,x):=\mathrm{F}^\varepsilon(x)+\lambda\, \mathrm{F}^{(1)}(x)\quad\mbox{for all $x\in \R^4$ and $\lambda \in I$}.
$$
We  see that $\mathrm{F}^{1\mathrm{st}}$ is continuous and bounded in the $\sup$-norm topology and satisfies
$$
(\mathrm{F}^{1\mathrm{st}}(\lambda,x)u)(t,\,\cdot\,)\in S(\R^3,\C^4)\quad\mbox{for all $x\in\R^4, \lambda\in I$ and $u\in\H_m^-$}.
$$

In summary, we have shown that external electromagnetic fields with compact support generate \textit{to first order}  $L^1$-type variations of the local correlation function.
Moreover, point (v) above shows that the support of the vacuum measure is modified to first order only in the future causal cone of the support of $A$, i.e. 
$$
\mathrm{F}^{1\mathrm{st}}(\lambda,\,\cdot\,)|_{J^\vee(\supp A)}=\mathrm{F}^\varepsilon|_{J^\vee(\supp A)}.
$$
This concludes our first-order analysis of the variation of the local correlation function induced by external electromagnetic fields.


\appendix
	
\section{Modified Bessel Functions of the Second Kind}\label{appendixbessel}

In this appendix we review some basic properties of the modified Bessel functions of the second type which will be used in the proof of Proposition \eqref{integrabilityP}.\\[-0.8em]

\textbf{Some Identities.} 
We adopt the convention $\arg z\in (-\pi,\pi]$. For any $\alpha\in (0,\pi]$ we define the open set 
$$
\Omega_\alpha:=\{z\in\C\setminus\{0\}\:|\: \arg z\in (-\alpha,\alpha) \}.
$$  
In particular, $\alpha=\pi$ gives the standard cut along the non-negative real axis. 

Let $K_n$ denote the general \textit{modified Bessel function of the second type of order $n\in \N_0$} on the cut plane $\Omega_\pi$. More precisely, $K_n$ is defined as the unique holomorphic function on $\Omega_\pi$ which solves the differential equation
$$
z^2\,w''(z)+z\,w'(z)-(z^2+n^2)w(z)=0
$$
and tends to zero in the limit $z\to \infty$ in the region $\Omega_{\pi/2}$. For more details we address the interested reader to \cite[Section 9.6]{AS}). In particular, the following recurrence relations can be found in (9.6.26).
\begin{Lemma}\label{lemmaderivK}
	For any $n\ge 1$ the following identities hold:
	\begin{equation*}\label{derivK}
	K_n'=-\frac{1}{2}\big(K_{n-1}+K_{n+1}\big)\quad\mbox{and}\quad K_{n+1}=K_{n-1}+\frac{2n}{z}K_n.
	\end{equation*}
\end{Lemma}
These Bessel functions appear in the explicit form of the kernel of the fermionic projector (see Section \ref{sectionexplicityproj}). In particular, the following identity proves useful in their computation (see (3.914-6) in \cite{GR}).
\begin{Lemma}\label{lemmaintK}
	Let $b,\gamma\in\R$ and $\beta\in\C$ with $\gamma>0$ and $\mathrm{Re}\,\beta> 0$. Then,
	\begin{align}
	\int_0^\infty \frac{r}{\sqrt{r^2+\gamma^2}}\,e^{-\beta\sqrt{r^2+\gamma^2}}\,\sin(br)\,dr&=\frac{\gamma b}{\sqrt{\beta^2+b^2}}\,K_1\big(\gamma\sqrt{\beta^2+b^2}\big)\label{D}
	\end{align}
\end{Lemma} 
Consider  $\beta=\beta_r+i\beta_i$ as in Lemma \ref{lemmaintK}. The argument of $K_1$ in the right-hand side of \eqref{D} has the form
$$
\sqrt{\beta^2+b^2}=\sqrt{(\beta_r^2-\beta_i^2+b^2)+2i\beta_r\beta_i}.
$$
If $\beta_i=0$ then the argument of the square root is real and strictly positive, as $\beta_r>0$. If, on the other hand, $\beta_i\neq 0$, the assumption $\beta_r>0$ ensures that the argument lies away from the semiaxis $(-\infty,0]$.  Therefore, taking into account that also the square root is analytic in this domain, we see that the function in \eqref{D} is analytic in $\beta$.
\vspace{0.4cm}

\textbf{Asymptotics}\label{asymptoticsbess}
In this section we are interested in the asymptotic behavior of the Bessel functions at the limit points zero and infinity. 

Let us first fix the notation. Consider two continuous functions
$
f,g
$
on a common open domain $\Omega\subset \C$.
Let $z_0\in\overline{\Omega}$. We say that \\[-0.8em] 
\begin{center}
$f\sim_{z_0} g$ whenever $f=g+h$ for some $h\in C^0(\Omega)$ satisfying: \\[0.4em] 
For all $\varepsilon>0$ there is $R>0$ such that $|h(z)|\le \varepsilon|g(z)|$ for all $z\in\Omega\cap B(z_0,R)$.\\[0.6em]
\end{center}
This definition applies also to $z_0=\infty$, with $B(\infty,R):=\{|z|\ge R\}$.
From the above condition one infers the following estimate:
\begin{equation}\label{estimateasympt}
|f(z)|\le (1+\varepsilon)|g(z)|\quad\mbox{for $z\in \Omega\cap B(z_0,R)$}.
\end{equation}
We then have the following result (see (9.6.8), (9.6.9) and (9.7.2) in \cite{AS}).
\begin{Lemma}
	The Bessel functions have the following asymptotic behaviors  
	\begin{equation}\label{asymptoticexpK}
	\begin{split}
	&K_{n\ge 0}(z)\sim_\infty \sqrt{\frac{\pi}{2z}}\,e^{-z}\quad\mbox{ on}\quad \Omega_{\pi/2},\\
	& K_0(z)\sim_0 -\ln z,\quad K_{n\ge 1}(z)\sim_0\frac{\Gamma(n)}{2}\left(\frac{2}{z}\right)^n\quad\mbox{ on}\quad\Omega_\pi.
	\end{split}
	\end{equation}
\end{Lemma}
\noindent Such asymptotic behaviors are exploited in the proof of the $L^4$-integrability of the fermionic projector.

\section{Proof of $L^4$-Integrability of the Fermionic Projector}\label{appendixproofL4}

This appendix is entirely devoted to the proof of Proposition \ref{integrabilityP}. \\[-0.8em]

Let us start by noting from \eqref{expansionP} and  $(a+b)^2\le 2(a^2+b^2)$ that, for some $C>0$,
\begin{equation}\label{estimateP}
\begin{split}
|P^\varepsilon(x,y)|_2^4&\le C\left[ \left|\frac{K_2\big(m\sqrt{-\xi_\varepsilon^2}\big)}{-\xi_\varepsilon^2}\right|^4|\slashed{\xi}_\varepsilon|_2^4+\left|\frac{K_1\big(m\sqrt{-\xi_\varepsilon^2}\big)}{\sqrt{-\xi_\varepsilon^2}}\right|^4\right].
\end{split}
\end{equation}
In what follows we will exploit the asymptotic behavior of the Bessel functions to estimate \eqref{estimateP}. Because of translation invariance, there is no loss of generality in assuming $x=0$. Moreover, since the set $\{\xi^0=0 \}$ is a null subset of $\R^4$, it suffices to restrict  attention to the subset
$$
\R^{1,3}_0=\{\xi\in\R^{1,3}\:|\: |\xi_0|>0 \}.
$$
From \eqref{squarereg} we know that, for any $\xi\in\R^{1,3}_0$, 
\begin{equation*}\label{squarereg2}
	\begin{cases}
\ -\xi_\varepsilon^2=-\xi^2+\varepsilon^2-2i\varepsilon\xi_0\\[0.2em]
\ \arg(-\xi_\varepsilon^2)=\arg(-\xi^2+\varepsilon^2-2i\varepsilon\xi_0)\in  \left(-\pi,\pi\right)
\end{cases}
\end{equation*}
Therefore, 
\begin{equation}\label{squarelower}
\arg\sqrt{-\xi_\varepsilon^2}\in \left(-\frac{\pi}{2},\frac{\pi}{2}\right),\quad\mbox{or equivalently }\sqrt{-\xi_\varepsilon^2}\in \Omega_{\pi/2}\quad \mbox{for all }\xi\in \R^{1,3}_0,
\end{equation}
In the limit $|\xi|_{\R^4}\to \infty$ the complex number $\sqrt{-\xi_\varepsilon^2}$ goes to infinity.  
This can be seen from the inequality
\begin{equation*}\label{estimatexi0}
|\sqrt{-\xi_\varepsilon^2}|=\sqrt[4]{(-\xi^2+\varepsilon^2)^2+4\varepsilon^2 \xi_0^2}\ge \frac{\sqrt{\varepsilon|\xi_0|}+\sqrt{|-\xi_0^2+\boldsymbol{\xi}^2+\varepsilon^2|}}{2}.
\end{equation*}
We can therefore apply the asymptotics of the Bessel functions of Appendix \ref{asymptoticsbess}: using \eqref{estimateasympt} and \eqref{asymptoticexpK}, we conclude that, for some $A>0$ and $R>0$,
\begin{equation}\label{firstestimate}
	|K_n(m\sqrt{-\xi_\varepsilon^2})|\le A\,\frac{e^{-m\,\mathrm{Re}(\sqrt{-\xi_\varepsilon^2})}}{\sqrt[4]{|\xi_\varepsilon^2|}},\quad \mbox{for all } \xi\in\R^{1,3}_0,\ |\xi|_{\R^4}\ge R.
\end{equation}
Without loss of generality,  we will henceforth assume that $\varepsilon< 1/2$.  For simplicity of notation, we will denote $|\boldsymbol{\xi}|$ by $r$ and $|\xi_0|$ by $t$.

If decaying, the exponential in \eqref{firstestimate} may prove useful in the proof of integrability. Let us study in which directions this is indeed the case. Let us write out the exponent explicitly: using\eqref{squarelower} we have, for any $\xi\in \R^{1,3}_0$ (see (3.7.27) in \cite{AS}),
\begin{equation}\label{exponent}
	\begin{split}
		\mathrm{Re}(\sqrt{-\xi_\varepsilon^2})&=\sqrt{\frac{|-\xi_\varepsilon^2|+\mathrm{Re}(-\xi_\varepsilon^2)}{2}}=\\
		&=\sqrt{\frac{\sqrt{(-t^2+r^2+\varepsilon^2)^2+4\varepsilon^2 t^2}+(-t^2+r^2+\varepsilon^2)}{2}},
	\end{split}
\end{equation}
We have the following cases.
\begin{itemize}[leftmargin=2em]
	\item[(1)] 
	Along the vertical $r=0$ line \eqref{exponent} is identically equal to $\varepsilon$. Along the vertical lines $r=r_0\neq 0$ \eqref{exponent} is not costant but it converges to $\varepsilon$ in the limit $t\to \infty$. In conclusion, there is no decaying along any vertical line.\\[-0.6em]
	\item[(2)] Let us consider non-vertical lines of the form $r=\mu t$, for $\mu \in (0,1)$. In the limit of large times, \eqref{exponent} converges again to a constant:
	\begin{equation}\label{limit}
		\begin{split}
			\qquad \mathrm{Re}(\sqrt{-\xi_\varepsilon^2})&=\sqrt{\frac{\sqrt{((\mu^2-1)t^2+\varepsilon^2)^2+4\varepsilon^2 t^2}+((\mu^2-1)t^2+\varepsilon^2)}{2}}\\
			&\stackrel{t\to\infty}\longrightarrow \frac{\varepsilon}{\sqrt{1-\mu^2}}.
		\end{split}
	\end{equation}
	Therefore, the exponential factor is not decaying to zero along these directions either.
\end{itemize}

\noindent Despite these negative assessments, we observe that the limit quantity in \eqref{limit} \textit{does} converge to infinity in the limit $\mu\nearrow 1$. The idea is then to consider a region which intersects all the lines $r=\mu t,\  \mu\in (0,1)$ as in point (2) but do not contain any of them entirely.	
With this in mind, let us choose $\lambda\in (1/2,1)$ and define the sets
\begin{equation*}
\begin{split}
	C_\lambda^0&:=\{\xi\in\R^{1,3}_0\:|\: t\ge 1 \mbox{ and } 0\le r\le \sqrt{t^2-t^{2\lambda}}\}\\
C_\lambda^{1,+}&:=\{\xi\in \R^{1,3}_0\:|\: 0\le r\le t\ \mbox{ if $t\in (0,1)$ or } \sqrt{t^2-t^{2\lambda}}\le r\le t \mbox{ if $1 \le t$}  \}\\
C_\lambda^{1,-}&:=\{\xi\in \R^{1,3}_0\:|\:t\le r\le \lambda^{-1}\, t\}\\
 C_\lambda^{2}&:=\{\xi\in \R^{1,3}_0\:|\: r\ge \lambda^{-1}\, t\}\\
C^1_\lambda&:= C_\lambda^{1,+}\cup C_\lambda^{1,-}.
\end{split}
\end{equation*}
Let us analyze the upper boundary of $C_\lambda^1$ for large times $t$:
\begin{equation*}
\begin{split}
&r(t)=\sqrt{t^2-t^{2\lambda}}=t\sqrt{1-t^{2\lambda-2}}\sim t\bigg(1-\frac{1}{2}\,t^{2\lambda-2}\bigg)\Rightarrow  t-r\sim \frac{1}{2}\,t^{2\lambda-1}
\end{split}
\end{equation*}
Since $2\lambda>1$, the curves $r=\sqrt{t^2-t^{2\lambda}}$ and $r=t$ go far apart as time increases, i.e.
$$
\lim_{t\to \infty} |r(t)-t|=\infty.
$$
In the set $C_\lambda$ the exponent \eqref{exponent} in \eqref{firstestimate} is decaying. More precisely, we have the following estimates.
\begin{Lemma}\label{propoexpdecay}
	For every $\lambda\in (1/2,1)$ there is  $C>0$ such that, 
	\begin{equation}\label{decayingUalpha}
	\begin{split}
	\mathrm{(i)}&\ \mathrm{Re}(\sqrt{-\xi_\varepsilon^2})\ge  C\,t^{1-\lambda}\quad\qquad \mbox{for all $\ \xi\in C_\lambda^1\ $ with $\ t\ge 1$}.\\[0.2em]
	\mathrm{(ii)}&\ \mathrm{Re}(\sqrt{-\xi_\varepsilon^2})\ge C\,\big(r+\sqrt{t}\big)\quad \mbox{for all $\ \xi\in C_\lambda^2$}.
	\end{split}
	\end{equation}
	In particular, there are constants $A,k>0$ such that, for all $|\xi|_{\R^4}\ge R$,
	\begin{equation}\label{estimateparticualr}
		\begin{split}
	\mathrm{(i)}&\ |K_n(m\sqrt{-\xi_\varepsilon^2})|\le A\,\dfrac{e^{-k\,t^{1-\lambda}}}{\sqrt[4]{|\xi_\varepsilon^2|}}\quad \quad \mbox{if $\xi\in C_\lambda^1$ with $\ t\ge 1$}\\[0.2em]
	 \mathrm{(ii)}&\ |K_n(m\sqrt{-\xi_\varepsilon^2})|\le A\,\frac{e^{-k\,(r+\sqrt{t})}}{\sqrt[4]{|\xi_\varepsilon^2|}}\quad \mbox{if $\xi\in C_\lambda^2$}.
	\end{split}
	\end{equation}
\end{Lemma}
\begin{proof}
(i) Let us first focus on $C_\lambda^{1,+}\cap \{t\ge 1\}$. To start with, note that
$$
\mbox{For all $a>0$ there exists $b>0$ such that $\sqrt{1+x}-1\ge b\,x\ $  for all $0<x<a$.} 
$$ 
We now note that expression \eqref{exponent} is monotone increasing in the variable $r$, if $t$ is kept constant.  Therefore, from the definition of $C_{\lambda}^{1,+}$ and using that $t^{2\lambda}-\varepsilon^2> 1-1/4=3/4>0$ and the fact that, on $C_\lambda^{1,+}\cap \{t\ge 1\}$,
$$
\frac{2\varepsilon t}{t^{2\lambda}-\varepsilon^2}\le \frac{2\varepsilon t}{t-\varepsilon^2}=\frac{2\varepsilon}{1-\varepsilon^2/t}\le \frac{2\varepsilon}{1-\varepsilon^2},
$$ 
we infer the existence of some $b>0$ such that, for all $\xi\in C_\lambda^{1,+}\cap \{t\ge 1\}$,
	\begin{equation*}
	\begin{split}
	\mathrm{Re}(\sqrt{-\xi_\varepsilon^2})&\ge \sqrt{\frac{\sqrt{(- t^{2\lambda}+\varepsilon^2)^2+4\varepsilon^2 t^2}+(-t^{2\lambda}+\varepsilon^2)}{2}}=\\
	&=\sqrt{\frac{(t^{2\lambda}-\varepsilon^2)\left[\sqrt{1+\frac{4\varepsilon^2 t^2}{(t^{2\lambda}-\varepsilon^2)^2}}-1\right]}{2}}\ge \sqrt{\frac{b}{2}}\frac{2\varepsilon t}{\sqrt{t^{2\lambda}-\varepsilon^2}}.
	\end{split}
	\end{equation*}
	Inequality \eqref{decayingUalpha}-$\mathrm{(i)}$ follows immediately. To conclude, let us now consider $\xi\in C_\lambda^{1,-}\cap\{t\ge 1\}$. Using that $t\le r$ and  $0<1-\lambda< 1/2$, identity \eqref{exponent} implies that
	$$
	\mathrm{Re}(\sqrt{-\xi_\varepsilon^2})\ge \sqrt{\varepsilon t}\ge \sqrt{\varepsilon}\,t^{1-\lambda}.
	$$
	(ii) Let us now consider  $C_\lambda^2$. From \eqref{exponent}, we can infer that
	$$
	\mathrm{Re}(\sqrt{-\xi_\varepsilon^2}) \ge 2^{-1/2}\sqrt{2\varepsilon t+(1-\lambda^2)r^2} \ge C(r+\sqrt{t}).
	$$
	The final estimates \eqref{estimateparticualr} follow from \eqref{decayingUalpha} and \eqref{firstestimate}. 
\end{proof}

We now prove the integrability of the fermionic projector. We split the proof into two parts, in accordance with the two different estimates of Lemma \ref{propoexpdecay}.
%

	\begin{Lemma}\label{lemmaintegrability1}
		There is $\lambda\in (0,1)$ such that $P^\varepsilon(0,\,\cdot\,)\in L^4(C_\lambda^0\cup C_\lambda^1)$
	\end{Lemma}
	\begin{proof}
	We split the analysis into two separate calculations, in which the two terms adding up in \eqref{estimateP} are treated separately. In the proof we focus on the region $|\xi|_{\R^4}\ge R$, on which estimates \eqref{estimateparticualr} can be used. The complementary region $|\xi|_{\R^4}\le R$ is compact and hence integrability thereon  follows direcly from the continuity of $P^\varepsilon$.\\[0.3cm]
	 \textit{Integrability of the second term}. \\[-0.4cm]
	\begin{itemize}[leftmargin=2em]
		\item[(i)] \textit{$C_\lambda^0\cap\{|\xi|_{\R^4}\ge R\}$:} In this set the exponential in the estimate \eqref{firstestimate}  does not contribute. Nevertheless, this is not necessary to ensure integrability, as we now explain. For some constant $B$, \eqref{firstestimate} gives (note that $\mathrm{Re} \sqrt{-\xi_\varepsilon^2}>0$)
		\begin{equation}\label{integrand2}
			\begin{split}
				\left|\frac{K_1\big(m\sqrt{-\xi_\varepsilon^2}\big)}{\sqrt{-\xi_\varepsilon^2}}\right|^4\le  \frac{B}{|\xi_\varepsilon^2|^3}= \frac{B}{\big((-t^2+r^2+\varepsilon^2)^2+4\varepsilon^2 t^2\big)^{3/2}}
			\end{split}
		\end{equation} 
		From $\varepsilon< 1/2$, we obtain
		$
		t^2-r^2-\varepsilon^2\ge t^{2\lambda}-\varepsilon^2>3/4>0
		$
		for all $\xi\in C_\lambda^{0}\cap \{|\xi|_{\R^4}\ge R\}$.
		The integral of the left-hand side of \eqref{integrand2} on this region can then be estimated by
		\begin{equation}\label{integrals1}
			\begin{split}
				&\le C\int_{1}^\infty\, dt \int_{0}^{\sqrt{t^2-t^{2\lambda}}}\frac{r^2}{\big((t^2-\varepsilon^2-r^2)^2+4\varepsilon^2 t^2\big)^{3/2}}\,dr\le \\
				&\le C  \int_{1}^\infty\, dt \int_{0}^{\sqrt{t^2-t^{2\lambda}}}\frac{r^2}{(t^{2\lambda}-\varepsilon^2)^3}\,dr\le D\int_{1}^\infty\, dt \left[\frac{t}{t^{2\lambda}-\varepsilon^2}\right]^3,
			\end{split}
		\end{equation}
		for suitable constants $C,D>0$.
		Assuming $\lambda>2/3$, the integral is finite.  \\[0.1em]
		\item[(ii)]
		\textit{$C_\lambda^1\cap \{|\xi|_{\R^4}\ge R\}$:} Choosing $R$ large enough  (in fact uniformly in $\lambda > 2/3$), we can assume that $t\ge 1$. In this region the exponential decay \eqref{estimateparticualr}-(i) holds. Therefore, for some positive $k,B$,
		\begin{equation}\label{integrand}
		\begin{split}
		\left|\frac{K_1\big(m\sqrt{-\xi_\varepsilon^2}\big)}{\sqrt{-\xi_\varepsilon^2}}\right|^4\le B \, \frac{e^{-k\, t^{1-\lambda}}}{|\xi_\varepsilon^2|^3}=B\, \frac{e^{-k\, t^{1-\lambda}}}{\big((-t^2+r^2+\varepsilon^2)^2+4\varepsilon^2 t^2\big)^{3/2}}.
		\end{split}
		\end{equation}
		The integral  of the left-hand side of \eqref{integrand}  on $C_\lambda^1\cap \{|\xi|_{\R^4}\ge R\}$  can then be estimated by
		\begin{equation*}\label{estimateA}
		\begin{split}
		&\le C\int_{1}^\infty dt\, e^{-k\, t^{1-\lambda}}\int_{\sqrt{t^2-t^{2\lambda}}}^{t/\lambda}\frac{r^2}{\big((-t^2+r^2+\varepsilon^2)^2+4\varepsilon^2 t^2\big)^{3/2}}\,dr\le \\
		&\le D\int_{1}^\infty  dt\, e^{-k\, t^{1-\lambda}}\int_{0}^{t/\lambda}\frac{r^2}{ t^3}\,dr\le K\int_{1}^\infty dt\, e^{-k\, t^{1-\lambda}}<\infty,
		\end{split}
		\end{equation*}
		for suitable constants $C,D,K>0$. 

	\end{itemize}
	\vspace{0.2cm}

	\noindent \textit{Integrability of the first term}\\[0.2cm]
	We now study the first term in \ref{estimateP}. Using that $|\gamma^j|_2=1$ for any $j=0,1,2,3$, we get, for some constant $C>0$,
	\begin{equation}\label{K2estimate}
	\begin{split}
	\left|\frac{K_2\big(m\sqrt{-\xi_\varepsilon^2}\big)}{-\xi_\varepsilon^2}\right|^4|\slashed{\xi}_\varepsilon|_2^4&\le C\left|\frac{K_2\big(m\sqrt{-\xi_\varepsilon^2}\big)}{\sqrt{-\xi_\varepsilon^2}}\right|^4\frac{(t^2+r^2+\varepsilon^2)^2}{|\sqrt{-\xi_\varepsilon^2}|^4}.
	\end{split}
	\end{equation}
	In point (ii) above we chose $R$ large enough to ensure that $t\ge 1$ on $C_\lambda^1\cap \{|\xi|_{\R^4}\ge R\}$ for any $\lambda>2/3$. On the set $(C_\lambda^0\cup  C_\lambda^{1})\cap\{|\xi|_{\R^4}\ge R\}$ one then has $r\le \lambda^{-1}t$ and $\varepsilon<1\le t$ and therefore the factor on the right-hand side of \eqref{K2estimate} can be bounded from above by
\begin{equation}\label{estimatet4}
	\frac{(t^2+r^2+\varepsilon^2)^2}{|\sqrt{-\xi_\varepsilon^2}|^4}\le A\,\frac{t^4}{|\sqrt{-\xi_\varepsilon^2}|^4},\quad\mbox{for some constant } A>0.
	\end{equation}
	Moreover, the Bessel functions $K_2$ and $K_1$ have the same asymptotic behavior at infinity, as stated in \eqref{asymptoticexpK}.
	 Therefore, to analyze the integrability of the first term on $(C_\lambda^0\cup  C_\lambda^{1})\cap\{|\xi|_{\R^4}\ge R\}$ we simply need to multiply the estimates carried out in points  (i) and (ii) above by the factor on the right-hand side of \eqref{estimatet4}.
	Adjusting the corresponding computations, we obtain:\\[-0.8em]
	\begin{itemize}[leftmargin=2em]
		\item[(i)] $C_\lambda^0\cap\{|\xi|_{\R^4}\ge R\}:$ From \eqref{integrand2}, the integrals \eqref{integrals1} become, for suitable $A,B>0$,
	\begin{equation*}
	\begin{split} 
	&\le A\int_{1}^\infty dt\int_0^{\sqrt{t^2-t^{2\lambda}}}\frac{t^4\,r^2}{(t^{2\lambda}-\varepsilon^2)^5}\,dr\le  B\int_{1}^\infty\, dt\, \frac{t^7}{(t^{2\lambda}-\varepsilon^2)^5}.
	\end{split}
	\end{equation*}
	Assuming  $\lambda>4/5>2/3$ the integral is finite.	\\[-0.6em]
	\item[(ii)] $C_\lambda^1\cap\{|\xi|_{\R^4}\ge R\}:$ From \eqref{integrand}, the integrals \eqref{estimatet4} become, for suitable $A,B>0$, 
	\begin{equation*}
		\begin{split}
			&\le A\int_{1}^\infty dt\, e^{-k\, t^{1-\alpha}}\int_0^{t/\lambda}\frac{r^2}{t}\,dr\le  B\int_{1}^\infty dt\,t^2\, e^{-k\, t^{1-\alpha}}<\infty.
		\end{split}
	\end{equation*}
		\end{itemize}
	The proof is complete.
\end{proof}

To conclude, we prove integrability on the complementary region.
\begin{Lemma}
	For every $\lambda\in (0,1)$, $P^\varepsilon(0,\,\cdot\,)\in L^4(C_\lambda^2)$
\end{Lemma}
\begin{proof}
Once again, let us study the two terms in \eqref{estimateP} separately. From \eqref{estimateparticualr}-(ii), we obtain, for  $|\xi|_{\R^4}\ge R$,
\begin{equation}\label{integrandA}
\begin{split}
\left|\frac{K_1\big(m\sqrt{-\xi_\varepsilon^2}\big)}{\sqrt{-\xi_\varepsilon^2}}\right|^4&\le B \, \frac{e^{-k\,(r+\sqrt{t})}}{|\xi_\varepsilon^2|^3}=B\, \frac{e^{-k\, (r+\sqrt{t})}}{\big((-t^2+r^2+\varepsilon^2)^2+4\varepsilon^2 t^2\big)^{3/2}}\\
&\le B\, \frac{e^{-k\, (r+\sqrt{t})}}{((1-\lambda^2)r^2+\varepsilon^2)^3}\le B\,\varepsilon^{-6}\, e^{-k\sqrt{t}}e^{-kr},
\end{split}
\end{equation}
Similarly, using \eqref{K2estimate} and \eqref{integrandA}, the second term in \eqref{estimateP} gives, for  $|\xi|_{\R^4}\ge R$,
\begin{equation*}
\begin{split}
\left|\frac{K_2\big(m\sqrt{-\xi_\varepsilon^2}\big)}{-\xi_\varepsilon^2}\right|^4|\slashed{\xi}_\varepsilon|_2^4&\le A\,\varepsilon^{-10}\, e^{-k\sqrt{t}}\,e^{-kr}\,(r^2+t^2+\varepsilon^2)^2
\end{split}
\end{equation*}
The exponential decays ensures integrability, concluding the proof.
\end{proof}

\section{Basics on Continuity of the Eigenvalues. }\label{appendixconteigen}
In this appendix we review some basic results on the continuity of the relation between eigenvalues and corresponding operators.

Let $\H$ be an infinite-dimensional separable Hilbert space and let $A\in\mathfrak{B}(\H)$ be a compact self-adjoint operator. For such operators,
$$
\sigma(A)=\sigma_e(A)\cup \sigma_d(A),\quad\mbox{with $\sigma_e(A)=\{0\}$},
$$
where $\sigma_d(A)$ is the \textit{discrete spectrum} (in this case the set of non-zero eigenvalues of $A$) and $\sigma_e(A)$ is the \textit{essential spectrum} (see for example \cite[Section 9.2]{BS}). The eigenspaces corresponding to the elements of $\sigma_d(A)$ have finite dimension. Now, let $n_\pm\in \N_0\cup\{\infty\}$ denote the number of strictly positive and strictly negative eigenvalues, respectively. Assume first that $n_+ =\infty$, then we arrange the strictly positive eigenvalues into a sequence
$$
\{\lambda^+_n(A)\}_{n\in\N}\subset \R_+\quad\mbox{with }\  \lambda_n^+(A)\ge \lambda_{n+1}^+(A)\ \mbox{ for all $n\in\N$}.
$$ 
If, on the other hand, there is only a finite number $0<n_+<\infty$ of them, we define an analogous finite family for $n\in\{1,\dots,n_+\}$ and complete it to  $\lambda^\pm_{n}(A)=0$ for any $n> n_+$. If $n_+=0$, we define $\lambda^+_n(A)=0$ for every $n\in \N$. 
Analogously, we construct
$$
\{\lambda^-_n(A)\}_{n\in\N}\subset \R_-\quad\mbox{with }\  -\lambda_n^-(A)\ge -\lambda_{n+1}^-(A)\ \mbox{ for all $n\in\N$}.
$$
From the properties of compact operators, it follows that $\lambda_n^\pm(A)\to 0$ as $n\to \infty$.
With these conventions, the following identity holds for any $n\in\N$, which is known as the \textit{Courant-Fischer (or Min-Max) Principle} (see for example \cite[Section 9.2]{BS}).
	\begin{equation*}
	\lambda^\pm_n(A)= \min\bigg\{\sup_{u\in \mathbb{S}\cap M^\perp}\la Au|u\ra\:\bigg|\: M\subset\H,\ \dim M=n-1\bigg\},
	\end{equation*}
where $\mathbb{S}:=\{u\in\H\:|\: \|u\|=1\}$. Exploiting this identity it is possible to show that the eigenvalues, orderded according to the conventions above, depend continuously on the operators in the operator norm (see \cite[Eq. (9.2.19)]{BS}).
\begin{Prp}\label{propcontiself}
	Let $S,T$ be compact self-adjoint operators on $\scH$. Then,  
	$$
	|\lambda^\pm_n(S)-\lambda^\pm_n(T)|\le \|S-T\|\quad\mbox{for every $n\in \N$}.
	$$
\end{Prp}
Given a  general (not necessarily self-adjoint) compact operator $T$, one can apply the above theorem to its absolute value $A:=|T|$.  
The corresponding eigenvalues, ordered as above, form the so-called \textit{singular values} of $T$ (see \cite[Section 11.1]{BS}):
$$
s_n(T):=\lambda^+_n(|T|)\quad\mbox{for all $n\in\N$}.
$$
By construction, $s_n(T)\to 0$ as $n\to \infty$. 
Applying the Min-Max Principle to $|T|$, a result similar to Proposition \ref{propcontiself} can be proven for the singular points (see \cite[(11.1.15)]{BS}).
\begin{Prp}\label{propconti}
	Let $S,T$ be compact operators on $\scH$. Then,  for every $n\in \N$,
	$$
	|s_n(S)-s_n(T)|\le \|S-T\|.
	$$
\end{Prp}
This is an important result, for it shows that the singular values are continuous with respect to the sup-norm topology. Continuity results for the eigenvalues of general compact operators can also be proved, although not in the strong form as in Proposition \ref{propcontiself}.

Let $T$ be a compact operator. If the non-zero eigenvalues of $T$, repeated according to their algebraic multiplicity, are infinitely-many, we enumerate them in an arbitary sequence $\{\nu_n(T)\}_{n\in\N}$ (without any order prescription).
If there is only a finite number $N>0$ of them, we define an analogous family for $n\in\{1,\dots,N\}$ and complete it to a countable sequence by setting $\nu_n(T)=0$ for any $n>N$. Finally, if $T=0$, we simply define $\nu_n(T)=0$ for every $n\in\N$. Such sequences are called \textit{enumerations of the eigenvalues of $T$} and denoted simply by $\nu(T)$. With this conventions, the following result holds (see \cite[Lemma 5, Ch.XI.9.5]{DS}).

\begin{Thm}\label{continuityenum}
	Suppose $T_m$ is a sequence of compact operators converging to $T$ in the sup-norm topology. Let $\nu(T)$ be an enumeration of the eigenvalues of $T$. Then there exist enumerations $\nu(T_m)$ of the  eigenvalues of the operators $T_m$ such that,
	$$
	\lim_{m\to \infty}\nu_n(T_m)=\nu_n(T)\quad\mbox{ for every $n\in\N$.}
	$$
\end{Thm}
This result is weaker than Proposition \ref{propconti}, for  it is not explicitly stated how the enumerations $\nu(T_m)$ depend on the original enumeration $\nu(T)$. In any case, this result is useful in the case of finite-rank operators and, in particular, in proving the continuity of the Lagrangian (see Proposition \ref{proplagrangian}).

\section{Miscellaneous Proofs}\label{appendixproofs}
This appendix is devoted to the proof of Proposition \ref{openreg}, Lemma \ref{lemmaregcont},  Lemma \ref{regularization}, Theorem \ref{localrepr}, Proposition \ref{propF} and Proposition \ref{propfirstderiv}.

\begin{proof}[Proof of Proposition \ref{openreg}]
	Let us start with point (i). Assume that
	$
	\mathrm{sign}(\x_0)=(p_0,q_0).
	$
	If $p_0=q_0=0$, then there is nothing to do. So, assume that at least one of $p_0,q_0$ is finite. Let $u^{-}_i\in S^-_{\x_0}$ with $i=1,\dots,p_0$ and $u^{+}_i\in S^+_{\x_0}$ with $i=1,\dots,q_0$ be orthonormal  bases (with respect to the Hilbert space structure) of $S_{x_0}^\pm$ made of eigenvectors of $\x_0$, i.e. 
	$$
	x_0\,u^{\pm}_i=\lambda^{\pm}_i\,u^{\pm}_i.
	$$
	The vectors $\x_0\,u^\pm_i$ are different from zero and are orthogonal to each other in both the Hilbert scalar product and the spin scalar product. 
	Next, let us define for any $\mu=1,\dots,p_0+q_0$ the functions (cf. \eqref{posneg})
	\begin{equation}\label{fmu}
	f_\mu:\F\ni x\mapsto f_\mu(\x):=\begin{cases}
	\dfrac{\x\,u^{-}_\mu-|\x|u^{-}_\mu}{2} \in S_\x^- & \mbox{if }1\le \mu\le p_0\\[0.6em]
	\dfrac{\x\,u^{+}_{\mu-p_0}+|\x|u^{+}_{\mu-p_0}}{2}\in S_\x^+ & \mbox{if }p_0< \mu\le p_0+q_0
	\end{cases}.
	\end{equation}
Note that the function $\x\mapsto |\x|$ is continuous in the operator norm: this can be shown both using the functional calculus or exploiting the general estimate
	\begin{equation*}\label{kato}
	\||A|-|B|\|\le \left[\frac{4}{\pi}+\frac{2}{\pi}\log\frac{\|A\|+\|B\|}{\|A-B\|}\right]\|A-B\|,
	\end{equation*}
	which holds for any couple of different self-adjoint operators (see \cite{kato}).
	The functions $f_\mu$ in \eqref{fmu} are then continuous. Thus, for any $\varepsilon>0$ there is $r>0$ such that
	\begin{equation*}
	\begin{split}
	\|f_\mu(\x)-f_\mu(\x_0)\|&<\varepsilon\quad\mbox{for all $\x\in B_r(\x_0)$.}
	\end{split}
	\end{equation*}
	Choosing $\varepsilon$ small enough so that
	$$
\varepsilon <
	\frac{\inf\{|\lambda^{\pm}_\mu|\:|\: \mu=1,\dots p_0+q_0\}}{2n},
	$$
	and using the fact that the vectors $\{f_\mu(\x_0)\,\mu=1,\dots,p_0+q_0\}$ are orthogonal and that $\mathrm{dim}\,S_\x\le 2n$, it follows that (see for example \cite[Lemma 5.2]{oppio}), for every $\x\in B_r(\x_0)$, 
	$$
	\{f_\mu(\x)\:|\: \mu=1,\dots,p_0+q_0\}\subset S_\x\quad\mbox{are linearly independent.}
	$$
	 By construction, they respect the decomposition \eqref{decompsign} of $\x$. Putting all together, this means that $n_-(\x)\ge p_0$ and $n_+(\x)\ge q_0$ and the claim follows.
	The proof of (i) is concluded. Let us not prove point (ii). 
	Let $\x_0\in\F^{\mathrm{reg}}$. From point (i) we infer that there exists an open neighborhood $U_0\subset\F$ where the local signature is $(n,n)$, and therefore made of regular points. Therefore $U_0\subset\F^{\mathrm{reg}}$. This proves that $\F^{\mathrm{reg}}$ is open in $\F$. To prove denseness, let $x\in \F$ and $U$ be any open neighborhood. If $x\in\F^{\mathrm{reg}}$ then there is nothing to do. Otherwise, by choosing $k:=2n-\dim\,S_\x$ normalized vectors $\{e_i\}_{i}\subset S_\x^\perp$ one can  construct a regular perturbation of $\x$, by defining
	\begin{equation}\label{pertx}
	\x(\varepsilon):=\x+\varepsilon\sum_{i=1}^ks_i\, \langle e_i,\,\cdot\,\rangle e_i\in\F^{\reg}
	\end{equation}
	where $\varepsilon>0$ is arbitrary and $s_i=\pm$ are to be chosen depending on $\mathrm{sign}(\x)$. Now, note that 
	$
	\|\x(\varepsilon)-\x\|\le \varepsilon.
	$
	Choosing $\varepsilon$ so small that $B_{2\varepsilon}(\x)\subset U$, the claim follows.
\end{proof}

\begin{proof}[Proof of Lemma \ref{lemmaregcont}]
	First, let us show that $\mathrm{g}$ is not continuous at any point $\x\in\F\setminus\F^\reg$. Because $\x$ is not regular, we can modify it in the orthogonal of $S_\x$ to $\x(\varepsilon)\in \F$ as in \eqref{pertx}. For this proof, it suffices to consider $k=1$, the general case being analogous. By construction, we see that, for $\varepsilon$ sufficiently small,
	$$
	\|\mathrm{g}(\x(\varepsilon))\|=\left\|\mathrm{g}(\x)+\frac{1}{s_1\varepsilon}\langle e_1,\,\cdot\,\rangle e_1\right\|=\frac{1}{\varepsilon}.
	$$
	This shows that $\mathrm{g}$ cannot be continuous, because $\|\mathrm{g}(\x(\varepsilon))\|\to \infty$, whereas $\x(\varepsilon)\to \x$. 
	
	We now study $\mathrm{g}$ on $\F^\reg$. In the following computations, we adapt some results from \cite{Wedin}, in particular Theorems 2.1 and 4.1. As a first step, note that, for any $\x,\y\in\F^\reg$, 
		 \begin{equation*}\label{identitiesg}
		 \begin{split}
		 \textrm{(1)}&\quad \mathrm{g}(\x)-\mathrm{g}(\y)= \mathrm{g}(\x)(\y-\x)\mathrm{g}(\y)+\mathrm{g}(\x)\pi_\x(\bI-\pi_\y)-(\bI-\pi_\x)\pi_\y\,\mathrm{g}(\y)\\
		 \textrm{(2)}&\quad (\bI-\pi_\x)\pi_\y=(\bI-\pi_\x)(\y-\x)\mathrm{g}(\y),\quad \pi_\x(\bI-\pi_\y)=\mathrm{g}(\x)(\x-\y)(\bI-\pi_\y).
		 \end{split}
		 \end{equation*}
		 Identities (2) follow immediately from $\z\,\mathrm{g}(\z)=\pi_\z$. Identity (1) follows by multiplying out the identity $\mathrm{g}(\x)-\mathrm{g}(\y)=(\pi_\x+(\bI-\pi_x))(\mathrm{g}(\x)-\mathrm{g}(\y))(\pi_\y+(\bI-\pi_\y))$. 
		 
		\noindent Consider the subspace $\K:=S_\x+S_\y$. The spin spaces $S_\x,S_\y$ are both subspaces of $\K$ of dimension $2n$. Working on the finite dimensional space $\K$, a standard result for matrices (see \cite[Theorem 7.1]{Wedin} or \cite[Theorem 2.3]{Stw}) yields
		 $$
		\mathrm{(3)}\quad \|(\bI-\pi_\x)\pi_\y\|=\| (\bI-\pi_\y)\pi_\x\|\,(=\| \pi_\x(\bI-\pi_\y)\|)\quad 
		 $$
		 Using identities  (1),(2),(3) above, we immediately obtain that
		 \begin{equation}\label{estimateg}
		 \|\mathrm{g}(\x)-\mathrm{g}(\y)\|\le 3\|\mathrm{g}(\x)\|\|\mathrm{g}(\y)\|\|\x-\y\|.
		 \end{equation}
		To conclude, we now show how this estimate implies local boundedness of $\mathrm{g}$. Let $\x\in\F^\reg$ be chosen and let $r>0$ be so small that (remember that $\F^\reg$ is open)
		$$
		B_r(\x)\subset\F^\reg\quad\mbox{and}\quad 3\|\mathrm{g}(\x)\|r<\frac{1}{2}.
		$$ 
		From this, \eqref{estimateg} and the reverse triangular inequality, we then obtain
		\begin{equation*}
		\begin{split}
		\frac{1}{2}\|\mathrm{g}(\y)\|\le \|g(\x)\|+\|g(\y)\|\big(3\|g(\x)\|\|\x-\y\|-1/2\big)\ \mbox{ for all $\y\in B_r(\x)$}.
		\end{split}
		\end{equation*}
	Putting all together we get the claim.
\end{proof}

\begin{proof}[Proof of Lemma \ref{regularization}]
	We here provide a proof which applies also in presence of a regular static electromagnetic field. From the theory of semigroups of operators (see for example \cite[Section Remarks 10.20-(2)]{moretti-book}), it follows that
	$$
	\psi_\varepsilon:=e^{-\varepsilon|H|}\psi\in \mathfrak{D}(|H|^{m})\quad\mbox{for all $m\in\N$}.
	$$ 
	On the other hand, 
	$$
	\mathfrak{D}(|H|^{2})=\mathfrak{D}(H^{2})=\mathfrak{D}(\overline{\Delta})=\{\varphi\in L^2(\R^3,\C^4)\:|\: (1+|\V{k}|^2)\,\hat{\varphi}\in L^2(\R^3,\C^4)\}.
	$$
	As a consequence, for all $m\in\mathbb{N}$,
	\begin{equation*}\label{riemannlebesgue}
	\psi_\varepsilon\in \mathfrak{D}(\overline{\Delta}^{2m})=\{\varphi\in L^2(\R^3)\:|\: (1+|\V{k}|^2)^{2m}\,\hat{\varphi}\in L^2(\R^3,\C^4) \}=W^{4m,2}(\R^3,\C^4).
	\end{equation*}
	The Sobolev embedding theorems ensure that $\psi_\varepsilon\in C^\infty(\R^3)$. 
\end{proof}


\begin{proof}[Proof of Theorem \ref{localrepr}]
	For simplicity of notation we here denote $\H_m^-$ by $\H$. Choose $x_0\in \R^4$. Referring to Proposition \ref{existenceframe}, let $(f_1,f_2,f_3,f_4,\Omega)$ be a local spin frame of $\mathrm{F}$ defined around $x_0$. 
	Referring to the corresponding spin scalar products, we introduce the isometries
	\begin{equation*}
		V_x:=S_{\mathrm{F}(x)}\rightarrow \C^4,\quad V_x(f_\mu(x)):=\mathfrak{e}_\mu\quad\mbox{for $x\in\Omega$}.
	\end{equation*}
	Then,  we define
	\begin{equation}\label{defpsi}
	\Psi:\Omega\rightarrow\mathfrak{B}(\H,\C^4),\quad \Psi(x)u :=V_x\,(\pi_{\mathrm{F}(x)}u).
	\end{equation}
	Referring to the corresponding adjoints in the spaces $\C^4,S_{\mathrm{F}(x)}$ and $\H$, we have, for arbitrary $u,v\in\scH$, 
	\begin{equation*}
	\begin{split}
	\langle u| \Psi(x)^*\Psi(x)v\rangle&=\Sl \Psi(x)u|\Psi(x)v\Sr=\Sl V_x(\pi_{\mathrm{F}(x)} u)\,|\,V_x(\pi_{\mathrm{F}(x)} v)\Sr=\\
	&=\Sl \pi_{\mathrm{F}(x)}u|\pi_{\mathrm{F}(x)}v\Sr_{\mathrm{F}(x)}=-\langle u|\mathrm{F}(x)v\rangle,
	\end{split}
	\end{equation*}
	where the last identity follows by definition of spin scalar product \eqref{ssp}. 
	The arbitrariness of $u$ and $v$ implies that 
	\begin{equation*}\label{finalform}
	\mathrm{F}(x)=-\Psi(x)^*\Psi(x)\quad \mbox{for every } x\in \R^4.
	\end{equation*} 
	The surjectivity of $\Psi(x)$ follows directly its definition \eqref{defpsi} and the fact that the operator $\mathrm{F}(x)$ has rank four (it being regular).
	
	Next, we need to show that $\Psi\in \mathcal{E}(\Omega,\H,\C^4)$. Consider the same local spin frame $(f_1,f_2,f_3,f_4,\Omega)$ introduced above, then, for any $u\in\scH$ and $x\in\Omega$,
	\begin{equation*}
	\begin{split}
	\Sl \mathfrak{e}_\mu|\Psi(x)u\Sr&=\Sl \mathfrak{e}_\mu|V_x(\pi_{\mathrm{F}(x)}u)\Sr=\Sl f_\mu(x)|\pi_{\mathrm{F}(x)}u\Sr_{\mathrm{F}(x)}=-\langle f_\mu(x)|\mathrm{F}(x)u\rangle
	\end{split}
	\end{equation*}
	Using that $\{\mathfrak{e}_\mu\}_\mu$ is a Hilbert basis of $\C^4$, we have
	\begin{equation}\label{expansionPhi}
	\begin{split}
	\Psi(x)u=\sum_{\mu=1}^4s_\mu\Sl \mathfrak{e}_\mu|\Psi(x)u\Sr\,\mathfrak{e}_\mu=-\sum_{\mu=1}^4 s_\mu\langle f_\mu(x)|\mathrm{F}(x)u\rangle.
	\end{split}
	\end{equation}
At this point, using that the functions $f_\mu$ and $\mathrm{F}$ are continuous, by restricting to a relatively compact open neighborhood $\Omega_1\subset\Omega$ of $x_0$, we can assume that, for some $B>0$,
\begin{equation}\label{boundedlocal}
\|f_\mu(x)\|+\|\mathrm{F}(x)\| \le B\quad\mbox{for all }x\in\Omega_1.
\end{equation}
Using this in \eqref{expansionPhi}, we see that there is a constant $C>0$ such that  $|\Psi(x)u|\le C\|u\|$ for all $u\in\H$ and all $x\in\Omega_1$. This gives 
$$
\sup_{x\in\Omega_1}\|\Psi(x)\|_{\mathfrak{B}(\H,\C^4)}\le C<\infty.
$$
To conclude, choose arbitrary $x,y\in\Omega_1$ and $u\in\H$. Then, from \eqref{expansionPhi}, the triangular inequality and \eqref{boundedlocal}, we obtain
\begin{equation*}
\begin{split}
|\Psi(x)u-\Psi(y)u|&\le \sum_{\mu=1}^4|\langle f_\mu(x)| \mathrm{F}(x)u-\mathrm{F}(y)u\rangle |+\sum_{\mu=1}^4|\langle f_\mu(y)-f_\mu(x)|\mathrm{F}(y)u\rangle |\\
&\le C\sum_{\mu=1}^4 \|\mathrm{F}(x)-\mathrm{F}(y)\|\|u\| +C\sum_{\mu=1}^4\| f_\mu(y)-f_\mu(x)\|\|u\|.
\end{split}
\end{equation*}
The continuity of $\mathrm{F}$ and $f_\mu$ implies that $\Psi\in C^0(\Omega,\mathfrak{B}(\H,\C^4))$, concluding the proof of the first part of the theorem.

The last statement of the theorem can be proved in the same way by noting that one now has the stronger condition
\begin{equation*}
\begin{split}
\|f_\mu(x)\|=\|\sqrt{|\mathrm{F}(x)|}^{-1}\,\hat{f}_\mu(x) \|\le \|\sqrt{|\mathrm{F}(x)|}^{-1}\|\le 1/\sqrt{k}\quad\mbox{for all $x\in\R^4$},
\end{split}
\end{equation*}
which follows from \eqref{defbasis}, condition  \eqref{assumptionbounded} and the spectral theorem.
\end{proof} 

\begin{proof}[Proof of Proposition \ref{propF}]
 For simplicity of notation we again drop  the indices $m,-$. By definition, the operator $\mathrm{F}(\lambda,x)$ is a well-defined bounded self-adjoint operator with rank no larger than four. From the surjectivity of $\Psi(\lambda,x)$,
	\begin{equation}\label{identityF}
	\mathrm{F}(\lambda,x)(\H)=\Psi(\lambda,x)^*\big(\Psi(\lambda,x)(\H)\big)= \Psi(\lambda,x)^*(\C^4).
	\end{equation}
	Now, surjectivity of $\Psi(\lambda,x)$ also imply injectivity of $\Psi(\lambda,x)^*$. This, together with \eqref{identityF} implies that $\mathrm{F}(\lambda,x)$ has rank four.
	
	To conclude, we need to show that $\mathrm{F}(\lambda,x)\in\F$.
	Because of self-adjointness, the operator $\mathrm{F}(\lambda,x)$ vanishes on the orthogonal of its image. Without loss of generality, let us assume by contradiction that $\mathrm{F}(\lambda,x)$ has (counting multiplicity) three strictly positive eigenvalues $\alpha_i>0$, $i=1,2,3$, with corresponding eigenvectors $u_i\in \mathrm{F}(\lambda,x)(\H)$. Because of self-adjointness, these eigenvectors can always be chosen to form an orthonormal set. Therefore, for $i,j=1,2,3$,
	\begin{equation*}
	\begin{split}
	\alpha_j\,\delta_{ij}&=\langle u_i|\mathrm{F}(\lambda,x) u_j\rangle=-\langle u_i|\Psi(\lambda,x)^*\Psi(\lambda,x) u_j\rangle=-\Sl \Psi(\lambda,x) u_i|\Psi(\lambda,x)u_j\Sr.
	\end{split}
	\end{equation*}
	The identity above implies that the vectors $\Psi(\lambda,x)u_i, i=1,2,3$ are orthogonal with respect to the spin scalar product of $\C^4$. In particular, $\Sl\cdot|\cdot\Sr$ has signature $(3,1)$ or $(4,0)$, which is a contradiction. 
	Now, let us  prove continuity. For $u\in\H$ and $s,t\in I\times \R^4$,
	\begin{equation*}
	\begin{split}
	|\langle u| \mathrm{F}(s)u-\mathrm{F}(t)u\rangle|& \le |\Sl \Psi(s)u|\Psi(s)u-\Psi(t)u\Sr|+|\Sl \Psi(s)u-\Psi(t)u|\Psi(t)u\Sr|\\
	&\le 2 \sup_{m\in I\times\R^4}\|\Psi(m)\|_{\mathfrak{B}(\H,\C^4)}\,\|\Psi(t)-\Psi(s)\|_{\mathfrak{B}(\H,\C^4)}\,\|u\|^2
	\end{split}
	\end{equation*}
	Because the operators $\mathrm{F}(s)$ are self-adjoint, it follows that
	$$
	\|\mathrm{F}(s)-\mathrm{F}(t)\|=\sup_{\|u\|=1}|\langle u| \mathrm{F}(s)u-\mathrm{F}(t)u\rangle|.
	$$
	The claim follows by putting the last two formulas together, and using the continuity and the boundedness assumptions on $\Psi$. The boundedness of $\mathrm{F}$ can be proved analogously.
\end{proof}

\begin{proof}[Proof of Proposition \ref{propfirstderiv}]
Before entering the proof we prove the following useful lemma, which applies with obvious changes also to vector-valued distributions. 
	\begin{Lemma}\label{lemmadistrib}
		Let $U\subset\R^m$ be open and $f\in C^\infty(\R^n\times U)$. For any $T\in \mathcal{E}'(\R^n)$ the function 
		\begin{equation}\label{functionf}
			U\ni k\mapsto  T(f(\,\cdot\,,k))\in\C
		\end{equation} 
		is smooth, with
		\begin{equation}\label{interchangederivativeaction}
			D^\beta\, T( f(\,\cdot\,,k)) = T( D_2^\beta f(\,\cdot\,,k)).
		\end{equation}
		Moreover, let $f$ have the property that for every $\alpha\in\N^n$ there is $g_\alpha\in L^1(U)$ such that, for all $x\in\R^4$,
		\begin{equation}\label{bounduniform}
			|D_1^\alpha f(x,\,\cdot\,)|\le |g_\alpha|\ \  \mbox{ almost everywhere on $U$}.
		\end{equation}
		Then, the following integral function is well-defined and smooth,
		\begin{equation}
			F(x):=\int_U f(x,k)\,d^mk,\quad x\in\R^{n}.
		\end{equation}
		Finally, \eqref{functionf} is integrable and
		$$
		T(F)=\int_{U} T(f(\,\cdot\,,k))\,d^mk.
		$$
	\end{Lemma}
	\begin{proof}
		Let $R$ be large enough so that $\supp T\subset B_R$ and let $\eta\in C_0^\infty(\R^n,\R^+)$ be a bump function which vanishes outside $B_{2R}$ and is identically equal to $1$ on $B_R$. Then, $T( \varphi) = T(\varphi\eta)$ for any $\varphi\in C^\infty(\R^n)$. Using \cite[Theorem 4.1.1]{fried} we conclude that the function $k\mapsto  T(f(\,\cdot\,,k),T)$ is smooth and satisfies \eqref{interchangederivativeaction}. 
		The fact that $F$ is well-defined and smooth follows from \eqref{bounduniform}, Lebesgue's dominated convergence theorem and the mean value theorem. 
		Now, using that $\mathcal{D}(\Omega)$ is dense in $\mathcal{D}'(\Omega)$ we can find a sequence $T_n\in \mathcal{D}(\R^n)$ such that $T_n\to T$ in $\mathcal{D}'(\R^n)$. Thus, 
		\begin{equation}\label{contdistrib}
			T_n(\varphi\eta)\to T(\varphi\eta)=  T(\varphi)\quad\mbox{for all $\varphi\in C^\infty(\R^n)$}.
		\end{equation}
		In particular, this applies to $F$ and $f(\,\cdot\,,k)$. The uniform boundedness principle for Fr\'echet spaces (in our case $\mathcal{D}'(\R^n)$) ensures the existence of $C>0$, $N\in\N$ such that
		\begin{equation}\label{estimatedistr}
			|T_n(\varphi)|\le C\sum_{|\alpha|\le N}\|D^\alpha\varphi\|_{\infty,B_{2R}}\quad\mbox{for all $\varphi\in \mathcal{D}(B_{2R})$ and $n\in\N$.}
		\end{equation}
		Hence, using \eqref{estimatedistr} and \eqref{bounduniform}, we obtain
		\begin{equation*}
			|T_n(f(\,\cdot\,,k)\eta)|\le C\sum_{|\alpha|\le N}a_\alpha|g_\alpha(k)|\quad\mbox{for all $n\in\N$},
		\end{equation*}
		and the function on the right-hand side is integrable on $U$.
		Thus, using \eqref{contdistrib} and Lebesgue's dominated convergence theorem, we conclude that $k\mapsto T(f(\,\cdot\,,k))$ is integrable and that
		\begin{equation}\label{epsilon}
			\begin{split}
				\int_U\,T_n( f(\,\cdot\,,k)\eta)\,d^mk\to \int_U T(f(\,\cdot\,,k))\,d^mk.
			\end{split}
		\end{equation}
		To conclude, note that $|f(x,k)\eta(x)T_n(x)|\le |g_0(k)||T_n(x)|$. Therefore, Fubini Theorem applies and yields
		\begin{equation}\label{finitelimit}
			\begin{split}
				\int_Ud^mk\, T_n( f(\,\cdot\,,k)\eta)&=\int_U d^mk\int_{\R^n}d^nx\, f(x,k)\eta(x)\,T_n(x)=\\
				&=\int_{R^n} d^nx\left(\int_{U}d^mk\, f(x,k)\right)\eta(x)T_n(x)= T_n( F\eta).
			\end{split}
		\end{equation}
		The claim follows from \eqref{epsilon}, \eqref{finitelimit}, and \eqref{contdistrib}.
	\end{proof}
Let us now go back to the main proof. For the sake of clarity and simplicity we will adopt formal integral expressions for vector-valued distribution in what follows. In doing this one should bear in mind the conventions introduced at the beginning of Section \ref{sectioncausalprop}. \\[-0.5em]

\noindent \textbf{Point (i)}
Let us first fix some notation. For simplicity, we denote by $s_x$ the distribution with kernel $s_x(y):=s_m^\wedge(x-y)$. Given that $A$ has compact support we infer that 
$$
s_x\,\slashed{A}\in\mathcal{E}'(\R^4,\mathrm{Mat}(4,\C)).
$$ 
Let us rewrite (minus-) \eqref{firstorder} as $\langle s_x\,\slashed{A},\gR_{\varepsilon} u\rangle$, with
\begin{equation}\label{defbracket}
\langle s_x\,\slashed{A},\varphi\rangle:= \int_{\R^4}s_m^\wedge(x-y)\,\slashed{A}(y)\,\varphi(y)\,d^4y\in \C^4\quad\mbox{for all }\varphi\in C^\infty(\R^4,\C^4).
\end{equation}
In a similar style, one can integrate the distribution $s_x\,\slashed{A}$ against a scalar function, which would then yield a matrix. We use the same notation, i.e.
\begin{equation}\label{matrixnotation}
\langle s_x\,\slashed{A},h\rangle\in \mathrm{Mat}(4,\C)\quad\mbox{if}\quad h\in C^\infty(\R^4).
\end{equation}
With our notations,
$$
\langle s_x\,\slashed{A},h\chi\rangle=\langle s_x\,\slashed{A},h\rangle\,\chi\quad\mbox{for all $h\in C^\infty(\R^4)$ and $\chi\in \C^4$}.
$$

\noindent Next, we recall that for any $u\in\H_m^-$ there exists a unique $\psi\in L^2(\R^3,\C^4)$ such that
	\begin{equation}\label{generalexpression}
		\gR_{\varepsilon}u(y)=\int_{\R^3}e^{-\varepsilon\omega(\V{k})}\,e^{i(\omega(\V{k})y^0+\V{k}\cdot\V{y})}\,\psi(\V{k})\,d^3\V{k}.
	\end{equation}
	Such a $\psi$ is the three-dimensional Fourier transform of the initial data of $u$ (at $z^0=0)$. In particular, $\|u\|_m=\|\psi\|_{L^2}$.  
	Let us focus for a moment on the dense subset $\mathcal{D}$ of solutions $u\in\H_m^-$ with $\psi\in S(\R^3,\C^4)$ (for a proof of $\overline{\mathcal{D}}=\H_m^-$ see for example \cite[Lemma 2.17]{oppio}). 
	In view of the exponential factor $e^{-\varepsilon\omega}$ we can apply Lemma \ref{lemmadistrib} to the integrand of \eqref{generalexpression} and get
	\begin{equation}\label{firstestimateconvergene}
		\begin{split}
			\langle s_x\,\slashed{A},\gR_{\varepsilon}u\rangle&=\int_{\R^3}e^{-\varepsilon \omega(\V{k})}\langle s_x\slashed{A},w(\,\cdot\,,\V{k}) \rangle\,\psi(\V{k})\,d^3\V{k}\quad\mbox{ for all $u\in\mathcal{D}$},
		\end{split}
	\end{equation}
where we defined
$
w(y,\V{k}):=e^{i(\omega(\V{k})y^0+\V{k}\cdot\V{y})}.
$

\noindent So far, we restricted attention to functions in $\mathcal{D}$ because they have a smooth counterpart $\psi$ and this made it possible to apply Lemma \ref{lemmadistrib}. Using a denseness argument we now extend the above identity to arbitrary wave functions. Before entering this, we state and prove the following lemma\footnote{By $|x|$ we denote the standard Euclidean norm  in the corresponding space $\C^N$. For simplicity of notation we do not indicate the dimension $N$ in the norm symbol, for it is clear from the context.}.
	\begin{Lemma}\label{uniformbounds}
		For every $\psi\in C^\infty(\R^4)$ there is $C>0$ such that (cf.  notation \eqref{matrixnotation})
		\begin{equation}\label{quantitytobound}
			|\langle s_x\,\slashed{A},\psi\rangle|\le C\quad\mbox{for all $x\in\R^{4}$}.
		\end{equation}
	As a consequence, there are $k>0$, $N\in\N$ and a compact $K\subset\R^4$ such that, 
	\begin{equation}\label{estimdistrib}
		|\langle s_x\,\slashed{A},\varphi\rangle|\le k\sum_{|\alpha|\le N}\|D^\alpha \varphi\|_{\infty,K}\quad\mbox{for all }\varphi\in C^\infty(\R^4)\ \mbox{and all}\  x\in\R^4.
	\end{equation}
	\end{Lemma}
	\begin{proof}
		In order to study the quantity $\langle s_x\,\slashed{A},\psi\rangle$ we make use of (2.2.6) in \cite{cfs}, where an explicit form of the Green's operators is given, namely
		\begin{equation}\label{greenoperator}
			\begin{split}
			s^\wedge=(i\slashed{\partial}+m)S^\wedge_m,\quad
			S^\wedge(\xi)=\alpha\,\delta(\xi^2)\,\Theta(\xi^0)+\beta\,\frac{J_1(m\sqrt{\xi^2})}{m\sqrt{\xi^2}}\,\Theta(\xi^2)\,\Theta(\xi^0),
			\end{split}
		\end{equation}
	where $\alpha,\beta\in\R$ and $J_1$ is the \textit{Bessel function of the first kind of order one} (see \cite[Section 9.1]{AS}). Note that the second addend in the definition of $S^\wedge$ vanishes outside $J_0$ and is continuous and bounded therein, as can be inferred from the asymptotic behaviour of $J_1$. For simplicity of notation we denote this second term by $f(\xi)$.
	
	Absorbing $\slashed{A}$ into $\psi$ and using \eqref{defbracket}, \eqref{greenoperator} and the definition of distributional derivative, our task boils down to studying the $x$-dependence of the following two expressions for general $\phi\in\mathcal{D}(\R^4)$ and $|\alpha|\le 1$,
		\begin{equation*}
			\begin{split}
				1)&\ \int_{\R^4}\delta(\xi^2)\Theta(\xi^0)\,(D^\alpha \phi)(x-\xi)\,d^4\xi=  \int_{B_R(\V{x})} \frac{(D^\alpha \phi)(x^0-|\boldsymbol{\xi}|,\V{x}-\boldsymbol{\xi})}{2|\boldsymbol{\xi}|}\,d^3\boldsymbol{\xi}=:M(x),\\
				2)&\  \int_{\R^4}f(\xi)\,(D^\alpha \phi)(x-\xi)\,d^4\xi=:N(x)
			\end{split}
		\end{equation*}
		where $R>0$ was chosen so large that $\supp\phi\subset \R\times B_R(0)$. We need to show that both $M$ and $N$ are bounded functions of $x$.\\[-0.3em]
		
		\noindent 1)
		Let $x\in\R^4$ fulfill $|\V{x}|_{\R^3}\ge R+1$. For any $\boldsymbol{\xi}\in B_R(\V{x})$ we then get  
		$
		|\boldsymbol{\xi}|_{\R^3}\ge |\V{x}|_{\R^3}-R>1$. We conclude that
		\begin{equation}
			|M(x)|\le 2\pi R^3\,\|D^\alpha\phi\|_\infty\quad\mbox{for all }x\in \R\times (\R^3\setminus B_{R+1}(0)).
		\end{equation}
	Let us now consider the complementary region $\R\times B_{R+1}(0)$.
	Note that the only  $\boldsymbol{\xi}$ contributing to $M(x)$ are those sastisfying
		$$
		(-|\boldsymbol{\xi}|,-\boldsymbol{\xi})\in (-x+\supp\phi) \cap L_0.
		$$
		However, being $|\V{x}|_{\R^3}\le R+1$ we see that the set on the right-hand side is empty when $|x^0|$ is sufficiently large and therefore $M(x)=0$.  In summary, there exists $K>0$ such that $M(x)=0$ for all $x\in \R\times B_{R+1}(0)$ with $|x^0|\ge K$. Using that $M$ is continuous on $\R^4$, being the convolution of a distribution and a compactly supported smooth function, we conclude that there is $C>0$ such that 
		$$
		|M(x)|\le C\quad\mbox{for all }x\in \R\times B_{R+1}(0).
		$$
		Taking the greater between the two upper bounds concludes the proof of part 1).\\[-0.3em]
		 
		\noindent 2) 
		We know that $f$ is a bounded measurable function. Therefore, for some $C>0$,
		$$
		|N(x)|\le C\int_{\R^4}|D^\alpha\phi(x-\xi)|\,d^4\xi = C\int_{\R^4}|D^\alpha \phi(w)|\,d^4 w<\infty\quad\mbox{for all $x\in\R^4$}.
		$$
	 \noindent To conclude, recall that $s_x\,\slashed{A}\in\mathcal{E}'(\R^4,\mathrm{Mat}(4,\C))$, and hence, from the properties of compactly supported distributions,  for every $x\in\R^4$ there exist $k,N,K$ as in \eqref{estimdistrib}. The fact that these quantities can be chosen uniformly in $x$ is a consequence of \eqref{quantitytobound} and the uniform boundedness principle.
		The proof of Lemma \ref{uniformbounds} is concluded.
	\end{proof}
	\noindent  Let us go back to the main proof. As  a consequence of \eqref{estimdistrib}, we infer that 
	\begin{equation}\label{estimatew}
		|\langle s_x\slashed{A},w(\,\cdot\,,\V{k})\rangle|\le f(\V{k})\quad\mbox{for all $x\in\R^4,\  \V{k}\in\R^3$},
	\end{equation}
	where $f$ is a continuous polynomial function in the variables $\omega(\V{k})$ and $|k^\alpha|$.
	
	\noindent Let now $u\in\H_m^-$ be arbitrary (with corresponding $\psi\in L^2(\R^3,\C^4)$) and let  $u_n\in\mathcal{D}$ (with corresponding $\psi_n\in S(\R^3,\C^4)$) be such that $\|\psi_n-\psi\|_{L^2}\to 0$.  From this, \eqref{generalexpression} and H\"older inequality it is not difficult to see that
	\begin{equation}\label{A}
		\|D^\alpha(\gR_{\varepsilon} u-\gR_{\varepsilon} u_n)\|_\infty \to 0\quad\mbox{for every multi-index $\alpha\in\N^4$.}
	\end{equation}
	Similarly, using \eqref{estimatew} and H\"older inequality, it can be shown that
	\begin{equation}\label{B}
		\int_{\R^3}e^{-\varepsilon \omega(\V{k})}\langle s_x\slashed{A},w(\,\cdot\,,\V{k}) \rangle\,\psi_n(\V{k})\,d^3\V{k}\to \int_{\R^3}e^{-\varepsilon \omega(\V{k})}\langle s_x\slashed{A},w(\,\cdot\,,\V{k}) \rangle\,\psi(\V{k})\,d^3\V{k}.
	\end{equation}
	Putting \eqref{firstestimateconvergene}, \eqref{A} and \eqref{B} together, we conclude that, for every $u\in\H_m^-$,
	$$
	\int_{\R^3}e^{-\varepsilon \omega(\V{k})}\langle s_x\slashed{A},w(\,\cdot\,,\V{k}) \rangle\,\psi(\V{k})\,d^3\V{k}=\lim_{n\to\infty} \langle s_x\,\slashed{A},\gR_{\varepsilon}u_n\rangle = \langle s_x\,\slashed{A},\gR_{\varepsilon}u\rangle.
	$$
	Using H\"older inequality and \eqref{estimatew} again, we then conclude that, for some $K>0$,
	\begin{equation*}
		\begin{split}
			|\langle s_x\slashed{A},\gR_{\varepsilon} u\rangle|&\le \|\psi\|_{L^2}\,\left(\int_{\R^3}|f(\V{k})|^2\,\,e^{-2\varepsilon \omega(\V{k})}\,\,d^3\V{k}\right)^{1/2}\\
			&\le \|u\|_m\,K\quad\mbox{for all $u\in\H_m^-$ and all $x\in\R^4$} .
		\end{split}
	\end{equation*}
	From this we infer the first important estimate on \eqref{firstorder}:
	$$
	\sup_{x\in \R^4}\|\Psi^{(1)}(x)\|_{\mathfrak{B}(\H_m^-,\C^4)}\le K<\infty.
	$$ 
	To prove the second important estimate, let $y\in\R^4$ be arbitrary. Following the same argument above, one can infer that, for every $u\in\H_m^-$,
	\begin{equation}\label{finalestimate}
		\begin{split}
		|\langle (s_{x}-s_y)\,\slashed{A},\gR_{\varepsilon}u\rangle|&\le \|\psi\|_{L^2}\left( \int_{\R^3}e^{-2\varepsilon\omega(\V{k})}|\langle (s_{x}-s_y)\,\slashed{A},w(\,\cdot\,,\V{k})\rangle|^2\,d^3\V{k}\right)^{1/2}\\
		&\le \|u\|_m\,H(x,y)\quad\mbox{for all $u\in\H_m^-$ and $x,y\in \R^4$},
		\end{split}
	\end{equation}
	which yields the second important estimate on \eqref{firstorder}:
	$$
	\|\Psi^{(1)}(x)-\Psi^{(1)}(y)\|_{\mathfrak{B}(\H_m^-,\C^4)}\le H(x,y)\quad\mbox{for all $x,y\in\R^4$}.
	$$
	Estimate \eqref{estimatew} ensures that the integrand on the right-hand side of \eqref{finalestimate} is bounded from above  uniformly in $x,y$ by an integrable function. Moreover, for fixed $x$ such an integrand is continuous in $y$. Applying Lebesgue's dominated convergence theorem we then conclude that
	$$
	\lim_{y\to x}H(x,y)=0\quad\mbox{for every $x\in\R^4$}.
	$$
	Putting all together, we infer that $\Psi^{(1)}\in \mathcal{E}(\R^4,\H_m^-,\C^4)$ and point (i) is proved.\\[-0.5em]
	
	\noindent \textbf{Point (ii)} The fact that $\mathrm{F}^{(1)}(x)$ is bounded follows immediately from the fact that 
	$$
	\gR_\varepsilon(x),\Psi^{(1)}(x)\in\mathfrak{B}(\H_m^-,\C^4)\quad\mbox{and hence}\quad \gR_\varepsilon(x)^*,\Psi^{(1)}(x)^*\in\mathfrak{B}(\C^4,\H_m^-).
	$$
	Self-adjointness is obvious from the definition.\\[-0.6em]
	
	\noindent \textbf{Point (iii)} Let us prove continuity first. It suffices to focus on the first addend, the other one being analogous. 
	Note that for every $A\in \mathfrak{B}(\H_m^-,\C^4)$ it holds that $A^*=A^\dagger \gamma^0$ (see footnote \ref{footnoteadjoint}). Therefore, $\gamma^0$ being unitary on $\C^4$, we get $$\|A^*\|_{\mathfrak{B}(\H_m^-,\C^4)}=\|A^\dagger\|_{\mathfrak{B}(\H_m^-,\C^4)}=\|A\|_{\mathfrak{B}(\H_m^-,\C^4)}.$$
	
	\noindent Let now $x,y\in\R^4$. Then, omitting for simplicity the indices indicating the $\sup$-norm,
	\begin{equation*}
		\begin{split}
			\|\gR_{\varepsilon}(x)^*&\Psi^{(1)}(x)-\gR_{\varepsilon}(y)^*\Psi^{(1)}(y)\|\\ &\le\|\gR_{\varepsilon}(x)^*(\Psi^{(1)}(x)-\Psi^{(1)}(y))\|+\|(\gR_{\varepsilon} (x)-\gR_{\varepsilon} (y))^*\Psi^{(1)}(y)\|\\
			&\le\|\gR_{\varepsilon}(x)^*\|\|\Psi^{(1)}(x)-\Psi^{(1)}(y)\|+\|(\gR_{\varepsilon} (x)-\gR_{\varepsilon} (y))^*\|\|\Psi^{(1)}(y)\|\\
			 &\le\|\gR_{\varepsilon}(x)\|\|\Psi^{(1)}(x)-\Psi^{(1)}(y)\|+\|\gR_{\varepsilon} (x)-\gR_{\varepsilon} (y)\|\|\Psi^{(1)}(y)\|
		\end{split}
	\end{equation*}
 The claim will then follow from point (i). In a similar way one can prove boundedness.\\[-0.5em]

 \noindent \textbf{Point (iv)} To prove this point we first need to compute $\mathrm{F}^{(1)}(x)u$ explicitly. From Proposition \ref{relationRF}-(ii) we know that
 \begin{equation*}
 	\begin{split}
  \gR_{\varepsilon}(x)^*\Psi^{(1)}(x)u&=-2\pi\,P^\varepsilon(\,\cdot\,,x)(\Psi^{(1)}(x)u)\\
 	&=P^\varepsilon(\,\cdot\,,x)\left(2\pi\int_{\R^4}s_m^\wedge(x-y)\,\slashed{A}(y)\,\gR_{\varepsilon} u(y)\,d^4y\right)
 	\end{split}
 \end{equation*}
From \eqref{L1initialdata} it follows that, for any $t\in\R$,
$$
(\gR_{\varepsilon}(x)^*\Psi^{(1)}(x)u)(t,\,\cdot\,)\in S(\R^3,\C^4).
$$ 
Let us now study the second addend $\Psi^{(1)}(x)^*\,\gR_{\varepsilon}(x)u$. To this aim, we need to  determine the explicit action of $\Psi^{(1)}(x)^*$.
We give a formal argument, a more rigorous statement can be carried out by testing the following identities on arbitrary vectors and spinors. Let us write \eqref{firstorder} as the action of the following formal operator,
$$
\Psi^{(1)}(x)=-\int_{\R^4}s_m^\wedge(x-y)\,\slashed{A}(y)\,\gR_{\varepsilon}(y)\,d^4y.
$$
From identities (2.1.9), (2.1.10) in \cite[Section 2.1.3]{cfs} we see that
$
s_m^\wedge(x-y)^*=s_m^\vee(y-x).
$
Using Proposition \ref{relationRF}-(ii), we then obtain, 
\begin{equation}\label{spinor}
	\begin{split}
		\Psi^{(1)}(x)^*\chi&=-\int_{\R^4}\gR_{\varepsilon}(y)^*(\slashed{A}(y)\,s_m^\vee(y-x)\chi)\,d^4y\\
		&=2\pi\int_{\R^4}P^\varepsilon(\,\cdot\,,y)\,\slashed{A}(y)\,s^\vee_m(y-x)\chi\,d^4y
	\end{split}
\end{equation}
We now want to show that the above function is smooth and that its restriction to any Cauchy surface $\{t=\mathrm{const}\}$ belongs to $S(\R^3,\C^4)$. For simplicity of notation, let $s_x$ denote the distribution with kernel $s_x(y):=s_m^\vee(y-x)$. Having $A$ compact support, it follows that
$$
\slashed{A}\,s_x\in\mathcal{E}'(\R^4,\mathrm{Mat}(4,\C)).
$$ 
The analysis of $\Psi^{(1)}(x)^*$ boils down to studying the properties of the following matrix-valued function:
\begin{equation}\label{convolP}
\R^4\ni z\mapsto\langle P^\varepsilon(z,\,\cdot\,),\slashed{A}\,s_x\rangle\in \mathrm{Mat}(4,\C),
\end{equation} 
where we defined
\begin{equation*}
	\langle M, \slashed{A}\,s_x\rangle  :=\int_{\R^4}M(y)\,\slashed{A}(y)\,s^\vee_m(y-x)\,d^4y\quad\mbox{for all $M\in C^\infty(\R^4,\mathrm{Mat}(4,\C))$}.
\end{equation*}

\noindent From Lemma \ref{lemmadistrib} we infer that the function \eqref{convolP} is smooth. Moreover, it solves the Dirac equation, as follows from \eqref{interchangederivativeaction} and the fact that $P^\varepsilon$ is itself a solution. 

\noindent We now prove that the restriction of this function to $\{z^0=\mathrm{const}\}$ is a Schwartz function. To this aim let us  consider identity  \eqref{3D}. In view of the exponential factor $e^{-\varepsilon\omega}$ we can apply Lemma \ref{lemmadistrib} to the integrand of \eqref{3D} and get
\begin{equation}\label{fourierexchange}
	\langle P^\varepsilon(z,\,\cdot\,),\slashed{A}\,s_x \rangle=-(2\pi)^{-4}\int_{\R^3} e^{-\varepsilon\omega(\V{k})} \,\langle  w(\,\cdot\,,\V{k}),\slashed{A}\,s_x \rangle\,e^{i(\omega(\V{k})z^0+\V{k}\cdot\V{z})}\,d^3\V{k},
\end{equation}
with $w(y,\V{k}):=e^{-i(\omega(\V{k})y^0+\V{k}\cdot\V{y})}\,p(\V{k})\gamma^0$. 

\noindent Now, using \eqref{interchangederivativeaction} and reasoning as in proof of \eqref{estimatew} one infers that, for any $\beta\in\N^3$,
$$
|D^\beta\langle  w(\,\cdot\,,\V{k}),\slashed{A}\,s_x \rangle|=|\langle D^\beta_2 w(\,\cdot\,,\V{k}),\slashed{A}\,s_x \rangle|\le f^\beta(\V{k})\quad\mbox{for any $\V{k}\in\R^3$,}
$$
where $f^\beta$ is a continuous rational function in the variables $\omega(\V{k})$ and $|k^\alpha|$. Therefore, thanks to the exponential factor $e^{-\varepsilon\omega}$,  the mapping
$$
\V{k}\mapsto e^{-\varepsilon\omega(\V{k})} \,\langle  w(\,\cdot\,,\V{k}),\slashed{A}\,s_x \rangle\,e^{i\omega(\V{k})z^0}
$$ 
is a Schwartz function. We conclude from \eqref{fourierexchange} that, for any choice of $z^0\in\R$, 
$$
\V{z}\mapsto \langle P^\varepsilon(z_0,\V{z},\,\cdot\,),\slashed{A}\,s_x \rangle \quad\mbox{belongs to}\  S(\R^3,\mathrm{Mat}(4,\C)).
$$ 
Restoring the spinor $\chi$ from \eqref{spinor} concludes the proof of point (iv).\\[0.3em]
 \noindent \textbf{Point (v)} This is clear from the support properties of the retarded Green's operator.

\end{proof}

\section*{Acknowledgments}
I am grateful to  Felix Finster for the helpful discussions and the several advices and hints, which helped me a lot in writing this paper. I would also like to thank my colleagues, in particular Christoph Langer and Maximilian Jokel, for the useful exchange of ideas.

\end{document}